\documentclass[journal]{IEEEtran}

\usepackage[utf8]{inputenc}
\usepackage[english]{babel}
\usepackage{epsfig}
\usepackage{amsfonts}
\usepackage{ifpdf}
\usepackage{multirow}
\usepackage{amssymb}
\usepackage{amsthm}
\usepackage{mathtools}
\usepackage{mathrsfs}
\usepackage{textgreek}
\usepackage{grffile}
\usepackage{xcolor}
\usepackage{bm}
\usepackage{cuted}

\usepackage{setspace}

\usepackage[usestackEOL]{stackengine}

  \stackMath

\usepackage{cite}

\ifCLASSINFOpdf
\else
\fi
\usepackage{amsmath}
\usepackage{algorithmic}
\usepackage{algorithm}
\usepackage{graphicx}

\newtheorem{definition}{Definition}
\newtheorem{assumption}{Assumption}
\newtheorem{remark}{Remark}
\newtheorem{lemma}{Lemma}
\newtheorem{theorem}{Theorem}

\newtheorem{proposition}{Proposition}

\newcommand{\norm}[1]{\left\lVert#1\right\rVert}

\ifCLASSOPTIONcompsoc
  \usepackage[caption=false,font=normalsize,labelfont=sf,textfont=sf]{subfig}
\else
  \usepackage[caption=false,font=footnotesize]{subfig}
\fi

\usepackage{stfloats}
\begin{document}
%


\title{RIS-Aided Localization under Position and Orientation Offsets in the Near and Far Field}
\author{Don-Roberts~Emenonye,
        Harpreet~S.~Dhillon,
        and~R.~Michael~Buehrer
\thanks{D.-R. Emenonye, H. S. Dhillon and R. M.  Buehrer are with Wireless@VT,  Bradley Department of Electrical and Computer Engineering, Virginia Tech,  Blacksburg,
VA, 24061, USA. Email: \{donroberts, hdhillon, rbuehrer\}@vt.edu. The support of the US National Science Foundation (Grants ECCS-2030215 and CNS-2107276) is gratefully acknowledged. 
}
}

\maketitle
\IEEEpeerreviewmaketitle

\begin{abstract}
This paper presents a rigorous Bayesian analysis of the information in the signal (consisting of both the line-of-sight (LOS) path and reflections from multiple reconfigurable intelligent surfaces (RISs)) that originate from a single base station (BS) and is received by a user equipment (UE). For a comprehensive Bayesian analysis, both near and far field regimes are considered.  The Bayesian analysis views both the location of the RISs and previous information about the UE as {\em a priori} information for UE localization. With outdated {\em a priori} information, the position and orientation offsets of the RISs become parameters that need to be estimated and fed back to the BS for correction. We first show that when the RIS elements have a half wavelength spacing, this RIS orientation offset is a factor in the pathloss of the RIS paths. Subsequently, we show through the Bayesian equivalent Fisher information matrix (EFIM) for the channel parameters that the RIS orientation offset cannot be corrected when there is an unknown phase offset in the received signal in the far-field regime. However, the corresponding EFIM for the channel parameters in the received signal observed in the near-field shows that this unknown phase offset does not hinder the estimation of the RIS orientation offset when the UE has more than one receive antenna. Furthermore, we use the EFIM for the UE location parameters to present bounds for UE localization in the presence of RIS uncertainty. We rigorously show that regardless of size and propagation regime, the RISs are only helpful for localization when there is {\em a priori} information about the location of the RISs.  Finally, through numerical analysis of the EFIM and its smallest eigenvalue, we demonstrate the loss in information when the far-field model is {\em incorrectly} applied to the signals received at a UE experiencing near-field propagation.
 \end{abstract}
\section{Introduction}
Location awareness is crucial for the success of numerous services, including internet-of-things and autonomous vehicles \cite{zekavat2011handbook,6894213,nguyen2017wireless,8454389,8306879,7707439}. The most common way to obtain location information is through the Global Positioning Systems (GPS); with its constellation of satellites across the globe, GPS provides sub-meter location accuracy in favorable propagation environments. However, GPS may be unavailable when there is no clear line-of-sight (LOS) to a sufficient number of satellites, such as in urban alleys, underground tunnels, and inside buildings. In such situations, the fallback option is to provide location information through existing wireless networks, such as the ubiquitous cellular network. However, localization using wireless networks depends strongly on the propagation environment. Since reconfigurable intelligent surface (RIS) allows us to control the propagation environment to some extent, its utility for localization has naturally been explored recently. Specifically, an RIS is envisioned as a planar surface comprising of sub-wavelength-sized meta-materials capable of controlling a wireless propagation channel by applying a desired transformation on the incoming signal through the software control of each meta-material. This unique ability to control the harsh wireless channel has led to many works investigating the suitability of RISs for localization \cite{9781656,9729782,8264743,9508872,9625826,9500663,9782100,9528041,9774917}. Although the usefulness of RISs for localization has been well established in the recent literature, the fundamental limits of RIS-aided localization in the near and far field propagation regimes under the effect of RIS misorientation and misalignment  have not been investigated, which is the main topic of this paper\footnote{ In this paper, misorientation is said to occur  when an antenna array is perturbed after deployment such that its orientation changes. Hence, there is an orientation offset. Likewise, misalignment is said to occur when an antenna array is perturbed such that the position of its centroid is changed after deployment.\color{black}}\color{black}.
\subsection{Background and Problem Statement}
A well-researched solution to location awareness involves the measurement/extraction of metrics such as time-of-arrival (TOA), time-difference-of-arrival (TDOA), and angle-of-arrival (AOA)  from the received signals \cite{5109631,4802193,6725647,7383332,6731596,6955781,4600186,8264804,8827486,5762798}.  A common analysis is to obtain the fundamental limits of a location awareness system while starting from the received signals.  In \cite{5571900,5571889,9606768,7364259}, the fundamental limits of location awareness are derived starting from the actual received signals while also including possible {\em a priori} information about various propagation characteristics. The {\em a priori} information is included by providing the Bayesian/General Fisher information bounds rather than the conventional Fisher information bounds. These two bounds differ in how we view the location parameters. In the derivation of the general Fisher Information Matrix (FIM), the location parameters are random variables, while in the derivations of the FIM, the location parameters are deterministic. In all cases, the FIM has a strong relationship to the signal-to-noise ratio (SNR).

The spacing between individual antennas/elements on an antenna array is on the order of the operating wavelength. Hence, exploring higher frequency bands naturally leads to considering numerous antennas/elements at the transmitters and receivers. Starting from the received signals, the investigation of fundamental limits of the location accuracy considering this setup of higher frequency bands with numerous antennas is an important research area \cite{garcia2017direct,8240645,8356190,8515231,8755880,fascista2021downlink,li2019massive,guerra2018single,9082200}. A key aspect of this literature is the possibility of localization using a single anchor, and this research direction is termed {\em single anchor localization}. Although non-line-of-sight (NLOS) propagation provides valuable information for location awareness in single anchor localization\cite{8515231}, the overall performance is heavily dependent on the presence of a LOS path from the transmitter to the receiver \cite{8240645,8356190,8515231}. However, an LOS path is not guaranteed at high operating frequencies, such as millimeter wave (mmWave) or terahertz, because of their high susceptibility to blockages  \cite{8240645}.

At such high frequencies, RISs have the potential to provide  a virtual LOS path from the transmitter to the receiver \cite{9781656,9729782,8264743,9508872,9625826,9500663,9782100,9528041,9774917}. This virtual LOS path can generally improve localization accuracy by supplementing the LOS path or serving as the primary LOS path. \emph{However, these RISs can be placed on movable objects\cite{9528041,boulogeorgos2022outage} or indoors \cite{9201330}, and information about their location can become outdated through misorientation or misalignment, hence, this work statistically characterizes the information available in the signals reflected by misoriented/misaligned RISs to a receiver operating in both near and far field propagation regimes. Moreover, the RIS location information (before misorientation/misalignment) may be imprecise as the ability to have exact location/orientation information is questionable. Hence, this work investigates the available information under these location uncertainties. 
}


\color{black}

\subsection{Related Works}
The following three existing research directions are of interest to this paper: i) RIS-aided localization, ii) single anchor localization and iii) Bayesian/general limits of localization. The pertinent works from these directions are reviewed next.
\subsubsection{RIS-aided localization} Prior works on RIS-aided localization can be grouped mainly into i) continuous RIS\cite{8264743,9729782,9781656} and discrete RIS \cite{9508872,9625826,9500663,9782100,9528041}, and ii) near-field\cite{8264743,9729782,9781656,9508872,9500663,9625826} and far-field propagation \cite{9782100,9528041,9774917}. In \cite{9781656}, a vision for RIS-aided localization, including the effects of using RISs with large apertures, controlling the phases of RIS elements in real-time, operating at high-frequency bands, and using highly directive and polarized antennas, is presented. In \cite{9729782}, physical structures such as walls are assumed to be coated with the requisite meta-materials, and the FIM for positioning is investigated. While authors in \cite{9729782} start their investigation of positioning bounds from the received signals, authors in \cite{8264743} begin their analysis of the location bounds from Maxwell's equations. In \cite{8264743}, the authors also show the potential of optimizing the individual phases of the RIS elements to improve both communications and localization performance. In \cite{9508872}, starting from the received signals, the fundamental limits for localization of a receiver are derived. Subsequently, these limits are used to develop an algorithm for improving SNR. Authors in \cite{9625826} use the degrees of freedom offered by the RIS phases to transform each corresponding RIS element into virtual anchors to develop ranging-based limits. The possibility of using RISs as a lens receiver along with directional and random beamforming is investigated in \cite{9500663}. A non-trivial case of joint synchronization and localization with a single receive antenna is investigated in \cite{9782100}; both theoretical bounds and algorithms to attain the bounds are presented. In \cite{9528041}, non-stationary users are equipped with RISs, and using the reflected signals from these RISs, the FIMs for the user positions are derived. In \cite{9774917}, the case of having a single transmitter, an RIS with multiple elements, and a single antenna receiver are investigated under receiver mobility. In our work in \cite{emenonye2022fundamentals}, the General FIMs are investigated for the case of multiple RISs, and it's shown that the angle of incidence and angles of reflection can not be estimated without {a priori} information about the complex path gains.

Note that RISs have also been studied in the radar context. The FIM for RIS-assisted radar has been studied in \cite{9963716,esmaeilbeig2022joint,9827797}. Also, in \cite{9732186}, the foundational radar detection problem is studied with two RISs (one close to the target and one close to the radar). In \cite{9264225}, enabling the coexistence of a cellular network and a radar system with RISs is investigated. Finally, authors in \cite{9454375} use an RIS to provide an additional echo of the target. Finally, note that the  performance of users under random RIS deployment and blockage is studied in \cite{9500724}.
\color{black}

\subsubsection{Single anchor localization} 
The deployment of a large number of antennas needed for single anchor localization and the deployment of a large number of elements required for RIS-aided localization is enabled by the use of higher frequencies. However, this is not the only similarity between the two research areas. For example, multipath signals are separable both by using orthogonal phases in the RIS elements or by using a large number of receive antennas in the single-anchor setup. Moreover, the techniques used to derive the closed-form expression of the FIM and analyze the impact of nuisance parameters are transferable from single anchor localization \cite{8240645,8356190,8515231} to RIS-aided localization \cite{emenonye2022fundamentals}. In \cite{garcia2017direct}, an alternative to two-step metrics-based AOA localization is provided by jointly processing the LOS and NLOS signals from multiple transmitters. In \cite{8240645}, the location bounds for single anchor localization are derived, and an algorithm based on the sparsity of mmWave channels is developed to achieve these bounds. Authors in \cite{8356190} provide closed-form expressions for the FIMs, present conditions in which the multipath observed at the receiver are separable, and analyze the impact of nuisance parameters on the location parameters of interest. The work in \cite{8515231} leverages  \cite{8356190} to show that NLOS can provide valuable information for localization. The non-trivial single receiver localization problem is considered with multiple observations \cite{8755880} and with a single observation \cite{fascista2021downlink}. In \cite{li2019massive}, the fine resolution obtainable by estimating the delays of multipath components is used for simultaneous localization and mapping. While \cite{guerra2018single} considers the localization performance under both spatial multiplexing and beamforming, localization performance is investigated under hardware impairments in \cite{9082200}.
\subsubsection{Bayesian/general limits of localization}
The general FIM derived in this work to incorporate RIS location as {\em a priori} information for localization is similar to previous works incorporating {\em a priori} information about the channel parameters and network configurations \cite{8264804,8827486,5762798,5571900,5571889,9606768,7364259}. In \cite{8264804}, hard decisions on the positioning are replaced by probabilistic values; more specifically, probabilities are assigned to all probable locations of the transmitter/receiver. Subsequently, the work is extended to include probabilistic location values, and {\em a priori} information about the context of operation \cite{8827486}. Finally, the previous two results are extended to include cooperative localization, and navigation in \cite{5762798}. In \cite{5571900}, {\em a priori} information about the channel parameters and {\em a priori} information about the transmitter or receiver location are included in a comprehensive investigation of the fundamental limits of wideband localization. This work also shows that NLOS is only helpful with {\em a priori} information about the channel parameters, and it is extended to cooperative localization in \cite{5571889}. The cooperative localization framework is then developed for collaborative positioning in massive networks \cite{9606768}. 

\subsection{Contributions}
This paper develops a rigorous Bayesian framework for the localization of a user equipment (UE) operating in the downlink of a multiple-input multiple-output (MIMO) system with multiple RISs (some with uncertain locations) reflecting signals from the base station (BS) to the UE in addition to the LOS path. Under this scenario,  our main contributions are highlighted next:
\subsubsection{Misorientation Aware Pathloss}
We present a generalized pathloss model for discrete RISs with  half a wavelength spacing such that the RIS's area is equivalent to the sum of the effective apertures of all RIS elements, i.e., the square of half the wavelength. The developed pathloss model is derived by viewing the RIS elements as a conducting screen with an electrically-small low-gain element. Hence, the gain due to each RIS element is almost entirely specified by the direction vector pointing out from that element. This direction vector depends on the orientation of the BS, RISs, and UE. With this pathloss model, we explicitly derive the first derivatives needed for the derivations of the FIM. Furthermore, unlike other works, the product rule is required to present some entries in the FIM because the RIS orientation is present both in the pathloss and the array response vectors. Hence, the resulting bound on the covariance of the estimation error for any unbiased estimator depends on the orientation of the RISs.
\subsubsection{Impact of RIS Orientation Offset and an Unknown Phase Offset}
We incorporate uncertainties in RIS location by strictly considering the position and orientation of the RISs as {\em a priori} information for localization and subsequently present a derivation of the general FIM. This viewpoint presents the opportunity of investigating the estimation of RIS orientation offset at the receiver. We analyze the possibility of estimating this offset through the developed Bayesian framework under both near-field and far-field propagation regimes. First, we rigorously show that with no {\em a priori} information about the channel parameters, when the signal reflected from a misoriented RIS is received in the far-field with an unknown phase offset, the resulting general FIM containing both the orientation offset and the unknown phase offset is rank deficient. This result indicates that irrespective of the number of receive antennas at the UE, in the absence of {\em a priori} information about the unknown phase offset, the orientation offset of an RIS can not be estimated and corrected in the far-field. Second, we show that the rank deficiency observed in the FIM containing both the orientation offset and the unknown phase offset for the far-field is not observed in the corresponding FIM for the near-field with more than one receive antenna. This result is non-trivial as {\em a priori} information about an unknown phase offset is challenging to obtain. Hence, practically, the possibility of estimating and correcting the RIS orientation offset through the received signal at the UE only exists in the case of near-field propagation. 
\subsubsection{General Bounds for UE Location under RIS Uncertainty}
We use the notion of the equivalent Fisher information matrix (EFIM) to derive fundamental limits of the UE location under RIS uncertainty. The derived result indicates that regardless of RIS size and propagation regime (near-field or far-field), RISs are only helpful when there is {\em a priori} information about their location. Also, to understand the loss in information when a far-field model is {\em incorrectly} applied to the signals received at a UE experiencing near-field propagation, we use simulations to analyze the EFIM in both cases (i.e., using the {\em correct} near-field model and the {\em incorrect} far-field model for this near-field propagation regime).

It is important to note that the Fisher information analysis of the available information in the signal received at a UE after reflections from RISs has been investigated in \cite{8264743,9729782,9781656, 9508872,9625826,9500663,9782100,9528041}. However, the analyses in those works fail to consider the possibility of uncertainty in the locations of the RISs. We consider these uncertainties in our work: i) by deriving a pathloss model that incorporates the RIS orientation offset, ii) by determining under what propagation regimes the RIS orientation offset can be estimated, and iii) by rigorously showing the impact of RIS uncertainty on the available information about the UE's location. 
\color{black}

\textit{Notation:} the superscript $(\cdot)^{[0]}$ is related to the LOS path; the superscript $(\cdot)^{[m]}, m \in \{1,2,\cdots \}$, is related to the $m^{\text{th}}$ RIS path;
the transpose operator is $(\cdot)^{\mathrm{T}}$; the hermitian transpose operator is $(\cdot)^{\mathrm{H}}$; the complex conjugate operator is $(\cdot)^{{*}}$;
the operation $[\bm{V}]_{[g_1:v_1,g_2:v_2 ]}$ extracts the submatrix represented by both the rows that lie in the range, $g_1:v_1$,  and the columns that lie in the range $g_2:v_2$ of the matrix $\bm{V}$; the matrix trace operator is  $\operatorname{Tr}(\cdot)$; $\bm{1}$ is a vector of all ones;  $\bm{I}_{N_{{V}}}$ is an $ N_{{V}} \times N_{{V}}$ identity matrix;
the operator for the Euclidean norm  is $\norm{\cdot}$; $ \succ$ characterizes the positive definiteness of a matrix;
the $\nabla$ operator represents the first derivative; 
 the notation $\mathbb{E}_{\bm{v}}\{\cdot\}$ is the expectation operator with respect to the random vector $\bm{v}$; the following function represents the FIM: 
 $
 \bm{F}_{\bm{v}}(\bm{w} ; \bm{x}, \bm{y}) \triangleq \mathbb{E}_{\bm{v}}\left\{\left[\nabla_{ \bm{x}} \ln f(\bm{w})\right]\left[\nabla_{\bm{y}} \ln f(\bm{w})\right]^{\mathrm{T}}\right\}.
 $

 \section{Misorientation and Misalignment Aware System Model}
We consider the downlink of an RIS-assisted MIMO system operating with $N_{{B}}$ antennas at a single BS, $N_{{R}}^{[{m}]}$ reflecting elements at the ${m}^{\text {th}}$ RIS where ${m} \in  \mathcal{M}_1 = \{1, 2, \cdots, {M}_1 \}$,  and $N_{{U}}$ antennas at the UE. This paper occasionally refers to the BS, RISs, and UE as communication entities. The $V^{\text{th}}$ entity has an arbitrary but known array geometry and the centroid of this entity  is {\em initially} located at $\Tilde{\bm{p}}_{{V}} = [\Tilde{x}_{{{V}}}, \Tilde{y}_{{{V}}}, \Tilde{z}_{{{V}}}]^{\mathrm{T}} \in \mathbb{R}^3$, the ${v}^{\text{th}}$ element on this entity is initially located at $\Tilde{\bm{s}}_{{v}} \in \mathbb{R}^3$.  The point, $\Tilde{\bm{s}}_{{v}}$, is defined with respect to the centroid of the $V^{\text{th}}$ entity and is  uniquely represented with respect to the global origin as $\Tilde{\bm{p}}_{{v}} = \Tilde{\bm{p}}_{{V}} + \Tilde{\bm{s}}_{{v}}$. A representation in matrix form  of the points describing all the elements on the $V^{\text{th}}$ entity is $ \Tilde{\bm{S}}_v = [\Tilde{\bm{s}}_1, \Tilde{\bm{s}}_2, \cdots, \Tilde{\bm{s}}_{N_V} ],$
where $N_V$ is the number of elements on the $V^{\text{th}}$ entity. For the BS and RISs, the initial location refers to the location at deployment and can be assumed to be known. On the other hand, the knowledge of the initial location of the UE is application dependent. These initial locations might change with time, and the location server might be unaware of the change, thereby the location server will have outdated information about the locations of these entities, which results in what we call misorientation and misalignment. \color{black} The effect is particularly prominent when some of these entities, such as RISs, are deployed indoors or on movable objects \cite{9528041,boulogeorgos2022outage}. The exact reason behind misorientation and misaligned is not important for our analysis. The new (misoriented and misaligned) position of an element on a given entity is 
\begin{equation}
    \label{equ:entity_element}
    \begin{aligned}
{\bm{p}}_{{v}} &=  \Tilde{\bm{p}}_{{V}}  +\bm{\xi}_{{V}} +  \bm{Q}_{V}\Tilde{\bm{s}}_{{v}}, \\
{\bm{p}}_{{v}} &=  {\bm{p}}_{{V}}  +  {\bm{s}}_{{v}},
    \end{aligned}
\end{equation}
 where ${\bm{p}}_{{V}}= [x_{{{V}}}, y_{{{V}}}, z_{{{V}}}]^{\mathrm{T}} = \Tilde{\bm{p}}_{{V}}  +\bm{\xi}_{{V}}$, ${\bm{s}}_{{v}} = \bm{Q}_{V}\Tilde{\bm{s}}_{{v}} $, $\bm{Q}_{V} = \bm{Q}\left(\alpha_{V}, \psi_{V}, \varphi_{V}\right)$ defines a $3$D rotation matrix \cite{lavalle2006planning}, and $\bm{\xi}_{{V}} \in \mathbb{R}^3 $ is a vector specifying misalignment. The orientation angles of the $V^{\text{th}}$ entity are vectorized as $\bm{\Phi}_V = \left[\alpha_{V}, \psi_{V}, \varphi_{V}\right]^{\mathrm{T}}$. In matrix form, all the points describing the elements on the $V^{\text{th}}$ entity that has been misoriented are collectively represented as $ {\bm{S}}_v = [{\bm{s}}_1, {\bm{s}}_2, \cdots, {\bm{s}}_{N_V} ].$
\begin{figure}[htb!]
\centering
{\includegraphics[scale=0.3]{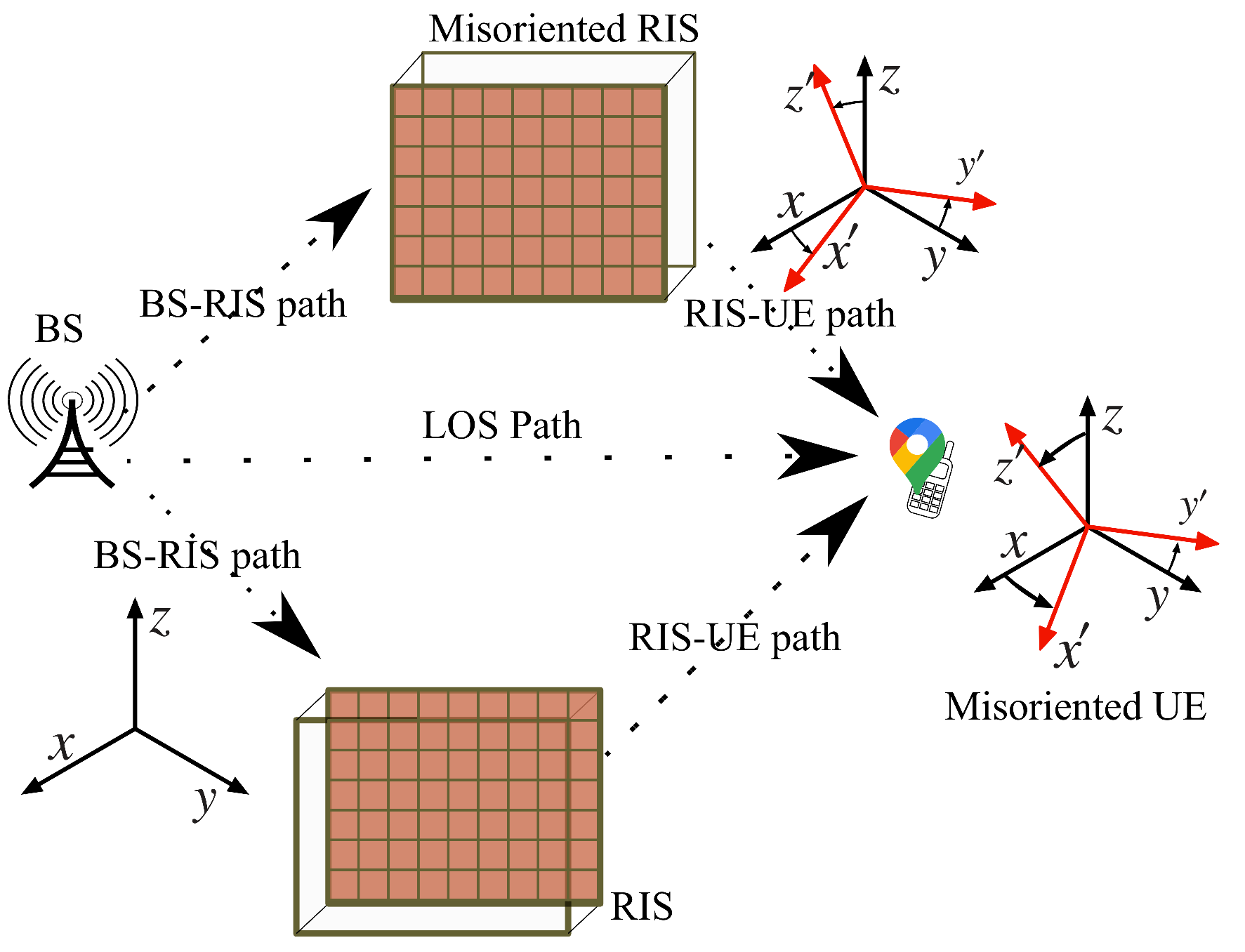}}
\caption{An illustration showing the BS which serves as the global origin with a misoriented RIS and a perfectly located RIS providing reflected paths to the UE in addition to an LOS path.} 
\label{fig:Results/system_model_final_1}
\end{figure}
The position of the $V^{\text{th}}$ entity's centroid located at ${\bm{w}}_{{V}}$ can be described in relation to the position of the $G^{\text{th}}$ entity's centroid located at ${\bm{w}}_{{G}}$ as ${\bm{w}}_{{V}} = {\bm{w}}_{{G}}     +d_{\bm{w}_{G} \bm{w}_{V}} \bm{\Delta}_{\bm{w}_{G} \bm{w}_{V}},$
where ${\bm{w}} \in \{{\bm{p}}, \Tilde{{\bm{p}}} \}$, $d_{\bm{w}_{G} \bm{w}_{V}}$ is the distance from point ${\bm{w}}_{{G}}$ to point ${\bm{w}}_{{V}}$ and $\bm{\Delta}_{\bm{w}_{G} \bm{w}_{V}}$ is the corresponding unit direction vector $\bm{\Delta}_{\bm{w}_{G} \bm{w}_{V}} = [\cos \phi_{{\bm{w}_{{G}} {\bm{w}}_{{V}}}} \sin \theta_{\bm{w}_{G} \bm{w}_{V}}, \sin \phi_{\bm{w}_{G} \bm{w}_{V}} \sin \theta_{\bm{w}_{G} \bm{w}_{V}}, \cos \theta_{\bm{w}_{G} \bm{w}_{V}}]^{\mathrm{T}}$. The points corresponding to the $v^{\text{th}}$ element on the $V^{\text{th}}$ entity is defined as $\bm{w}_{{v}}^{} = d_{\bm{w}_{{v}}^{}} \bm{\Delta}_{\bm{w}_{{v}}^{}},$
where ${\bm{w}} \in \{{\bm{s}}, \Tilde{{\bm{s}}} \}$,  $d_{\bm{w}_{{v}}^{}}$ is the distance from the centroid of the $V^{\text{th}}$ entity to its ${v^{\text{th}}}$ element, and $\bm{\Delta}_{\bm{w}_{{v}}^{}}$ is the corresponding unit direction vector, $\bm{\Delta}_{\bm{w}_{{v}}^{}} = [\cos \phi_{\bm{w}_{{v}}^{}} \sin \theta_{\bm{w}_{{v}}^{}}, \sin \phi_{\bm{w}_{{v}}^{}} \sin \theta_{\bm{w}_{{v}}^{}}, \cos \theta_{\bm{w}_{{v}}^{}}]^{\mathrm{T}}$.  The angles $\phi$ and $\theta$ in the respective unit vectors correspond to the appropriate azimuth and elevation angles, respectively.

In this paper, we consider the localization of the UE in the downlink of an RIS-assisted wireless system. Hence, the position and orientation of the UE are unknowns that need to be estimated. In addition to the traditional localization setup, we consider the possibility of an RIS location getting perturbed after placement. The RIS can be perturbed through misorientation and misalignment. In this paper, misorientation is said to occur  when the orientation of an entity is changed after deployment. Likewise, misalignment is said to occur when the centroid of an entity is moved after deployment. Mathematically, the misorientation of the $m^{\text{th}}$ RIS can be described as 
$
{\bm{s}}_{{r}}^{[m]} = \bm{Q}_{{R}}^{[m]}\Tilde{\bm{s}}_{{r}}^{[m]},
$
where $\Tilde{\bm{s}}_{{r}}^{[m]}$ is the location of the $r^{\text{th}}$ element on the $m^{\text{th}}$ RIS after initial deployment and ${\bm{s}}_{{r}}^{[m]}$ is the location of the $r^{\text{th}}$ element on the $m^{\text{th}}$ RIS after it has been misoriented. Also, the misalignment of the $m^{\text{th}}$ RIS can be mathematically described as 
 $
{\bm{p}}_{{R}}^{[m]} =  \Tilde{\bm{p}}_{{R}}^{[m]}  +\bm{\xi}_{{R}}^{[m]},
$
where $\Tilde{\bm{p}}_{{R}}^{[m]}$ is the location of the centroid of  the $m^{\text{th}}$ RIS after initial deployment and ${\bm{p}}_{{R}}^{[m]}$ is the location of the  of the centroid of  the $m^{\text{th}}$ RIS after it has been misaligned.  Finally, the misorientation and misalignment can be jointly described as
 $
 \begin{aligned}
{\bm{p}}_{{r}}^{[m]} =  \Tilde{\bm{p}}_{{R}}^{[m]}  +\bm{\xi}_{{R}}^{[m]} +  \bm{Q}_{{R}}^{[m]}\Tilde{\bm{s}}_{{r}}^{[m]}, 
{\bm{p}}_{{r}}^{[m]} =  {\bm{p}}_{{R}}^{[m]}  +  {\bm{s}}_{{r}}^{[m]}.
\end{aligned}
$
Hence, in this paper, we investigate the effect of having a perturbed RIS on the pathloss and localization performance of a UE. In addition, we also explore the possibility of estimating the RIS orientation offset in the near and far field propagation regimes.
\color{black}

\begin{lemma}
\label{lemma:distance_v_w}
The distance from the  $g^{\text{th}}$ element (located at ${\bm{p}}_{{g}^{}}$) on the $G^{\text{{th}}}$ entity (located at ${\bm{p}}_{{G}^{}}$)  to the  $v^{\text{th}}$ element (located at ${\bm{p}}_{{v}^{}}$)  on the $V^{\text{{th}}}$ entity (located at ${\bm{p}}_{{V}^{}}$), is given by 
\begin{equation}
\label{equ_lemma:distance_v_w}
\begin{aligned}
d_{{{\bm{p}}}_{{g}} {{\bm{p}}}_{{v}}}^2 &= \norm{\Tilde{\bm{p}}_{ V} - \Tilde{\bm{p}}_{ G}   }^2 + \norm{\bm{Q}^{}_{ V}\Tilde{\bm{s}}_{v} - \bm{Q}^{}_{ G}\Tilde{\bm{s}}_{g}   }^2 + \norm{\bm{\xi}_{ V} - \bm{\xi}_{ G}}^2 \\
&+2(\Tilde{\bm{p}}_{ V} -\Tilde{\bm{p}}_{ G} )^{T}
(\bm{\xi}_{ V} - \bm{\xi}_{ G}) \\ 
&+ 2(\Tilde{\bm{p}}_{ V} - \Tilde{\bm{p}}_{ G} )^{T}
(\bm{Q}^{}_{ V}\Tilde{\bm{s}}_{v} - \bm{Q}^{}_{ G}\Tilde{\bm{s}}_{g} ) \\ 
& +2(\bm{Q}^{}_{ V}\Tilde{\bm{s}}_{v} - \bm{Q}^{}_{ G}\Tilde{\bm{s}}_{g} )^{T}
(\bm{\xi}_{ V} - \bm{\xi}_{ G}).
\end{aligned}
\end{equation}
\end{lemma}
\begin{proof}
The proof follows from the definition of the Euclidean norm.
\end{proof} 
\begin{remark}
Lemma \ref{lemma:distance_v_w} indicates clearly that the orientation of the source and destination entities are contained in the distance expression and this affects the derivations of the FIM (See Appendix \ref{Appendix:first_derivative}).
\end{remark}

From Lemma \ref{lemma:distance_v_w},  the distance from elements on an entity to elements on another entity is impacted by the orientation of both entities. Therefore, orientation affects the TOA measured at one entity with respect to another. Hence, if an entity is misoriented,  any location information inferred at any entity using this misoriented entity will have an inherent error. This partly motivates the need to study the impact of RIS misorientation on localization performance.
\color{black}


In the rest of this paper, we assume that there are at least ${M} \geq 2$ paths between the $b^{\text{th}}$ BS antenna and the $u^{\text{th}}$ UE antenna. The path, $m = 0$,  corresponds to the LOS path between the $b^{\text{th}}$ BS antenna and the $u^{\text{th}}$ UE antenna, while the remaining paths are RIS paths provided to the  $u^{\text{th}}$ UE antenna through the RISs. All other paths from natural scatters or natural reflectors are assumed to be much weaker than the LOS and the RIS paths. Hence, these paths are ignored in this paper \footnote{ This is not restrictive, as the non-RIS paths can always be removed by dividing the transmission into two stages. In the first stage, the RISs are turned off, and the signals from the non-RIS paths are collected. Then, the RISs are turned on in the second stage, and the signals from the non-RIS paths and RIS paths are collected. Subsequently, the signals from the RIS paths are isolated by subtracting the signals collected during the first stage from the signals collected during the second stage. After this, the LOS path can be identified as the first arriving path.\color{black}}.\color{black}

\subsection{Transmit Processing}
 We consider the transmission of $T$ Orthogonal frequency division multiplexing (OFDM) symbols, each containing $N$ subcarriers. The BS precodes $N_D$ data streams at each subcarrier using a precoder $\bm{F}[n] \in \mathbb{C}^{N_B \times N_D}$, subsequently the BS transforms the signal to the time domain using $N_B$ $N-$point inverse fast Fourier transforms (IFFT), and adds a cyclic prefix of sufficient length $N_{cp}$ to the transformed  symbols. In the time domain, this cyclic prefix has length $T_{cp} = N_{cp}T_s$ where $T_s = 1/B$ represents the sampling period, and the final transmitted symbol is $\bm{x}^{'}[n] = \bm{F}[n]\bm{x}[n]$
where $\boldsymbol{x}[n]=\left[x_{1}[n], \ldots, x_{N_{D}}[n]\right]^{\mathrm{T}} \in \mathbb{C}^{N_{D}}$.  
The ${d}^{\text{th}}$ column of the beamforming matrix, $\bm{F}$ directs its beam to a specific point defined by  ${\bm{p}_d}$. Hence, $\bm{f}_{{d}} \triangleq\left[e^{-j 2 \pi f_{n}\tau_{{\bm{b}_{1}} {\bm{p}_d}} },e^{-j 2 \pi f_{n}\tau_{\bm{b}_{2} {\bm{p}_d}}}, \cdots, e^{-j 2 \pi f_{n}\tau_{\bm{b}_{N_B} {\bm{p}_d}} } \right]^{\mathrm{H}},$
where  the frequency of the $n^{\text{th}}$ subcarrier is $f_{n}=f_{c}+n \Delta f-B / 2$, the  carrier frequency is $f_{c}$, $\Delta f$ is the subcarrier spacing, $B=N \Delta f$ is the signal bandwidth, and $\tau_{{\bm{b}_{b}} {\bm{p}_d}} = d_{{\bm{b}_{b}} {\bm{p}_d}}/ c$ is the delay from the $b^{\text{th}}$ BS antenna (located at ${\bm{b}_{b}}$) to the predefined point ${\bm{p}_d}$. Here, $c$ is the speed of light. All predefined points $\bm{D} = [\bm{p}_1, \bm{p}_2, \cdots, \bm{p}_{N_{D}} ]$ can be obtained during the initial deployments of BS and RISs. Note to ensure a power constraint,  we set $\operatorname{Tr}\left(\bm{F}^{{H}} \bm{F}\right)=$
1, $\mathbb{E}\left\{\bm{x}[n] \bm{x}^{{H}}[n]\right\}=\bm{I}_{N_{{D}}}$. 
 \subsection{Received Signal}
After the removal of the cyclic prefix and the application of an $N$-point fast Fourier transform (FFT) to the $t^{\text{th}}$ received OFDM symbol, the signal at the $n^{\text{th}}$ subcarrier can be described as
\begin{equation}
\label{equ:receive_processing}
\begin{aligned}
\bm{y}_{t}[n] &=  \bm{H}_{t, {{B}{U}}}^{[0]}[n] \bm{F} \bm{x}[n] \\ &+ \sum_{{m} = 1}^{{M}_1} \bm{H}_{t, {{R}{U}}}^{[{m}]}[n] \bm{\Omega}_{t}^{[{m}]} \bm{H}_{t, {{B}{R}}}^{[{m}]}[n]\bm{F} \bm{x}[n] 
+ \bm{n}_{t}[n] ,\\
&=   \bm{\mu}^{}_{t}[n] +  \bm{n}_{t}[n],
\end{aligned}
\end{equation}
where  $\bm{\mu}^{}_{t}[n]$ and $\bm{n}_{t}[n] \sim \mathcal{C}\mathcal{N}(0,N_0)$ are the noise-free part (useful part) of the signal and the Fourier transformed thermal noise local to the UE's antenna array, respectively. The channel matrix for the LOS BS-UE path during the $t^{\text{th}}$ OFDM symbol is $\bm{H}_{t, {{B}{U}}}^{[0]}[n] \in \mathbb{C}^{N_U \times N_B}$. The channel matrices for the RIS-UE and BS-RIS links corresponding to the $m^{\text{th}}$ RIS are $\bm{H}_{t, {{R}{U}}}^{[{m}]}[n] \in \mathbb{C}^{N_U \times N_R^{[m]}}$ and $\bm{H}_{t, {{B}{R}}}^{[{m}]}[n] \in  \mathbb{C}^{ N_R^{[m]} \times N_B}$, respectively. The reflection coefficients of the $m^{\text{th}}$ RIS during the $t^{\text{th}}$ OFDM symbol can be decomposed into $\bm{\Omega}_{t}^{[{m}]} = \gamma_{t}^{[{m}]} \bm{\Gamma}^{[{m}]}
$ where $\gamma_{t}^{[{m}]}$ is a complex scalar value, $\bm{\Gamma}^{[{m}]} = \text{diag}(e^{j{\vartheta}_1^{[m]}}, e^{j{\vartheta}_2^{[m]}},\cdots,e^{j{\vartheta_{N_R^{[m]}}^{[m]}}} )$ is a diagonal matrix, and ${\vartheta_{r}^{[m]}}$ is the phase of the $r^{\text{th}}$ element of the $m^{\text{th}}$ RIS. 

The complex scalar, $ \gamma_{t}^{[{m}]}$, can be viewed as the fast-varying part of the RIS reflection coefficients because it changes from one OFDM symbol to another. On the other hand, the diagonal matrix, $\bm{\Gamma}^{[{m}]}$, can be viewed as the slow-varying part of the RIS reflection coefficients because it remains constant across all OFDM symbols. These reflection coefficients can be discrete or continuous. However, practically, $\bm{\Omega}_{t}^{[{m}]}$ can be obtained from a codebook such as a discrete Fourier transform (DFT) matrix. Hence, $\bm{\Omega}_{t}^{[{m}]}$ can be selected as a row from a DFT matrix. It is important to note that such a row in this DFT can be factored into a scalar which will correspond to $\gamma_{t}^{[{m}]}$ and a vector part which will correspond to $\Gamma_{}^{[{m}]}$.
\color{black}
\begin{assumption}
\label{assumption:assumption_constant}
Misorientation and misalignment changes at a rate far slower than the $T$ OFDM symbols. Further, the channel is assumed constant during the transmission and reception of $T$ OFDM symbols.
\end{assumption}

To see that the above assumption is not highly restrictive, note that at practical frequency bands, the OFDM symbols have  a duration at most on the order of microseconds, and the location information can practically remain constant on the order of microseconds.
\color{black}
Using Assumption \ref{assumption:assumption_constant}, the time $t$ is dropped from the channel matrices and the useful part of the received signal at the $u^{\text{th}}$ UE antenna is presented in (\ref{equ:receive_processing_2})
\cite{hampton2013introduction,9782100}.
    \begin{figure*}
\begin{align}
    \label{equ:receive_processing_2}
    \begin{split}
&{\mu}_{t, {{u}}}[n] = \\ &\sum_{{{b}} = 1}^{N_{{B}}} \left(\sum_{{d} = 1}^{{N}_{D}} e^{j 2 \pi f_{n}\tau_{ { {\bm{b}_{b}} \bm{p}_{{d}}}} } x_{{d}}[n]  \right)  \times   \left( \beta^{[{0}]} {\rho_{}^{[{0}]}} e^{-j2\pi f_{n}({\epsilon}_{}^{[{0}]} + \tau_{{\bm{b}_{b}} {\bm{u}_{u}}}^{[0]})}   +  \sum_{{m} = 1}^{{M}_1} \gamma_{t}^{[{m}]}\beta^{[{m}]}  e^{-j2\pi f_{n}{\epsilon}_{}^{[{m}]}}  \sum_{{{r}} = 1}^{N_{{R}}^{[{m}]}} {\rho_{\bm{r}_r}^{[{m}]}} e^{j  \vartheta_{{{r}}}^{[{m}]} } e^{-j 2 \pi f_{n}(\tau_{{\bm{b}_{b}} {\bm{r}_{r}}}^{[{m}]} + \tau_{{\bm{r}_{r}} {\bm{u}_{u}}}^{[{m}]}) }
   \right).
\end{split}
\end{align}
 \end{figure*}
In (\ref{equ:receive_processing_2}), $\{\bm{b},\bm{r}, \bm{u} \} \in  \{\bm{p}, \Tilde{\bm{p}} \}$ where $\textbf{p}$ is used for entities that are misoriented and misaligned and $\tilde{\textbf{p}}$ is used for the entities that are at their original locations.  
The parameters $\beta^{[{0}]}$, $\rho_{}^{[{0}]}$, and ${\epsilon}_{}^{[{0}]}$ are the complex channel coefficients, pathloss, and synchronization error of the LOS paths, respectively. The LOS pathloss is $\rho_{}^{[{0}]} = (\lambda/4\pi)(1/{d_{{\bm{b}}_{{B}^{}}{\bm{u}}_{{U}^{}}}^{[0]}})$. The delay  related to the LOS path  between the $b^{\text{th}}$ BS antenna located at ${\bm{b}_{b}}$ and the $u^{\text{th}}$ UE antenna located at  ${\bm{u}_{u}}$ is defined by $\tau_{{\bm{b}_{b}} {\bm{u}_{u}}}^{[0]} = d_{{\bm{b}_{b}} {\bm{u}_{u}}}^{[0]} / c$. For the $m^{\text{th}}$ RIS, the complex channel coefficients and the synchronization error are specified by $\beta^{[m]}$ and ${\epsilon}_{}^{[m]}$, respectively, while the pathloss through the $r^{\text{th}}$ element related to the $m^{\text{th}}$ RIS is specified by  ${\rho_{\bm{r}_r}^{[{m}]}}$ (see Lemma \ref{lemma:pathloss}). The delay between the $b^{\text{th}}$ BS antenna and the $r^{\text{th}}$ element on the $m^{\text{th}}$ RIS is specified as $\tau_{ \bm{b}_{b} {\bm{r}_{r}}}^{[m]} = d_{ \bm{b}_{b} {\bm{r}_{r}}}^{[m]} / c$. Likewise in the $m^{\text{th}}$ RIS-UE link, the delay from the $r^{\text{th}}$ element on the $m^{\text{th}}$ RIS to  the $u^{\text{th}}$ UE antenna is specified as $\tau_{\bm{r}_{r} {\bm{u}_{u}}}^{[m]}  = d_{\bm{r}_{r} \bm{u}_{u}}^{[m]} / c$.

To analyze the information available in the far-field and the near-field, we present the far-field approximation of the useful part of the signal received in the near-field propagation regime.
\color{black}

\begin{proposition}
\label{proposition:far_field_1}
The far-field model (with a different parameterization than the model in \cite{8240645}) of the useful part of the received signal can be obtained from (\ref{equ:receive_processing_2}) and expressed as
\begin{equation}
\label{proposition_equ:far_field_1}
\begin{aligned}
&\bm{\mu}_{t}[n] =  \beta^{[{0}]} {\rho_{}^{[{0}]}} {\epsilon}_{n}^{[{0}]} \bm{a}_{UB}(\Delta_{\bm{b}_{B}\bm{u}_U}) \bm{a}_{BU}^{\mathrm{H}}(\Delta_{\bm{b}_{B}\bm{u}_U})  e^{-j 2 \pi f_{n}
\tau_{\bm{b}_{B}\bm{u}_U}^{[0]}} \\ &+ 
  \sum_{{m} = 1}^{{M}_1} \gamma_{t}^{[{m}]}\beta^{[{m}]}{\rho_{\bm{r}_R}^{[{m}]}} {\epsilon}_{n}^{[{m}]}  \bm{a}_{UR}(\Delta_{\bm{r}_{R}\bm{u}_U}^{[m]}) \bm{a}_{RU}^{\mathrm{H}}(\Delta_{\bm{r}_{R}\bm{u}_U}^{[m]})  \bm{\Gamma}^{[m]} \\& \bm{a}_{RB}(\Delta_{\bm{b}_{B}\bm{r}_{R}}^{[m]}) \bm{a}_{BR}^{\mathrm{H}}(\Delta_{\bm{b}_{B}\bm{r}_R}^{[m]}) \bm{F}[n] \bm{x}[n]  e^{-j 2 \pi f_{n} (\tau_{{\bm{b}_{B}} \bm{r}_{R}}^{[m]}   + \tau_{\bm{r}_{R} {\bm{u}_{U}}}^{[m]})},
\end{aligned}
\end{equation}
where ${\epsilon}_{n}^{[{m}]} = e^{-j2\pi f_{n}{\epsilon}_{}^{[{m}]}}$. 
Proposition \ref{proposition:far_field_1} is used in Appendix \ref{Appendix:farfield_channel_parameter_estimatable} to analyze the relationship between certain RIS-related channel parameters under the far-field propagation regime. 
\end{proposition}
\begin{proof}
After applying the far-field distance approximation of Appendix \ref{appendix_distance_approximation:far_field_1}, the proof follows by  defining the array responses between the $G^{\text{th}}$ entity and the $V^{\text{th}}$ entity as $\bm{a}_{GV}(\Delta_{\bm{g}_{G}\bm{v}_V}) = e^{-j 2 \frac{\pi}{\lambda_n}  \Delta_{\bm{g}_{G}\bm{v}_V}^{\mathrm{T}} {\bm{S}}_{g}} $, $\bm{a}_{VG}(\Delta_{\bm{g}_{G}\bm{v}_V}) = e^{-j 2 \frac{\pi}{\lambda_n}  \Delta_{\bm{g}_{G}\bm{v}_V}^{\mathrm{T}} {\bm{S}}_{v}} $. 
\end{proof}

Proposition 1 and (4) can be combined appropriately to present the case where some RISs provide the effects of near-field propagation to the UE, while others provide the effects of far-field propagation to the UE.
\color{black}

\section{Available Information in the Received Signal}
This section serves as an intermediary to analyze UE localizability. Here, we present bounds on the information available  from the received signals. \subsection{ Error Bounds on Parameters}
The analysis in this work is based both on the received signal specified by (\ref{equ:receive_processing}) and some of the parameters present in this signal. The state of these parameters is an indicator of the performance of both communication and localization systems; hence, in this section, we more clearly highlight these parameters.  The vector of LOS parameters is
\begin{equation}
\begin{aligned}
\begin{array}{llll}
\bm{\eta}^{[0]} &\triangleq\left[{\theta}_{\bm{b}_{B}\bm{u}_{U}}^{{[0]}}, {\phi}_{\bm{b}_{B}\bm{u}_{U}}^{{[0]}},
{\tau}_{\bm{b}_{B}\bm{u}_{U}}^{{[0]}}, \bm{\Phi}^{\mathrm{T}}_{U},{\epsilon}^{{[0]}},
  {\beta}_{\mathrm{R}}^{{[0]}}, {\beta}_{\mathrm{I}}^{{[0]}}  \right]^{\mathrm{T}} \\ &= \left[ \bm{\eta}_{1}^{[0]\mathrm{T}},   \bm{\eta}_{2}^{[0]\mathrm{T}}\right]^{\mathrm{T}},
\end{array}
\end{aligned}
\end{equation}
such that $\bm{\eta}_1^{[0]} = \left[{\theta}_{\bm{b}_{B}\bm{u}_{U}}^{{[0]}}, {\phi}_{\bm{b}_{B}\bm{u}_{U}}^{{[0]}},
{\tau}_{\bm{b}_{B}\bm{u}_{U}}^{{[0]}},\bm{\Phi}^{\mathrm{T}}_{U}  \right]^{\mathrm{T}}$ represents the LOS related geometric parameters and $\bm{\eta}_2^{[0]} = [  {\epsilon}^{{[0]}},
  {\beta}_{\mathrm{R}}^{{[0]}}, {\beta}_{\mathrm{I}}^{{[0]}}]^{\mathrm{T}}$ represents the LOS related nuisance parameters. The RIS related channel parameters are presented in (\ref{equ:ris_parameters}).
  
    \begin{figure*}
\begin{align}
    \label{equ:ris_parameters}
        \begin{split}
&\bm{\theta}_{{\bm{r}_{R}\bm{u}_{U}}} \triangleq\left[{\theta}_{{\bm{r}_{R}\bm{u}_{U}}}^{[{1}]},  \cdots, {\theta}_{{\bm{r}_{R}\bm{u}_{U}}}^{[{M}_1]}\right]^{\mathrm{T}}, 
\bm{\phi}_{{\bm{r}_{R}\bm{u}_{U}}} \triangleq\left[{\phi}_{{\bm{r}_{R}\bm{u}_{U}}}^{[{1}]},  \cdots, {\phi}_{{\bm{r}_{R}\bm{u}_{U}}}^{[{M}_1]}\right]^{\mathrm{T}}, 
\bm{\theta}_{{\bm{b}_{B}\bm{r}_{R}}} \triangleq\left[{\theta}_{{\bm{b}_{B}\bm{r}_{R}}}^{[{1}]},  \cdots, {\theta}_{{\bm{b}_{B}\bm{r}_{R}}}^{[{M}_1]}\right]^{\mathrm{T}}, \\
&\bm{\phi}_{{\bm{b}_{B}\bm{r}_{R}}} \triangleq\left[{\phi}_{{\bm{b}_{B}\bm{r}_{R}}}^{[{1}]},  \cdots, {\phi}_{{\bm{b}_{B}\bm{r}_{R}}}^{[{M}_1]}\right]^{\mathrm{T}}, 
\bm{\tau}_{{\bm{r}_{R}\bm{u}_{U}}} \triangleq\left[{\tau}_{{\bm{r}_{R}\bm{u}_{U}}}^{[{1}]},  \cdots, {\tau}_{{\bm{r}_{R}\bm{u}_{U}}}^{[{M}_1]}\right]^{\mathrm{T}}, 
\bm{\tau}_{{\bm{b}_{B}\bm{r}_{R}}} \triangleq\left[{\tau}_{{\bm{b}_{B}\bm{r}_{R}}}^{[{1}]},  \cdots, {\tau}_{{\bm{b}_{B}\bm{r}_{R}}}^{[{M}_1]}\right]^{\mathrm{T}}, \\
&\bm{\beta}  \triangleq\left[\beta^{[1]},  \ldots, \beta^{[{M}_1]}\right]^{\mathrm{T}}, 
\bm{\epsilon}  \triangleq\left[\epsilon^{[1]},  \ldots, \epsilon^{[{M}_1]}\right]^{\mathrm{T}}, 
\bm{\Phi}_{R}  \triangleq\left[\bm{\Phi}_{R}^{[1]\mathrm{T}},  \ldots, \bm{\Phi}_{R}^{[{M}_1]\mathrm{T}}\right]^{\mathrm{T}}.
    \end{split}
\end{align}
    \end{figure*}
These unknown channel parameters related to the RIS presented in (\ref{equ:ris_parameters}) can be represented by the vector in (\ref{equ:channel_params}).
\begin{figure*}
\begin{align}
\label{equ:channel_params}
\begin{array}{llll}
\bm{\eta}_{} \triangleq\left[\bm{\eta}^{[0]\mathrm{T}},  \bm{\theta}_{{\bm{r}_{R}\bm{u}_{U}}}^{\mathrm{T}}, \bm{\phi}_{{\bm{r}_{R}\bm{u}_{U}}}^{\mathrm{T}},
\bm{\theta}_{{\bm{b}_{B}\bm{r}_{R}}}^{\mathrm{T}}, \bm{\phi}_{{\bm{b}_{B}\bm{r}_{R}}}^{\mathrm{T}},
\bm{\tau}_{{\bm{r}_{R}\bm{u}_{U}}}^{\mathrm{T}},\bm{\tau}_{{\bm{b}_{B}\bm{r}_{R}}}^{\mathrm{T}}, \bm{\Phi}^{\mathrm{T}}_{R}, \bm{\epsilon}^{\mathrm{T}},
  \bm{\beta}_{\mathrm{R}}^{\mathrm{T}}, \bm{\beta}_{\mathrm{I}}^{\mathrm{T}}  \right]^{\mathrm{T}}.
\end{array}
\end{align}
\end{figure*}
In (\ref{equ:channel_params}), $\bm{\beta}_{\mathrm{R}} \triangleq \Re\{\bm{\beta}\}$, and $\bm{\beta}_{\mathrm{I}} \triangleq \Im\{\bm{\beta}\}$ are the real and imaginary parts of $\bm{\beta}$, respectively.  Now, we define the geometric channel parameters $ \bm{\eta}_{1} \triangleq\left[\bm{\eta}^{[0]\mathrm{T}}_{1},\bm{\theta}_{{\bm{r}_{R}\bm{u}_{U}}}^{\mathrm{T}}, \bm{\phi}_{{\bm{r}_{R}\bm{u}_{U}}}^{\mathrm{T}},
\bm{\theta}_{{\bm{b}_{B}\bm{r}_{R}}}^{\mathrm{T}}, \bm{\phi}_{{\bm{b}_{B}\bm{r}_{R}}}^{\mathrm{T}},
\bm{\tau}_{{\bm{r}_{R}\bm{u}_{U}}}^{\mathrm{T}},\bm{\tau}_{{\bm{b}_{B}\bm{r}_{R}}}^{\mathrm{T}},\bm{\Phi}_{R}^{\mathrm{T}}\right]^{\mathrm{T}}$
and  the nuisance parameters as $ \bm{\eta}_{2} \triangleq\left[\bm{\eta}^{[0]\mathrm{T}}_{2},\bm{\epsilon}^{\mathrm{T}},
  \bm{\beta}_{\mathrm{R}}^{\mathrm{T}}, \bm{\beta}_{\mathrm{I}}^{\mathrm{T}}
\right]^{\mathrm{T}}
$, hence, $ \bm{\eta}_{} \triangleq\left[\bm{\eta}_{1}^{\mathrm{T}}, \bm{\eta}_{2}^{\mathrm{T}}
\right]^{\mathrm{T}}$. An alternative representation according to the parameters associated with the $m^{\text{th}}$ RIS path is (\ref{equ:channel_params_1}).
\begin{figure*}
\begin{align}
\label{equ:channel_params_1}
\begin{array}{llll}
\bm{\eta}^{[m]} \triangleq\left[{\theta}_{{\bm{r}_{R}\bm{u}_{U}}}^{[m]}, {\phi}_{{\bm{r}_{R}\bm{u}_{U}}}^{[m]},
{\theta}_{{\bm{b}_{B}\bm{r}_{R}}}^{[m]}, {\phi}_{{\bm{b}_{B}\bm{r}_{R}}}^{[m]},
{\tau}_{{\bm{r}_{R}\bm{u}_{U}}}^{[m]},{\tau}_{{\bm{b}_{B}\bm{r}_{R}}}^{[m]},\bm{\Phi}_{R}^{[m]\mathrm{T}}, {\epsilon}^{[m]},
  {\beta}_{\mathrm{R}}^{[m]}, {\beta}_{\mathrm{I}}^{[m]}  \right]^{\mathrm{T}}.
\end{array}
\end{align}
\end{figure*}
Next, $\bm{\eta}$ is defined as $\bm{\eta} \triangleq\left[\bm{\eta}^{[0]\mathrm{T}},\bm{\eta}^{[1]\mathrm{T}}, \bm{\eta}^{[2]\mathrm{T}}, \cdots, \bm{\eta}^{[M_1]\mathrm{T}}\right]^{\mathrm{T}},$
where the parameters associated with the $m^{\text{th}}$ RIS path can also be divided into the   geometric channel parameters,  $ \bm{\eta}_{1}^{[m]}$, and  the nuisance parameters, $ \bm{\eta}_{2}^{[m]}$, such that $ \bm{\eta}^{[m]} \triangleq\left[\bm{\eta}_{1}^{{[m]}\mathrm{T}}, \bm{\eta}_{2}^{{[m]}\mathrm{T}}
\right]^{\mathrm{T}}$. 
\begin{remark}
The parameterization of channel parameters discussed in this section has implicitly assumed that the BS is perfectly located i.e., $\bm{p}_b = \Tilde{\bm{p}}_b, \forall b$. A perfect knowledge of the location of all RISs enables the exclusion of  $\{\bm{\theta}_{{\bm{b}_{B}\bm{r}_{R}}}^{}, \bm{\phi}_{{\bm{b}_{B}\bm{r}_{R}}}^{},
\bm{\tau}_{{\bm{b}_{B}\bm{r}_{R}}}^{},\bm{\Phi}^{}_{R} \}$ from the parameter vector. 
Henceforth, the BS is assumed perfectly located, and it is considered to serve as the global origin.  
\end{remark}

\subsection{Path Loss}
Before proceeding, we show that the orientation offset of a misoriented RIS is present in the pathloss; this affects the derivations of the error bounds for the channel parameters. The power received at the $u^{\text{th}}$ UE receive antenna  through the $r^{\text{th}}$ element of the $m^{\text{th}}$ RIS from the $b^{\text{th}}$ transmit antenna is expressed as \cite{ellingson2021path}
\begin{equation}
\label{equ:receive_power}
\begin{aligned}
P_{o,bru}^{[m]} &= P_{o,b} G_{br}^{}\left(\bm{a}_{{\bm{p}_b}{\bm{p}_r}}^{{[m]}}\right) G_{ru}^{}\left(\bm{a}_{{\bm{p}_r}{\bm{p}_u}}^{{[m]}}\right) \\ &\times  \frac{G_{rb}^{}\left(\bm{a}_{{\bm{p}_r}{\bm{p}_b}}^{{[m]}}\right) G_{ur}^{}\left(\bm{a}_{{\bm{p}_u}{\bm{p}_r}}^{{[m]}}\right)}{(d^{{[m]}}_{{\bm{p}_b}{\bm{p}_r}})^2 (d^{{[m]}}_{{\bm{p}_r}{\bm{p}_u}})^2  } \bigg(\frac{\lambda}{4\pi}\bigg)^4 \epsilon_{p}^2,
\end{aligned}
\end{equation}
where $P_{o,b}$ is the transmit power at the $b^{\text{th}}$ antenna, $\bm{a}_{{\bm{g}_g}{\bm{v}_v}}^{}$ is the unit direction vector from the  $g^{\text{th}}$ element  on the $G^{\text{{th}}}$ entity  to the  $v^{\text{th}}$ element on the $V^{\text{{th}}}$ entity, $G_{gv}\left(\bm{a}_{{\bm{g}_g}{\bm{v}_v}}^{}\right)$ is the gain in the direction specified by the unit vector, and $\epsilon_{p}$ is the efficiency for practical antennas.

\begin{proposition}
\label{proposition:direction}
The direction vector from the $r^\text{th}$ element on the $m^\text{th}$ RIS which is  misoriented  and misaligned  to the $b^\text{th}$ antenna on a misoriented  and misaligned BS is defined by  $\bm{a}_{{\bm{p}}_{{r}^{}}{\bm{p}}_{{b}^{}}^{}}^{[m]} =  ({\bm{p}}_{{b}^{}}^{}  - {\bm{p}}_{{r}^{}}^{[m]} ) / d_{{\bm{p}}_{{r}^{}}{\bm{p}}_{{b}^{}}^{}}^{[m]}$. Also, the direction vector from the $r^\text{th}$ element on the $m^\text{th}$ RIS which is  misoriented  and misaligned  to the $u^\text{th}$ antenna on a  misoriented  and misaligned UE is defined by  $\bm{a}_{{\bm{p}}_{{r}^{}}{\bm{p}}_{{u}^{}}^{}}^{[m]} =  ({\bm{p}}_{{u}^{}}^{}  - {\bm{p}}_{{r}^{}}^{[m]} ) / d_{{\bm{p}}_{{r}^{}}{\bm{p}}_{{u}^{}}^{}}^{[m]}$. The corresponding gains in these directions can be expressed as $ G_{rb}^{}\left(\bm{a}_{{\bm{p}}_{{r}^{}}{\bm{p}}_{{b}^{}}}^{[m]}\right) = \pi\left(\bm{a}_{{\bm{p}}_{{r}^{}}{\bm{p}}_{{b}^{}}^{}}^{[m]{\mathrm{T}}} \bm{a}_{{\bm{p}}_{{R}^{}}}^{[m]}\right)^{2 q_{0}} $ and $ G_{ru}^{}\left(\bm{a}_{{\bm{p}}_{{r}^{}}{\bm{p}}_{{u}^{}}^{}}^{[m]}\right) = \pi\left(\bm{a}_{{\bm{p}}_{{r}^{}}{\bm{p}}_{{u}^{}}^{}}^{[m]{\mathrm{T}}} \bm{a}_{{\bm{p}}_{{R}^{}}}^{[m]}\right)^{2 q_{0}} $, respectively, where $q_{0}$ is a constant used to vary the gain of the RIS and $\bm{a}_{{\bm{p}}_{{R}^{}}}^{[m]}$ is the unit vector pointing out of the misoriented  and misaligned  RIS. 
\end{proposition}
\begin{proof}
Since, the RIS elements are devised as a conducting ground screen with electrically-small low-gain elements, the gain is almost entirely specified by the direction cosine of the vector pointing out of the $r^\text{th}$ element on the $m^\text{th}$ RIS and towards the BS or UE\cite{stutzman2012antenna}. Hence, the proof follows under the condition that the RISs have
elements spaced by half a wavelength resulting in the RIS's area being equal to the sum of the effective apertures of all RIS elements, i.e., square of half the wavelength.
\end{proof}

The term $q_0$ defines a gain intrinsically related to the physical design of the antennas or RIS elements \cite{stutzman2012antenna,ellingson2021path}. The investigation of this parameter depends on the electromagnetic properties connected to the antenna design, which is out of the scope of this work.
\color{black}

\begin{proposition}
\label{proposition:unit_direction}
 The unit vector pointing out of the misoriented  and misaligned  RIS is $\bm{a}_{{\bm{p}}_{{R}^{}}}^{[m]} =  \bm{Q}_{R}^{[m]} \bm{a}_{\Tilde{\bm{p}}_{{R}^{}}}^{[m]}$, where $\bm{a}_{\Tilde{\bm{p}}_{{R}^{}}}^{[m]}$ is the unit vector of a perfectly located RIS.
\end{proposition}
\begin{proof}
The proof follows from the definition of misorientation and misalignment.
\end{proof}

\begin{assumption}
\label{assumption:assumption_RIS_normal_pathloss_modeling}
Each of the transmit and receive antennas have isotropic gains such that $G_{br}^{}\left(\bm{a}_{{\bm{p}_b}{\bm{p}_r}}^{{[m]}}\right) = G_{B}^{}$ and $ G_{ru}^{}\left(\bm{a}_{{\bm{p}_r}{\bm{p}_u}}^{{[m]}}\right)  = G_{U}^{}$.
\end{assumption}

\begin{proposition}
\label{proposition:gain}
With Assumption \ref{assumption:assumption_RIS_normal_pathloss_modeling}, if the $m^{\text{th}}$ RIS is misoriented and misaligned,  and it reflects signals to a UE which is misoriented and misaligned, the corresponding gain $G_{ru}^{}\left(\bm{a}_{{\bm{p}}_{{r}^{}}{\bm{p}}_{{u}^{}}^{}}^{[m]}\right)$ is $$\begin{aligned}
&G_{ru}^{}\left(\bm{a}_{{\bm{p}}_{{r}^{}}{\bm{p}}_{{u}^{}}^{}}^{[m]}\right) = \pi / \left(d_{{\bm{p}}_{{r}^{}}{\bm{p}}_{{u}^{}}^{}}^{[m]}\right)^{2q_0}  \\ & \times \bigg[ \left( d_{{\bm{p}}_{{R}^{}}{  \bm{p}}_{{U}^{}}}^{[m]}{\Delta_{{\bm{p}}_{{R}^{}}{\bm{p}}_{{U}^{}}}^{[m]}} + \bm{Q}_{U} \Tilde{\bm{s}}_{u} - \bm{Q}_{R}^{[m]} \Tilde{\bm{s}}_{r}^{[m]} \right)^{\mathrm{T}} \bm{Q}_{R}^{[m]} \bm{a}_{\Tilde{\bm{p}}_{{R}^{}}}^{[m]} \bigg]^{2q_0}.\end{aligned} $$ Similarly,  the gain from the $r^\text{th}$ element on the  $m^{\text{th}}$ RIS which is misoriented and misaligned in the direction  of the $b^\text{th}$ antenna on a misoriented  and misaligned BS is  $$ \begin{aligned} & G_{rb}^{}\left(\bm{a}_{{\bm{p}}_{{r}^{}}{\bm{p}}_{{b}^{}}^{}}^{[m]}\right) = \pi / \left(d_{{\bm{p}}_{{r}^{}}{\bm{p}}_{{b}^{}}^{}}^{[m]}\right)^{2q_0} \\ & \times  \bigg[ \left( d_{{\bm{p}}_{{R}^{}}{  \bm{p}}_{{B}^{}}}^{[m]}{\Delta_{{\bm{p}}_{{R}^{}}{\bm{p}}_{{B}^{}}}^{[m]}} + \bm{Q}_{B} \Tilde{\bm{s}}_{b} - \bm{Q}_{R}^{[m]} \Tilde{\bm{s}}_{r}^{[m]} \right)^{\mathrm{T}} \bm{Q}_{R}^{[m]} \bm{a}_{\Tilde{\bm{p}}_{{R}^{}}}^{[m]} \bigg]^{2q_0}.\end{aligned}$$
\begin{proof}
The proof follows from the definitions of misorientation and misalignment and both Propositions \ref{proposition:direction} and \ref{proposition:unit_direction}.
\end{proof}
\end{proposition}

The above propositions only present the antenna gains as a function of RIS misorientation. These propositions are an intermediary step to present the pathloss in Lemma 2. Because the derived pathloss is a function of the antenna gain, it is also a function of the orientation of the RIS. Now,  because the FIM depends on the pathloss, the presence of an orientation offset affects the lower bound on any unbiased estimator.
\color{black}

\begin{lemma}
\label{lemma:pathloss}
With Assumption \ref{assumption:assumption_RIS_normal_pathloss_modeling}, Proposition \ref{proposition:direction}, and Proposition \ref{proposition:gain}, the pathloss experienced by a signal transmitted  from the $b^{\text{th}}$ transmit antenna on a misoriented and misaligned BS through the $r^{\text{th}}$ element of the $m^{\text{th}}$ RIS which is also misoriented and misaligned to the $u^{\text{th}}$ receive antenna on a misoriented and misaligned  UE is presented in (\ref{equ_lemma:pathloss}).
\begin{figure*}
    \begin{align}
    \begin{split}
        \label{equ_lemma:pathloss}
        {\rho_{\bm{p}_{r}}^{[{m}]}} &= \frac{\lambda^2}{16 \pi \left(d_{{\bm{p}}_{{b}^{}}{\bm{p}}_{{r}}}^{[m]}\right)^{q_0 +1}
        \left(d_{ {\bm{p}}_{{r}}{\bm{p}}_{{u}^{}}^{}}^{[m]}\right)^{q_0 +1}}  \bigg[ \bigg( d_{{\bm{p}}_{{R}}{\bm{p}}_{{B}^{}}}^{[m]}{\Delta_{{\bm{p}}_{{R}}{\bm{p}}_{{B}^{}}}^{[m]}} +  \bm{Q}_{B} \Tilde{\bm{s}}_{b} - \bm{Q}_{R}^{[m]} \Tilde{\bm{s}}_{r}^{[m]} \bigg)^{\mathrm{T}} \bm{Q}_{R}^{[m]} \bm{a}_{\Tilde{\bm{p}}_{{R}^{}}}^{[m]}  \bigg]^{q_0} \\ &\times\bigg[ \bigg( d_{{\bm{p}}_{{R}}{\bm{p}}_{{U}^{}}}^{[m]}{\Delta_{{\bm{p}}_{{R}}{\bm{p}}_{{U}^{}}}^{[m]}}   + \bm{Q}_{U} \Tilde{\bm{s}}_{u} - \bm{Q}_{R}^{[m]} \Tilde{\bm{s}}_{r}^{[m]} \bigg)^{\mathrm{T}} \bm{Q}_{R}^{[m]} \bm{a}_{\Tilde{\bm{p}}_{{R}^{}}}^{[m]}  \bigg]^{q_0} \epsilon_{p} =        \frac{\lambda^2 \rho_{r, {\bm{p}}_{{r}}{\bm{p}}_{{b}^{}}}^{[m]{q_0}}  \rho_{r,{\bm{p}}_{{r}}{\bm{p}}_{{u}^{}}}^{[m]{q_0}}}{16 \pi \left(d_{{\bm{p}}_{{b}^{}}{\bm{p}}_{{r}}}^{[m]}\right)^{q_0 +1}
        \left(d_{ {\bm{p}}_{{r}}{\bm{p}}_{{u}^{}}^{}}^{[m]}\right)^{q_0 +1}}    \epsilon_{p}.
            \end{split}
    \end{align}
\end{figure*}
In that equation, $ \rho_{r, {\bm{p}}_{{r}}{\bm{p}}_{{b}^{}}}^{[m]} =  \bigg( d_{{\bm{p}}_{{R}}{\bm{p}}_{{B}^{}}}^{[m]}{\Delta_{{\bm{p}}_{{R}}{\bm{p}}_{{B}^{}}}^{[m]}} +  \bm{Q}_{B} \Tilde{\bm{s}}_{b} - \bm{Q}_{R}^{[m]} \Tilde{\bm{s}}_{r}^{[m]} \bigg)^{\mathrm{T}} \bm{Q}_{R}^{[m]} \bm{a}_{\Tilde{\bm{p}}_{{R}^{}}}^{[m]}$ and $\rho_{r,{\bm{p}}_{{r}}{\bm{p}}_{{u}^{}}}^{[m]} =  \bigg( d_{{\bm{p}}_{{R}}{\bm{p}}_{{U}^{}}}^{[m]}{\Delta_{{\bm{p}}_{{R}}{\bm{p}}_{{U}^{}}}^{[m]}}   + \bm{Q}_{U} \Tilde{\bm{s}}_{u} - \bm{Q}_{R}^{[m]} \Tilde{\bm{s}}_{r}^{[m]} \bigg)^{\mathrm{T}} \bm{Q}_{R}^{[m]} \bm{a}_{\Tilde{\bm{p}}_{{R}^{}}}^{[m]}$. 
 \end{lemma}
\begin{proof}
The proof follows from (\ref{equ:receive_power}) and Proposition \ref{proposition:gain}.
\end{proof}
\begin{remark}
Lemma  \ref{lemma:pathloss} indicates that the orientation of the RIS is contained in the pathloss. Because the orientation is also present in any array response, any entries in the FIM related to RIS orientation would require first derivatives obtained through the product rule (see Appendix \ref{Appendix:first_derivative}).
\end{remark} 
  \subsection{Mathematical Preliminaries}
Here, we present mathematical preliminaries needed to derive bounds for both deterministic and random channel parameters. We note that the error covariance matrix of an unbiased estimator, $\hat{\bm{\eta}}$, satisfies the following information inequality
$
\mathbb{E}_{\bm{y}; \boldsymbol{\eta}}\left\{(\hat{\boldsymbol{\eta}}-\boldsymbol{\eta})(\hat{\boldsymbol{\eta}}-\boldsymbol{\eta})^{\mathrm{T}}\right\} \succeq \mathbf{J}_{ \bm{\bm{y}}; \bm{\eta}}^{-1},
$
where $\mathbf{J}_{ \bm{\bm{y}}; \bm{\eta}}$ is the general FIM for the parameter vector $\boldsymbol{\eta}.$
\begin{definition}
\label{definition_FIM_1}
The general FIM for a parameter vector, $\bm{\eta}$, defined as $\mathbf{J}_{ \bm{\bm{y}}; \bm{\eta}} =  \bm{F}_{\bm{\bm{y}}; \bm{\eta} }(\bm{y}; \bm{\eta} ; \bm{\eta}, \bm{\eta})$ is the summation of the FIM obtained from the likelihood due to the observations defined as  $\mathbf{J}_{\bm{y}|\bm{\eta}} =  \bm{F}_{{\bm{y} }}(\bm{y}| \bm{\eta} ;\bm{\eta},\bm{\eta})$ and the FIM from {\em a priori} information about the parameter vector defined as $\mathbf{J}_{ \bm{\eta}} =  \bm{F}_{{\bm{\eta} }}( \bm{\eta} ;\bm{\eta},\bm{\eta})$. In mathematical terms, we have
\begin{equation}
\label{definition_equ:definition_FIM_1}
\begin{aligned}
\mathbf{J}_{ \bm{\bm{y}}; \bm{\eta}} &\triangleq 
\mathbb{E}_{\bm{y};\bm{\eta}_{}}\left[-\frac{\partial^{2} \ln \chi(\bm{y}_{};  \bm{\eta}_{} )}{\partial \bm{\eta}_{} \partial \bm{\eta}_{}^{\mathrm{T}}}\right] \\
&= -\mathbb{E}_{\bm{y}  }\left[\frac{\partial^{2} \ln \chi(\bm{y}_{}|  \bm{\eta}_{} )}{\partial \bm{\eta}_{} \partial \bm{\eta}_{}^{\mathrm{T}}}\right] -\mathbb{E}_{ \bm{\eta}_{}}\left[\frac{\partial^{2} \ln \chi(  \bm{\eta}_{} )}{\partial \bm{\eta}_{} \partial \bm{\eta}_{}^{\mathrm{T}}}\right] \\ &= \mathbf{J}_{\bm{y}|\bm{\eta}} + \mathbf{J}_{ \bm{\eta}},
\end{aligned}
\end{equation}
where  $\chi(\bm{y}_{};  \bm{\eta}_{} )$ denotes the probability density function (PDF) of $\bm{y}$ and $\bm{\eta}$.
\end{definition}
The number of elements in the Fisher information matrix increases quadratically with an increase in the number of parameters; hence, due to the potentially large dimensions of the Fisher information matrix, it is beneficial to focus on the partition corresponding to the parameters of interest. The Schur's complement\cite{horn2012matrix} provides a method of achieving this partition and the resulting partition is the EFIM.
\begin{definition}
\label{definition_EFIM}
Given a parameter vector, $ \bm{\eta}_{} \triangleq\left[\bm{\eta}_{1}^{\mathrm{T}}, \bm{\eta}_{2}^{\mathrm{T}}
\right]^{\mathrm{T}}$, where $\bm{\eta}_{1}$ is the parameter of interest, the resultant FIM has the structure 
$$
\mathbf{J}_{ \bm{\bm{y}}; \bm{\eta}}=\left[\begin{array}{cc}
\mathbf{J}_{ \bm{\bm{y}}; \bm{\eta}_1}^{}  & \mathbf{J}_{ \bm{\bm{y}}; \bm{\eta}_1, \bm{\eta}_2}^{} \\
 \mathbf{J}_{ \bm{\bm{y}}; \bm{\eta}_1, \bm{\eta}_2}^{\mathrm{T}} &\mathbf{J}_{ \bm{\bm{y}}; \bm{\eta}_2}^{}
\end{array}\right],
$$
where $\bm{\eta} \in \mathbb{R}^{N}, \bm{\eta}_{1} \in \mathbb{R}^{n}, \mathbf{J}_{ \bm{\bm{y}}; \bm{\eta}_1}^{} \in \mathbb{R}^{n \times n},  \mathbf{J}_{ \bm{\bm{y}}; \bm{\eta}_1, \bm{\eta}_2}\in \mathbb{R}^{n \times(N-n)}$, and $\mathbf{J}_{ \bm{\bm{y}}; \bm{\eta}_2}^{}\in$ $\mathbb{R}^{(N-n) \times(N-n)}$ with $n<N$, 
and the EFIM \cite{5571900} of  parameter ${\bm{\eta}_{1}}$ is given by 
$\mathbf{J}_{ \bm{\bm{y}}; \bm{\eta}_1}^{\mathrm{e}} =\mathbf{J}_{ \bm{\bm{y}}; \bm{\eta}_1}^{} - \mathbf{J}_{ \bm{\bm{y}}; \bm{\eta}_1}^{nu} =\mathbf{J}_{ \bm{\bm{y}}; \bm{\eta}_1}^{}-
\mathbf{J}_{ \bm{\bm{y}}; \bm{\eta}_1, \bm{\eta}_2}^{} \mathbf{J}_{ \bm{\bm{y}}; \bm{\eta}_2}^{-1} \mathbf{J}_{ \bm{\bm{y}}; \bm{\eta}_1, \bm{\eta}_2}^{\mathrm{T}}.$
This EFIM captures all the required information about the parameters of interest present in the FIM; as observed from the relation $(\mathbf{J}_{ \bm{\bm{y}}; \bm{\eta}_1}^{\mathrm{e}})^{-1} = [\mathbf{J}_{ \bm{\bm{y}}; \bm{\eta}}^{-1}]_{[1:n,1:n]}$.
\end{definition}
 \begin{definition}
 \label{definition:parameter_estimatable}
The parameter vector, $\bm{\eta}$, is estimatable iff the general FIM, $\mathbf{J}_{ \bm{\bm{y}}; \bm{\eta}} \succ 0$.
\end{definition}
\begin{proposition}
\label{proposition:EFIM_channel}
Defining a parameter vector $\bm{\eta}= \left[\bm{\eta}_{1}^{\mathrm{T}}, \eta_{N_{\eta} }\right]^{\mathrm{T}}  = [\eta_1, \eta_2, \cdots , \eta_{N_{\eta - 1}}, \eta_{N_{\eta} }]^{\mathrm{T}}$, where $\bm{\eta}_{1}$ is the parameter of interest and $\eta_{N_{\eta} }$ is a nuisance parameter. For a subset of the parameter vector, $\Tilde{\bm{\eta}} = [\eta_v , \eta_{N_{\eta}}]^{\mathrm{T}}$  where $v \ \in \{ 1,2,\cdots, {N_{\eta - 1}}\}$; if the resultant EFIM, $ \mathrm{J}_{ \bm{\bm{y}}; \Tilde{\bm{\eta}}}^{\mathrm{e}} = 0$, then the EFIM, $\mathbf{J}_{ \bm{\bm{y}}; \bm{\eta}_1}^{\mathrm{e}} \nsucc 0$.
\end{proposition}
\begin{proof}
See Appendix \ref{appendix:proof_PD}.
\end{proof}

In subsequent sections, Definition 3 and Proposition 5 will be used to show under what propagation regimes the orientation offset can be estimated.
\color{black}
\subsection{Fisher Information Matrix for Channel Parameters}
To derive the FIM for the channel parameters from the received signals, we consider the $t^{\text{th}}$ OFDM transmission with $N$ subcarriers, and the likelihood expression conditioned on the parameter vector, $\bm{\eta}$,  is\footnote{The likelihood is easily expressible in scalar terms by considering the signal at each receive antenna as independent observations.}  
\begin{equation}
\label{equ:likelihood}
\begin{aligned}
    &\chi(\bm{y}_{t}[n]|  \bm{\eta}_{}) \\ & \propto \exp \left\{\frac{2}{N_{0}} \sum_{n=1}^{N} \Re\left\{\bm{\mu}_{t}^{\mathrm{H}}[n] \bm{y}_{t}[n]\right\}-\frac{1}{N_{0}} \sum_{n=1}^{N}\|\bm{\mu}_{t}[n]\|_{2}^{2}\right\}.
    \end{aligned}
\end{equation}
The FIM from observations, $\mathbf{J}_{\bm{y}|\bm{\eta}}$, of the random vector, $\bm{y}$, is obtained by substituting (\ref{equ:likelihood}) into  (\ref{definition_equ:definition_FIM_1}) in Definition \ref{definition_FIM_1}. The {\em a priori} information about channel parameters and RIS orientation is incorporated through the likelihood which is expressed as
    $\chi(\bm{\eta}_{})  = \prod_{m = 1}^{M_1}\chi( \bm{\eta}^{[0]}|\bm{p}_U )\chi( \bm{\eta}^{[m]}|\bm{p}^{[m]}_R,  \bm{p}_U ).$
The equation above assumes that there is independence between channel parameters from distinct RISs, and the LOS and RIS channel parameters are independent. Furthermore, we assume that the UE orientation is independent of the orientation of the RISs and the parameters in each distinct path\footnote{This relates to both LOS and RIS paths.} are independent of each other. The  Fisher information matrix from {\em a priori} information is $\mathbf{J}_{\bm{\eta}}=\text{diag}\left[
 \bm{\Xi}_{\bm{\eta}^{[0]},\bm{\eta}^{[0]}}^{\bm{\bm{\eta}}^{[0]}},\bm{\Xi}_{\bm{\eta}^{[1]},\bm{\eta}^{[1]}}^{\bm{\eta}^{[1]}}, \cdots,\bm{\Xi}_{\bm{\eta}^{[M_1]},\bm{\eta}^{[M_1]}}^{\bm{\eta}^{[M_1]}} \right],$
where  $\bm{\Xi}_{\bm{\eta}^{[0]},\bm{\eta}^{[0]}}^{\bm{\eta}^{[0]}} = \bm{F}_{{\bm{\eta} }}( \bm{\eta}^{[0]}|\bm{p}_U ;\bm{\eta^{[0]}},\bm{\eta}^{[0]})$. The FIM from the {\em a priori} information related to the likelihood $\chi(\bm{\eta}^{[m]}| \bm{p}^{[m]}_{R},\bm{p}_U)$ is $\bm{\Xi}_{\bm{\eta}^{[m]},\bm{\eta}^{[m]}}^{\bm{\eta}^{[m]}} = \bm{F}_{{\bm{\eta} }}( \bm{\eta}^{[m]}|\bm{p}^{[m]}_R,\bm{p}_U ;\bm{\eta}^{[m]},\bm{\eta}^{[m]})$.  Note that $\bm{\Xi}_{{\bm{\eta}^{[m]}},{\bm{\eta}^{[m]}}}^{\bm{\eta}^{[m]}}, m \in \{0,1,\cdots, M_1 \}$  are  diagonal matrices as the parameters in each path are assumed to be independent from each other\cite{kay1993fundamentals}. \color{black} The general FIM is obtained from $\mathbf{J}_{\bm{y}|\bm{\eta}}$ and $\mathbf{J}_{\bm{\eta}}$ according to Definition \ref{definition_FIM_1}. The FIM of the location parameters, $\bm{\kappa}$\footnote{The location parameter is properly defined in Section \ref{section_fim_for_kappa}.}, is obtained by the bijective transformation $\mathbf{J}_{\bm{y}|\bm{\kappa}} \triangleq \mathbf{\Upsilon}_{\bm{\kappa}} \mathbf{J}_{\bm{y}|\bm{\eta}} \mathbf{\Upsilon}_{\bm{\kappa}}^{\mathrm{T}}$, where $\mathbf{\Upsilon}_{\bm{\kappa}}$ represents derivatives of the non-linear relationship between the channel parameters and the location parameters. After this transformation, the general FIM for the location parameters is obtained through Definition \ref{definition_FIM_1}\footnote{A more rigorous explanation and investigation of the FIMs for the location parameters is presented in Section \ref{section_fim_for_kappa}.}. The FIM $\mathbf{J}_{ \bm{\bm{y}}| \bm{\eta}}$ of the channel parameters specified by (\ref{equ:channel_params}) has $(M_1 + 1)^2$ submatrices, and these submatrices are obtained using
\begin{equation}
\label{equ:observation_FIM_1}
[\mathbf{J}_{ \bm{\bm{y}}| \bm{\eta}}]_{[v,g]} = 
\frac{2}{N_{0}} \sum_{u=1}^{N_U}\sum_{n=1}^{N}\sum_{t=1}^{T} \Re\left\{\nabla^{\mathrm{H}}_{{[\bm{\eta}]_{[v]}} } {\mu}_{t, {{u}}}[n] \nabla_{{[\bm{\eta}]_{[g]}} } {\mu}_{t, {{u}}}[n]\right\},
\end{equation}
where the first derivatives are presented in Appendix \ref{Appendix:first_derivative}. Note that the FIM obtained through the first derivatives in Appendix \ref{Appendix:first_derivative} is challenging to analyze. Hence, we make some assumptions for tractability in the presentation of the necessary conditions for estimating both deterministic and random channel parameters useful for localization.
\begin{assumption}
\label{assumption:RIS_restrict}
The degrees of freedom provided by the $T$ OFDM symbols are used to impose the following constraints\cite{9528041,9625826,keykhosravi2021multi,emenonye2022fundamentals}
$$
\begin{aligned}
\sum_{t = 1}^{T} \gamma_t^{[{m}]} = 0, \; \; \sum_{t = 1}^{T} \gamma_t^{[{m}] \mathrm{H}}\gamma_t^{[{m}]} = 1, \; \; \forall  \; m, \text{ and } \\ 
\sum_{t = 1}^{T} \gamma_t^{[{m}_1] \mathrm{H}}\gamma_t^{[{m}_2]} = 0, \; \; \forall  \; m_1 \neq m_2.
\end{aligned}
$$
These constraints make the LOS and the distinct RIS paths separable, resulting in the block diagonalization of the FIM from the observations, $\mathbf{J}_{\bm{y}|\bm{\eta}}$. This block diagonalization helps analyze the information provided by each path and the constraints are achieved by assigning distinct columns of discrete Fourier transforms to distinct RISs.
\end{assumption}
\begin{assumption}
\label{assumption:pathloss_subsumed}
For tractability, we assume perfect synchronization\footnote{To acquire a fundamental understanding of the localization problem, this assumption is common in the literature \cite{8240645,8515231,8356190,fascista2021downlink,8755880,9082200}. Although it appears restrictive, it can be attained using a preliminary two-way synchronization approach or a joint localization       and synchronization approach \cite{9064586,1097833}. These approaches are beyond the scope of this paper but could certainly be incorporated without significant difficulty. \color{black}} \color{black} and narrowband transmissions  in the rest of this work such that $\lambda_{n} = \lambda \; \forall n$. To analyze the impact of misorientation, we assume that the pathloss of the $m^{\text{th}}$ RIS path is constant, ${\rho_{\bm{r}_r}^{[{m}]}} = {\rho_{\bm{r}_R}^{[{m}]}} \; \forall r$. We combine the pathloss and the noise power into a composite noise power, $(\sigma^{[m]})^2$. Finally, we define the SNR as $P_{}/(\sigma^{[m]})^2$, where $P =  P_{o,b} G_{R}^{} G_{U}^{}.$ 
\end{assumption}
Assumption \ref{assumption:RIS_restrict} is used to ensure the diagonalization of the FIM. 
With this diagonalization, the general FIM obtained from $\mathbf{J}_{\bm{y}|\bm{\eta}}$ and $\mathbf{J}_{\bm{\eta}}$ according to Definition \ref{definition_FIM_1} is also a diagonal matrix $\mathbf{J}_{ \bm{\bm{y}}; \bm{\eta}} = \text{diag}\left[\mathbf{J}_{ \bm{\bm{y}}; \bm{\eta}^{[0]}}, \mathbf{J}_{ \bm{\bm{y}}; \bm{\eta}^{[1]}}, \cdots, \mathbf{J}_{ \bm{\bm{y}}; \bm{\eta}^{[M_1]}}\right].$
With Definition \ref{definition:parameter_estimatable} and the block diagonal nature of the general FIM, $\mathbf{J}_{ \bm{\bm{y}}; \bm{\eta}}$, it is sufficient to determine estimatability of a parameter vector, $\bm{\eta}$, by establishing the positive definiteness of all the diagonal entries in $\mathbf{J}_{ \bm{\bm{y}}; \bm{\eta}}$. However, not all parameters are useful for localization, hence a more efficient metric is the general EFIM $\mathbf{J}_{ \bm{\bm{y}}; \bm{\eta}_1}^{\mathrm{e}} = \text{diag}\left[\mathbf{J}_{ \bm{\bm{y}}; \bm{\eta}_1^{[0]}}^{\mathrm{e}}, \mathbf{J}_{ \bm{\bm{y}}; \bm{\eta}_1^{[1]}}^{\mathrm{e}}, \cdots,\mathbf{J}_{ \bm{\bm{y}}; \bm{\eta}_1^{[M_1]}}^{\mathrm{e}}\right].$
The term $\mathbf{J}_{ \bm{\bm{y}}; \bm{\eta}_1^{[m]}}^{\mathrm{e}}$, is the general EFIM of the $m^{\text{th}}$ path  derived based on the vector of parameters of interest $\bm{\eta}_1^{[m]}$\footnote{Definition \ref{definition:parameter_estimatable} is easily extended to define necessary and sufficient conditions for the estimatability of the parameter of interest, $\bm{\eta}_1^{[m]}$ based on the EFIM.}.

  \subsection{Fisher Information for the LOS Parameters}
The FIM for the LOS parameters based on the observations is obtained by using
(\ref{equ:observation_FIM_1}) and the general FIM for the LOS parameters is obtained using Definition \ref{definition_FIM_1}. The EFIM of the LOS  geometric channel parameters, $\mathbf{J}_{ \bm{\bm{y}}; \bm{\eta}_1^{[0]}}^{\mathrm{e}}$, is obtained through Definition \ref{definition_EFIM}. This EFIM always satisfies the  estimatability condition in Definition \ref{definition:parameter_estimatable}.

 \subsection{Fisher Information for the RIS Parameters}
The RIS paths have identical channel parameters. Hence it suffices to analyze a single RIS path. Therefore, we apply Assumption \ref{assumption:pathloss_subsumed} and investigate the estimatability of the geometric channel parameters related to the $m^{\text{th}}$ RIS.
\begin{lemma}
\label{lemma:farfield_channel_parameter_estimatable}
In the far-field, the vector of geometric channel parameters is not estimatatible at the UE without {\em a priori} information about the orientation of the RIS or {\em a priori} information about the channel complex path gains.
\end{lemma}

\begin{proof}
See Appendix \ref{Appendix:farfield_channel_parameter_estimatable}
\end{proof}

\begin{lemma}
\label{lemma:near-field_channel_parameter_estimatable}
In the near-field, the vector of geometric channel parameters is not estimatatible at the UE without {\em a  priori} information about the channel complex path gains if there is also no {\em a priori} information about the orientation of the RIS and if $N_U < 2$.  If $N_U > 1$, the possibility of estimating the parameter vector exists.
\end{lemma}

\begin{proof}
See Appendix \ref{Appendix:near-field_channel_parameter_estimatable}.
\end{proof}

\begin{theorem}
\label{theorem:channel_parameter_estimatable}
In the far-field, the estimation and correction of RIS orientation offset based on the received signals at the UE is not possible due to the absence of {\em a priori} information about the channel's complex path gains. However, in the near-field, this correction of RIS orientation offset is not hindered by the absence of {\em a priori} information about the channel complex path gains when $N_U > 1$.
\end{theorem}

\begin{proof}
The proof follows as a consequence of Lemmas \ref{lemma:farfield_channel_parameter_estimatable} and \ref{lemma:near-field_channel_parameter_estimatable}.
\end{proof}
\begin{remark}
The complex path gain can be normalized through the SNR to reduce from a complex path gain to an unknown phase offset term. Hence, we can restate the above Lemmas and Theorem, focusing only on the {\em a priori} information about the unknown phase offset. 
Because quantifying {\em a priori} information about this unknown phase offset is virtually impossible, Theorem \ref{theorem:channel_parameter_estimatable} practically implies the possibility of estimating and correcting an RIS orientation offset only exists in the near-field.
\end{remark}
\section{FIM for Location Parameters}
\label{section_fim_for_kappa}
We investigate the FIM for location parameters in this section. The UE location parameter is defined as $\bm{\kappa}_{0} = \left[ \bm{p}^{\mathrm{T}}_{U}, \bm{\Phi}^{\mathrm{T}}_{U}  \right]^{\mathrm{T}}$.  This section focuses exclusively on cases in which the vector of geometric channel parameters is estimatatible. The goal is to provide bounds for the UE position while considering the impact of the location of the misoriented and/or misaligned RISs.  Two ways of incorporating the RIS uncertainties are: i) to consider the position and orientation vectors of the RIS as estimation parameters or ii) to consider the misalignment and orientation vectors of the RIS  as estimation parameters.

\subsection{Position and Orientation Vectors of the RISs as Estimation Parameters}
In this scheme, the location parameter related  to the $m^{\text{th}}$ RIS is  $\bm{\kappa}^{[m]} = \left[ \bm{p}^{[m]\mathrm{T}}_{R}, \bm{\Phi}^{[m]\mathrm{T}}_{R}  \right]^{\mathrm{T}}$.  The location parameters of the UE and the location parameters of the uncertain RISs are collectively represented as $\bm{\kappa}_{} \triangleq\left[\bm{\kappa}^{\mathrm{T}}_{0}, \bm{\kappa}^{[1]\mathrm{T}},\cdots,\bm{\kappa}^{[M_1]\mathrm{T}} \right]^{\mathrm{T}}.$
Focusing on the UE location, the parameter $\bm{\kappa}_{}$ can be grouped into the UE location parameter of interest and  RIS location parameters which are not of interest, i.e., $\bm{\kappa}_{} = [\bm{\kappa}_{1}^{\mathrm{T}}, \bm{\kappa}_{2}^{\mathrm{T}}]^{\mathrm{T}}$ where $\bm{\kappa}_{2} = \left[ \bm{\kappa}^{[1]\mathrm{T}},\cdots,\bm{\kappa}^{[M_1]\mathrm{T}} \right]^{\mathrm{T}}$. The FIM from the observations is obtained by the bijective transformation\footnote{Note that this transformation is applied to the EFIM obtained by applying Definition \ref{definition_EFIM} to the FIM of the channel parameters derived from the observations. It is never applied in this paper to the general EFIM.}, $\mathbf{J}_{\bm{y}|\bm{\kappa}} \triangleq \mathbf{\Upsilon}_{\bm{\kappa}} \mathbf{J}_{ \bm{\bm{y}}| \bm{\eta}_1}^{\mathrm{e}} \mathbf{\Upsilon}_{\bm{\kappa}}^{\mathrm{T}}$, where $\mathbf{\Upsilon}_{\bm{k}}$ represents derivatives of the non-linear relationship between the channel parameters and the location parameters\cite{kay1993fundamentals}. The {\em a priori} information about the location parameters is incorporated through the likelihood 
\begin{equation}
\label{equ:prior_location_likelihhood_1}
    \chi(\bm{\kappa}_{})  = \chi( \bm{\kappa}^{}_{0})\prod_{m = 1}^{M_1}\chi( \bm{p}^{[m]}_{R}|\bm{\kappa}^{}_{0}) \chi( \bm{\Phi}^{[m]}_{R}|\bm{p}^{[m]}_{R},\bm{p}^{}_{U}),
\end{equation}
which is obtained by assuming that  the location parameters are independent across different RIS paths and the resulting FIM has the structure
\begin{equation}
\label{equ:prior_location_FIM_1}
\begin{aligned}
&\mathbf{J}_{\bm{\kappa}}=\left[\begin{array}{cccccccc}
\Tilde{\bm{\Xi}}_{\bm{\kappa}^{}_{0},\bm{\kappa}^{}_{0}} &\Tilde{\bm{\Xi}}_{\bm{\kappa}^{}_{0},\bm{\kappa}^{[1]}} & \cdots & \Tilde{\bm{\Xi}}_{\bm{\kappa}^{}_{0},\bm{\kappa}^{[M_1]}}  \\
\Tilde{\bm{\Xi}}_{\bm{\kappa}^{}_{0},\bm{\kappa}^{[1]}}^{\mathrm{T}} &  \Tilde{\bm{\Xi}}_{\bm{\kappa}^{[1]},\bm{\kappa}^{[1]}} & & 0  \\
\vdots & & \ddots & \\
\Tilde{\bm{\Xi}}_{\bm{\kappa}^{}_{0},\bm{\kappa}^{[M_1]}}^{\mathrm{T}} & 0 & &  \Tilde{\bm{\Xi}}_{\bm{\kappa}^{[M_1]},\bm{\kappa}^{[M_1]}}
\end{array}\right],
\end{aligned}
\end{equation}
where $\Tilde{\bm{\Xi}}_{\bm{\kappa}^{}_{0},\bm{\kappa}^{}_{0}} = \bm{\Xi}_{\bm{\kappa}^{}_{0},\bm{\kappa}^{}_{0}} + \sum_{m=1}^{M_1} \bigg[\bm{\Xi}_{\bm{\kappa}^{}_{0},\bm{\kappa}^{}_{0}}^{\bm{p}^{[m]}_R}+\bm{\Xi}_{\bm{\kappa}^{}_{0},\bm{\kappa}^{}_{0}}^{\bm{\Phi}^{[m]}_{R}} \bigg]$, $\Tilde{\bm{\Xi}}_{\bm{\kappa}^{}_{0},\bm{\kappa}^{[m]}} =\bm{\Xi}_{\bm{\kappa}^{}_{0},\bm{\kappa}^{[m]}}^{\bm{p}^{[m]}_R}+\bm{\Xi}_{\bm{\kappa}^{}_{0},\bm{\kappa}^{[m]}}^{\bm{\Phi}^{[m]}_{R}}$, $\bm{\Xi}_{\bm{\kappa}^{}_{0},\bm{\kappa}^{}_{0}} = \bm{F}_{{\bm{\kappa} }}( \bm{\kappa}^{}_{0} ;\bm{\kappa}^{}_{0},\bm{\kappa}^{}_{0})$, $\bm{\Xi}_{\bm{\kappa}^{}_{0},\bm{\kappa}^{}_{0}}^{\bm{p}^{[m]}} = \bm{F}_{{\bm{\kappa} }}(\bm{p}^{[m]}_{R}|\bm{\kappa}^{}_{0} ;\bm{\kappa}^{}_{0},\bm{\kappa}^{}_{0})$, $\bm{\Xi}_{\bm{\kappa}^{}_{0},\bm{\kappa}^{[m]}}^{p^{[m]}} = \bm{F}_{{\bm{\kappa} }}(\bm{p}^{[m]}_{R}|\bm{\kappa}^{}_{0} ;\bm{\kappa}^{}_{0},\bm{\kappa}^{[m]})$. Also, $\bm{\Xi}_{\bm{\kappa}^{}_{0},\bm{\kappa}^{}_{0}}^{\bm{\Phi}^{[m]}_{R}} = \bm{F}_{{\bm{\kappa} }}(\bm{\Phi}^{[m]}_{R}|\bm{p}^{[m]}_{R},\bm{p}^{}_{U} ;\bm{\kappa}^{}_{0},\bm{\kappa}^{}_{0})$ and $\bm{\Xi}_{\bm{\kappa}^{[m]},\bm{\kappa}^{[m]}}^{\bm{\Phi}^{[m]}_{R}} = \bm{F}_{{\bm{\kappa} }}(\bm{\Phi}^{[m]}_{R}|\bm{p}^{[m]}_{R},\bm{p}^{}_{U} ;\bm{\kappa}^{[m]},\bm{\kappa}^{[m]})$. Finally, we write  $\Tilde{\bm{\Xi}}_{\bm{\kappa}^{[m]},\bm{\kappa}^{[m]}} =\bm{\Xi}_{\bm{\kappa}^{[m]},\bm{\kappa}^{[m]}}^{p^{[m]}}+\bm{\Xi}_{\bm{\kappa}^{[m]},\bm{\kappa}^{[m]}}^{\bm{\Phi}^{[m]}_{R}}$. Subsequently, the general FIM, $\mathbf{J}_{\bm{y};\bm{\kappa}}$, is obtained from combining $\mathbf{J}_{\bm{y}|\bm{\kappa}}$ and $\mathbf{J}_{\bm{\kappa}}$ according to definition \ref{definition_FIM_1}.
\begin{lemma}
\label{lemma:prior_near_location_EFIM_1}
The general EFIM,  $\mathbf{J}_{\bm{y};\bm{\kappa}_1}^{\mathrm{e}}$, of the UE location parameters is presented in (\ref{equ:prior_location_EFIM_1}).
\begin{figure*}
\begin{align}
\begin{split}
\label{equ:prior_location_EFIM_1}
&\mathbf{J}_{\bm{y};\bm{\kappa}_1}^{\mathrm{e}}= \sum_{m=0}^{M_1}  \overline{\mathbf{\Upsilon}}_{\bm{\kappa}}^{[m]} \mathbf{J}_{\bm{y}|\bm{\eta}^{[m]}_1}^{\mathrm{e}} \overline{\mathbf{\Upsilon}}_{\bm{\kappa}}^{[m]\mathrm{T}} + \Tilde{\bm{\Xi}}_{\bm{\kappa}^{}_{0},\bm{\kappa}^{}_{0}} - \\
&\sum_{m = 1}^{M_1} \Bigg[ \Bigg[ \overline{\mathbf{\Upsilon}}_{\bm{\kappa}}^{[m]} \mathbf{J}_{\bm{y}|\bm{\eta}^{[m]}_1}^{\mathrm{e}} \overline{\overline{\mathbf{\Upsilon}}}_{\bm{\kappa}}^{[m]\mathrm{T}} + \Tilde{\bm{\Xi}}_{\bm{\kappa}^{}_{0},\bm{\kappa}^{[m]}} \Bigg] \left(\overline{\overline{\mathbf{\Upsilon}}}_{\bm{\kappa}}^{[m]} \mathbf{J}_{\bm{y}|\bm{\eta}^{[m]}_1}^{\mathrm{e}} \overline{\overline{\mathbf{\Upsilon}}}_{\bm{\kappa}}^{[m]\mathrm{T}} + \Tilde{\bm{\Xi}}_{\bm{\kappa}^{[m]},\bm{\kappa}^{[m]}}\right)^{-1}\Bigg[ \overline{\mathbf{\Upsilon}}_{\bm{\kappa}}^{[m]} \mathbf{J}_{\bm{y}|\bm{\eta}^{[m]}_1}^{\mathrm{e}} \overline{\overline{\mathbf{\Upsilon}}}_{\bm{\kappa}}^{[m]\mathrm{T}} + \Tilde{\bm{\Xi}}_{\bm{\kappa}^{}_{0},\bm{\kappa}^{[m]}} \Bigg]^{\mathrm{T}}\Bigg],
\end{split}
\end{align}
\end{figure*}
In that equation, $\overline{\mathbf{\Upsilon}}_{\bm{\kappa}}^{[0]}$ relates the LOS channel parameters to the UE location parameters, $\overline{\mathbf{\Upsilon}}_{\bm{\kappa}}^{[m]}$ relates the channel parameters of the $m^{\text{th}}$ RIS  to the UE location parameters, and $\overline{\overline{\mathbf{\Upsilon}}}_{\bm{\kappa}}^{[m]}$ relates the channel parameters of the $m^{\text{th}}$ RIS  to the location parameters of the $m^{\text{th}}$ RIS. The matrices $\overline{\mathbf{\Upsilon}}_{\bm{\kappa}}^{[0]}$, $\overline{\mathbf{\Upsilon}}_{\bm{\kappa}}^{[m]}$, and $\overline{\overline{\mathbf{\Upsilon}}}_{\bm{\kappa}}^{[m]}$ are submatrices in ${{\mathbf{\Upsilon}}}_{\bm{\kappa}}$.
\end{lemma}
\begin{proof}
The proof directly follows by applying Definition \ref{definition_EFIM} to  the general FIM,  $\mathbf{J}_{\bm{y};\bm{\kappa}}$.
\end{proof}
\subsection{Misalignment and Orientation Vectors of the RISs  as Estimation Parameters}
This scheme defines the location parameter related  to the $m^{\text{th}}$ RIS as  $\bm{\zeta}^{[m]} = \left[ \bm{\xi}^{[m]\mathrm{T}}_{R}, \bm{\Phi}^{[m]\mathrm{T}}_{R}  \right]^{\mathrm{T}}$. The UE location parameter, $\bm{\kappa}_0$, and the  location parameters of the RISs are specified in vector form as $\bm{\zeta}_{} \triangleq\left[\bm{\kappa}^{\mathrm{T}}_{0}, \bm{\zeta}^{[1]\mathrm{T}},\cdots,\bm{\zeta}^{[M_1]\mathrm{T}} \right]^{\mathrm{T}}.$
Specifying the parameter of interest as the UE location parameter, the parameter vector can be structured as $\bm{\zeta}_{} = [\bm{\zeta}_{1}^{\mathrm{T}}, \bm{\zeta}_{2}^{\mathrm{T}}]^{\mathrm{T}}$ where $\bm{\zeta}_{1} = \bm{\kappa}_{0}$ and $\bm{\zeta}_{2} = \left[ \bm{\zeta}^{[1]\mathrm{T}},\cdots,\bm{\zeta}^{[M_1]\mathrm{T}} \right]^{\mathrm{T}}$. Similar to the first scheme, the FIM from the observations is obtained by the bijective transformation, $\mathbf{J}_{\bm{y}|\bm{\zeta}} \triangleq \mathbf{\Upsilon}_{\bm{\zeta}} \mathbf{J}_{ \bm{\bm{y}}| \bm{\eta}_1}^{\mathrm{e}} \mathbf{\Upsilon}_{\bm{\zeta}}^{\mathrm{T}}$, where $\mathbf{\Upsilon}_{\bm{\zeta}}$ represents derivatives of the non-linear relationship between the geometric channel parameters, $ \bm{\eta}_1$, and the location parameters\cite{kay1993fundamentals}. In this scheme, {\em a priori} information is incorporated through a likelihood function similar to (\ref{equ:prior_location_likelihhood_1}) and the resulting FIM has identical structure to (\ref{equ:prior_location_FIM_1}), where the submatrices are defined similarly.

\begin{lemma}
\label{lemma:prior_near_location_EFIM_2}
The general EFIM,  $\mathbf{J}_{\bm{y};\bm{\zeta}_1}^{\mathrm{e}}$, of the UE location parameters is presented in (\ref{equ:prior_location_EFIM_2}).
\begin{figure*}
\begin{align}
\begin{split}
\label{equ:prior_location_EFIM_2}
&\mathbf{J}_{\bm{y};\bm{\zeta}_1}^{\mathrm{e}}= \sum_{m=0}^{M_1}  \overline{\mathbf{\Upsilon}}_{\bm{\zeta}}^{[m]} \mathbf{J}_{\bm{y}|\bm{\eta}^{[m]}_1}^{\mathrm{e}} \overline{\mathbf{\Upsilon}}_{\bm{\zeta}}^{[m]\mathrm{T}} + \Tilde{\bm{\Xi}}_{\bm{\zeta}^{}_{0},\bm{\zeta}^{}_{0}} - \\
&\sum_{m = 1}^{M_1} \Bigg[ \Bigg[ \overline{\mathbf{\Upsilon}}_{\bm{\zeta}}^{[m]} \mathbf{J}_{\bm{y}|\bm{\eta}^{[m]}_1}^{\mathrm{e}} \overline{\overline{\mathbf{\Upsilon}}}_{\bm{\zeta}}^{[m]\mathrm{T}} + \Tilde{\bm{\Xi}}_{\bm{\zeta}^{}_{0},\bm{\zeta}^{[m]}} \Bigg] \left(\overline{\overline{\mathbf{\Upsilon}}}_{\bm{\zeta}}^{[m]} \mathbf{J}_{\bm{y}|\bm{\eta}^{[m]}_1}^{\mathrm{e}} \overline{\overline{\mathbf{\Upsilon}}}_{\bm{\zeta}}^{[m]\mathrm{T}} + \Tilde{\bm{\Xi}}_{\bm{\zeta}^{[m]},\bm{\zeta}^{[m]}}\right)^{-1}\Bigg[ \overline{\mathbf{\Upsilon}}_{\bm{\zeta}}^{[m]} \mathbf{J}_{\bm{y}|\bm{\eta}^{[m]}_1}^{\mathrm{e}} \overline{\overline{\mathbf{\Upsilon}}}_{\bm{\zeta}}^{[m]\mathrm{T}} + \Tilde{\bm{\Xi}}_{\bm{\zeta}^{}_{0},\bm{\zeta}^{[m]}} \Bigg]^{\mathrm{T}}\Bigg].
\end{split}
\end{align}
\end{figure*}
In that equation, $\overline{\mathbf{\Upsilon}}_{\bm{\zeta}}^{[0]}$ relates the LOS channel parameters to the UE location parameters, $\overline{\mathbf{\Upsilon}}_{\bm{\zeta}}^{[m]}$ relates the channel parameters of the $m^{\text{th}}$ RIS  to the UE location parameters, and $\overline{\overline{\mathbf{\Upsilon}}}_{\bm{\zeta}}^{[m]}$ relates the channel parameters of the $m^{\text{th}}$ RIS  to the location parameters of the $m^{\text{th}}$ RIS. The matrices $\overline{\mathbf{\Upsilon}}_{\bm{\zeta}}^{[0]}$, $\overline{\mathbf{\Upsilon}}_{\bm{\zeta}}^{[m]}$, and $\overline{\overline{\mathbf{\Upsilon}}}_{\bm{\zeta}}^{[m]}$ are submatrices in ${{\mathbf{\Upsilon}}}_{\bm{\zeta}}$.
\end{lemma}
\begin{proof}
The proof directly follows by applying Definition \ref{definition_EFIM} to  the general FIM,  $\mathbf{J}_{\bm{y};\bm{\zeta}}$.
\end{proof}

Lemma \ref{lemma:prior_near_location_EFIM_1} and Lemma \ref{lemma:prior_near_location_EFIM_2} apply to both near and far field propagation regimes. The entries in the matrix in these lemmas related to the $m^{\text{th}}$ RIS depend on whether the UE is experiencing near or far field propagation with respect to that particular RIS. All the UEs could be in the near-field of all the RISs, and all the UEs could be in the far-field of all the RISs. Finally, there can be some UEs in the near-field of some RISs, and there can be some UEs in the far-field of some RISs.
\color{black}
\begin{theorem}
The differentiating factor of $\mathbf{J}_{\bm{y};\bm{\kappa}_1}^{\mathrm{e}}$ from $\mathbf{J}_{\bm{y};\bm{\zeta}_1}^{\mathrm{e}}$ is the {\em a priori} information available about both set of parameters. If the off diagonals in both {\em a priori} information matrices are zero matrices then $\mathbf{J}_{\bm{y};\bm{\kappa}_1}^{\mathrm{e}} \succeq \mathbf{J}_{\bm{y};\bm{\zeta}_1}^{\mathrm{e}}$ iff $\mathbf{J}_{\bm{\kappa}} \succeq \mathbf{J}_{\bm{\zeta}}$.
\end{theorem}
\begin{proof}

Due to the structure of $\mathbf{J}_{\bm{y};\bm{\kappa}_1}^{\mathrm{e}}$ and $\mathbf{J}_{\bm{y};\bm{\zeta}_1}^{\mathrm{e}}$, when the off diagonals are zero, we only have to prove that $\mathbf{\Upsilon}_{\bm{\kappa}} =\mathbf{\Upsilon}_{\bm{\zeta}}$. This can be seen by noting that when $\mathbf{\Upsilon}_{\bm{\kappa}} =\mathbf{\Upsilon}_{\bm{\zeta}}$, the values of  $\mathbf{J}_{\bm{y};\bm{\kappa}_1}^{\mathrm{e}}$ and $\mathbf{J}_{\bm{y};\bm{\zeta}_1}^{\mathrm{e}}$ depends on the diagonal terms in,  $\mathbf{J}_{\bm{\kappa}}$ and $ \mathbf{J}_{\bm{\zeta}}$, respectively.  The diagonal entries in $\mathbf{J}_{\bm{\kappa}}$ and $ \mathbf{J}_{\bm{\zeta}}$ are $\Tilde{\bm{\Xi}}_{\bm{\kappa}^{[m]},\bm{\kappa}^{[m]}}$ and $\Tilde{\bm{\Xi}}_{\bm{\zeta}^{[m]},\bm{\zeta}^{[m]}}$, respectively. 
The proof for $\mathbf{\Upsilon}_{\bm{\kappa}} =\mathbf{\Upsilon}_{\bm{\zeta}}$ is found by observing the derivatives in Appendix F.\color{black}
\end{proof}
\begin{remark}
\label{remark_no_priori_RIS}
From Lemmas \ref{lemma:prior_near_location_EFIM_1} and \ref{lemma:prior_near_location_EFIM_2}, as shown in Appendix \ref{appendix_no_priori_RIS}, irrespective of the RIS size,  near-field or far-field propagation, an RIS is only useful for localization when there is {\em a priori} information about the RIS location. 
\end{remark}

The position error bound (PEB) and orientation error bound (OEB) are obtained by inverting either $\mathbf{J}_{\bm{y};\bm{\zeta}_1}^{\mathrm{e}}$  and $\mathbf{J}_{\bm{y};\bm{\zeta}_1}^{\mathrm{e}}$, then summing the appropriate diagonals. For the PEB of a UE, the first three diagonals in either  $(\mathbf{J}_{\bm{y};\bm{\zeta}_1}^{\mathrm{e}})^{-1}$  and $(\mathbf{J}_{\bm{y};\bm{\zeta}_1}^{\mathrm{e}})^{-1}$ are summed, while for the OEB of a UE, the fourth, fifth, and sixth diagonals in either  $(\mathbf{J}_{\bm{y};\bm{\zeta}_1}^{\mathrm{e}})^{-1}$  and $(\mathbf{J}_{\bm{y};\bm{\zeta}_1}^{\mathrm{e}})^{-1}$ are summed.
\color{black}


\section{Numerical Results}
This section uses Monte-Carlo simulations 
 to verify  analytical results and obtain useful system design insights. The system setup includes a single perfectly located BS whose centroid is also the global origin of the coordinate system i.e., ${\bm{p}}_{{B}}= [0, 0, 0]^{\mathrm{T}} = \Tilde{\bm{p}}_{{B}}$ and $\bm{Q}_B = \bm{I}$. All position vectors are in meters, and all orientation vectors are in radians. The system setup includes an LOS path and two RISs that reflect the transmitted signal to the UE. While the LOS pathloss is described by the free-space propagation model, the pathloss of the RIS paths is described by Lemma \ref{lemma:pathloss} with the gain controlling factor, $q_0 = 0.285$ and the element efficiency, $\epsilon_p = 0.5$ \cite{ellingson2021path}. The wavelength is $\lambda = 3 \text{ cm}$ with the elements in each RIS spaced by $1.5 \text{ cm}$. There are $N = 256$ subcarriers, a single antenna at the BS and no transmit beamforming, the transmit power is $ 23 \text{ dBm}$, the noise power spectral density (PSD) is $N_0 = -174 \text{ dBm/Hz}$, and the transmit and receive gains antenna are set to $G_B = G_U = 2 \text{ dB}$, respectively. One RIS is misoriented  but perfectly aligned with its centroid located at $\bm{p}^{[1]}_{R} = \Tilde{\bm{p}}_{{R}}^{[1]} = [10,8,4]^{\mathrm{T}}$ with the following rotation angles $\bm{\Phi}_{R}^{[1]} = [0.1,0.2, 0.1]^{\mathrm{T}}$, while the other RIS is perfectly located at $\bm{p}^{[2]}_{R} = \Tilde{\bm{p}}_{{R}}^{[2]} = [10,8.5,4]^{\mathrm{T}}$. The UE location is described with $\bm{p}_{U} = \Tilde{\bm{p}}_{{U}} = [12,10,3]^{\mathrm{T}}$ and $\bm{Q}_U = \bm{I}$. We focus on the case with uniform rectangular arrays (URAs) at the BS, RIS, and the UE with their respective normal vectors originally pointing in the $z$ direction. At the considered UE position and with the Fraunhofer distance defined in \cite{9335528}, the UE is in the near-field of the first RIS when $N_R^{[1]} \geq 100$ and in the near-field of the second RIS when $N_R^{[2]} \geq 90$.  Five cases of parameterization of unknown parameters are investigated: a) a case showing a simplified and unrealistic parameterization that serves to depict the relationship between the RIS orientation offset and the complex path gain, $\bm{\eta}^{[1]} \triangleq\left[\bm{\Phi}_{R}^{[1]\mathrm{T}}, {\beta}_{\mathrm{R}}^{[1]}, {\beta}_{\mathrm{I}}^{[1]}  \right]^{\mathrm{T}}$, b) a more realistic case that includes the unknown angles in the RIS-UE link and the complex path gain, $\bm{\eta}^{[1]} \triangleq\left[\bm{\Phi}_{R}^{[1]\mathrm{T}},{\theta}_{{\bm{r}_{R}\bm{u}_{U}}}^{[1]}, {\phi}_{{\bm{r}_{R}\bm{u}_{U}}}^{[1]}, {\beta}_{\mathrm{R}}^{[1]}, {\beta}_{\mathrm{I}}^{[1]} \right]^{\mathrm{T}}$, c) a third case with $\bm{\eta}^{[1]} \triangleq\left[{\tau}_{{\bm{r}_{R}\bm{u}_{U}}}^{[1]},{\theta}_{{\bm{r}_{R}\bm{u}_{U}}}^{[1]}, {\phi}_{{\bm{r}_{R}\bm{u}_{U}}}^{[1]},
\bm{\Phi}_{R}^{[1]\mathrm{T}}, {\beta}_{\mathrm{R}}^{[1]}, {\beta}_{\mathrm{I}}^{[1]} \right]^{\mathrm{T}}$, d) a fourth case with $\bm{\eta}^{[2]} \triangleq\left[{\tau}_{{\bm{r}_{R}\bm{u}_{U}}}^{[2]},{\theta}_{{\bm{r}_{R}\bm{u}_{U}}}^{[2]}, {\phi}_{{\bm{r}_{R}\bm{u}_{U}}}^{[2]}, {\beta}_{\mathrm{R}}^{[2]}, {\beta}_{\mathrm{I}}^{[2]} \right]^{\mathrm{T}}$, and e) a fifth case considering both $\bm{\eta}^{[1]} \triangleq\left[{\tau}_{{\bm{r}_{R}\bm{u}_{U}}}^{[1]},{\theta}_{{\bm{r}_{R}\bm{u}_{U}}}^{[1]}, {\phi}_{{\bm{r}_{R}\bm{u}_{U}}}^{[1]},
\bm{\Phi}_{R}^{[1]\mathrm{T}}, {\beta}_{\mathrm{R}}^{[1]}, {\beta}_{\mathrm{I}}^{[1]} \right]^{\mathrm{T}}$ and $\bm{\eta}^{[2]} \triangleq\left[{\tau}_{{\bm{r}_{R}\bm{u}_{U}}}^{[2]},{\theta}_{{\bm{r}_{R}\bm{u}_{U}}}^{[2]}, {\phi}_{{\bm{r}_{R}\bm{u}_{U}}}^{[2]}, {\beta}_{\mathrm{R}}^{[2]}, {\beta}_{\mathrm{I}}^{[2]} \right]^{\mathrm{T}}$. For the first two cases of parameterization of the unknown parameters, the location algorithm at the UE attempts to use either the {\em correct} near-field model (\ref{equ:receive_processing_2}) or the {\em incorrect} far-field model (Proposition \ref{proposition:far_field_1}) to correct the orientation offset of the first RIS. In case (c), the UE location algorithm attempts to use only the reflected signals from the misoriented RIS to position the UE; case (d) considers the situation where the UE location algorithm attempts to use the reflected signal from the perfectly located RIS to position the UE. Finally, case (e) considers the situation where signals reflected from both RISs are used to position the UE. Hence, for the first two cases, the OEB of the orientation offset of the first RIS is plotted, while the PEB of the UE is investigated for the other three cases. In all applicable plots, the prefix ``FF'' is used to distinguish the {\em incorrect} case where a far-field model is applied to this near-field simulation setup from the {\em correct} case where the near-field model is used for the near-field setup. There is no {\em a priori} information about the complex path gains, and the {\em a priori} information about the orientation offset  is quantified as a fraction of the SNR, $P/\sigma^2$. 
The PEB and the OEB could serve as benchmarks for future algorithms.


\subsection{Effect of Number of Receive Antennas on the Estimation of RIS Orientation Offset}
We present metrics to investigate the RIS orientation offset correction for the first two parameterization cases  under the near-field and far-field models with varying {\em a priori} information and a varying number of receive antennas. In Fig. \ref{fig:Results/phi_r_oeb_vs_NU}, when the near-field model is used, the OEB is shown to decrease with an increase in the number of receive antennas; however, when the far-field model is used, the OEB stays relatively constant for a varying number of receive antennas. 
\begin{figure}[htb!]
\centering
\subfloat[]{\includegraphics[width=\linewidth]{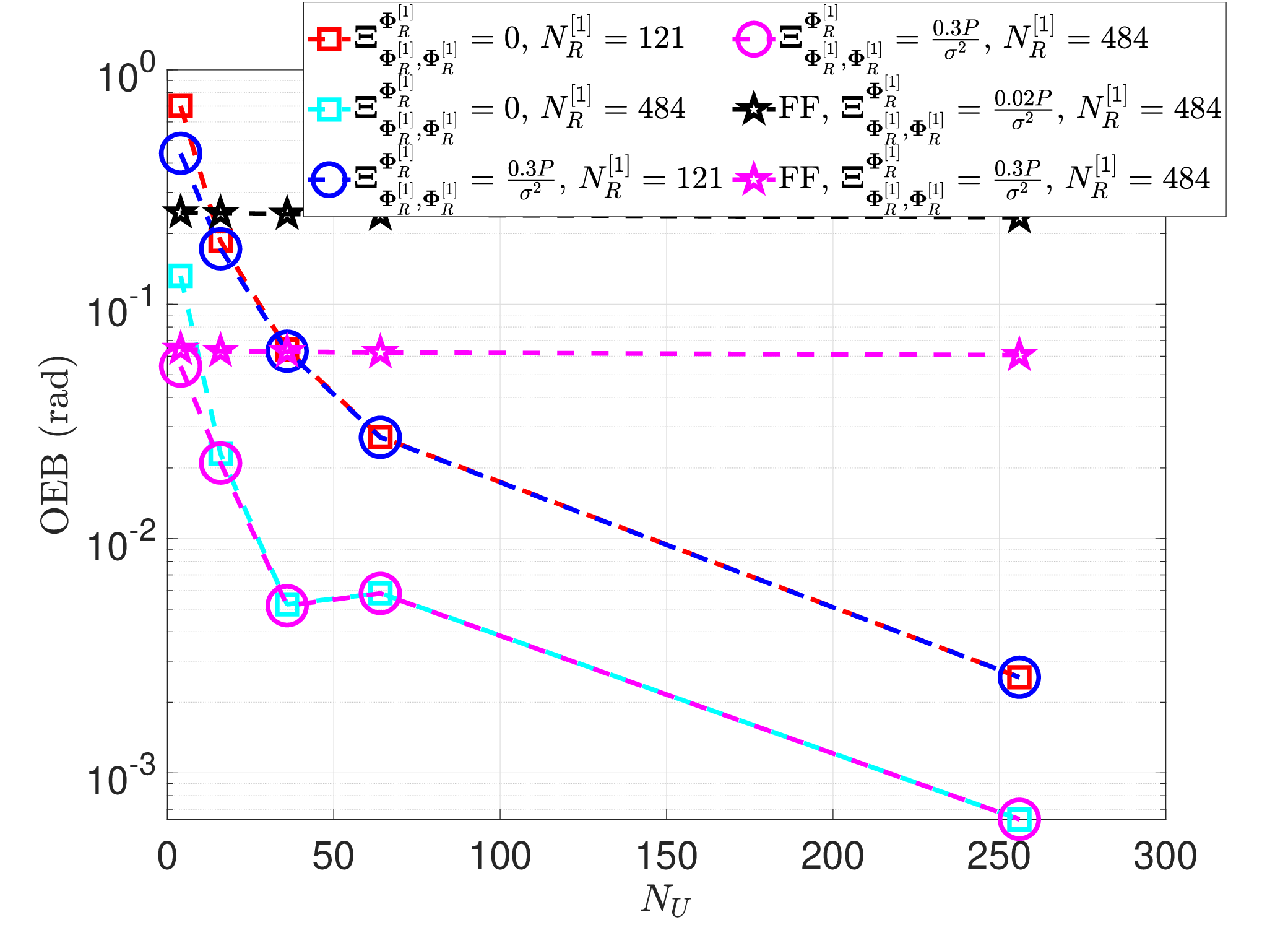}
\label{fig:Results/phi_r_oeb_vs_NU}}
\hfil
\subfloat[]{\includegraphics[width=\linewidth]{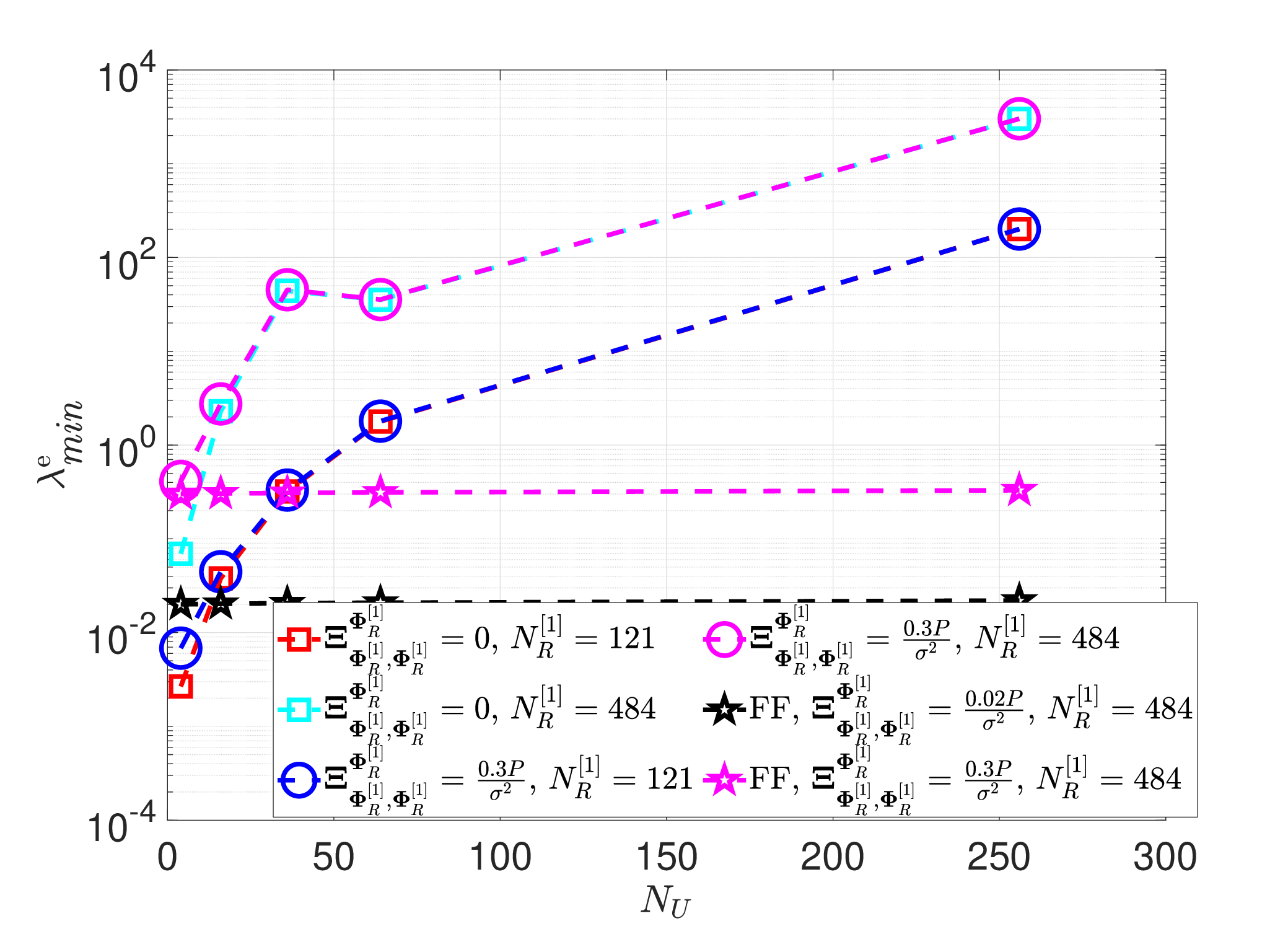}
\label{fig:Results/phi_r_lambda_vs_NU}}
\caption{(a) Case (a) - this investigates the direct relationship between the RIS orientation offset and the complex path gain by showing the OEB of the orientation offset of the first RIS vs. the number of receive antennas 
(b)  Case (a) - this investigates the direct relationship between the RIS orientation offset and the complex path gain by presenting 
$\lambda^{\mathrm{e}}_{\textit{min}}$ vs. $N_U$ 
}
\label{Results:partial_and_full_phi_r_oeb_vs_NU}
\end{figure}
\begin{figure}[htb!]
\centering
\subfloat[]{\includegraphics[width=\linewidth]{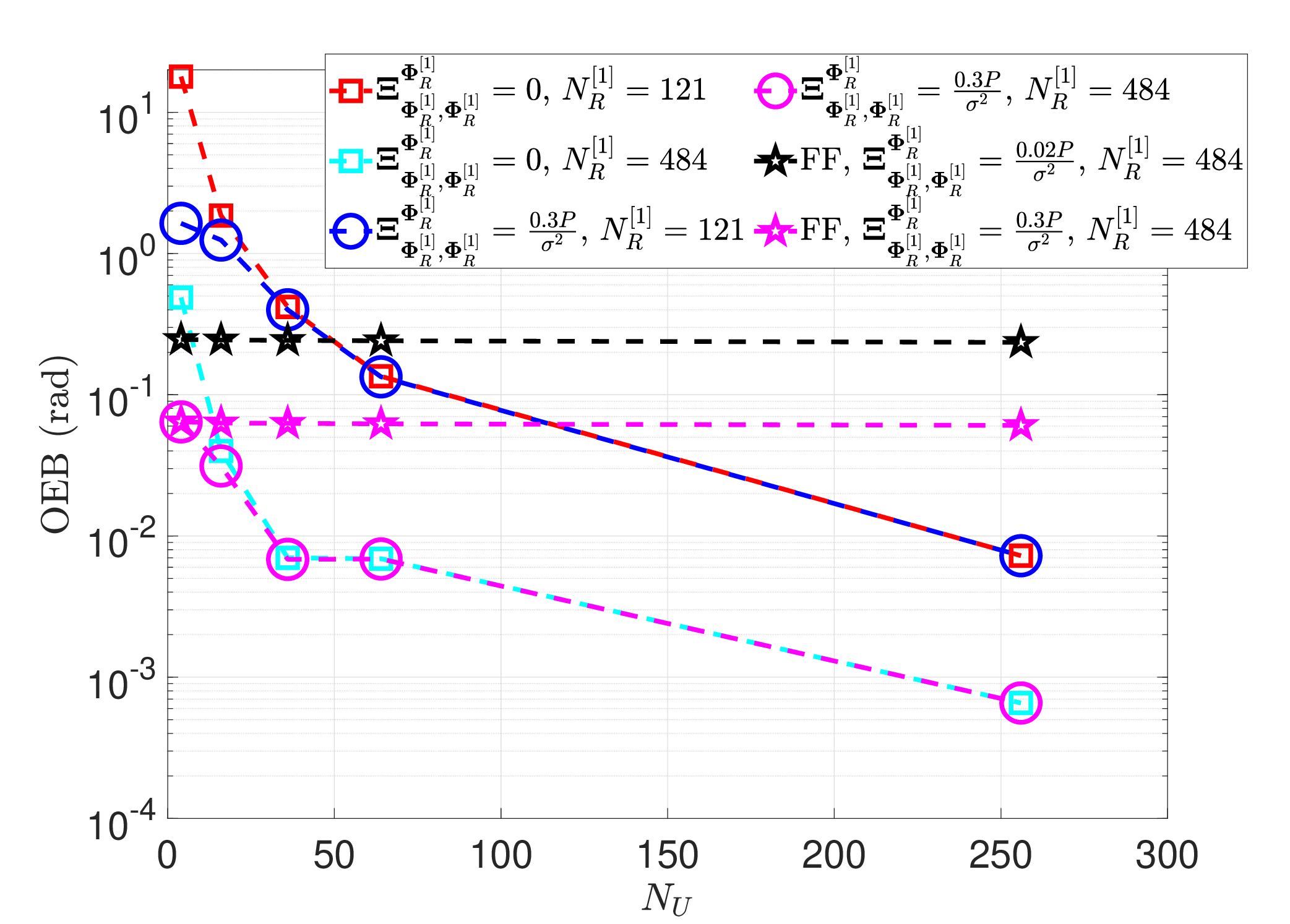}
\label{fig:Results/full_phi_r_oeb_vs_NU}}
\hfil
\subfloat[]{\includegraphics[width=\linewidth]{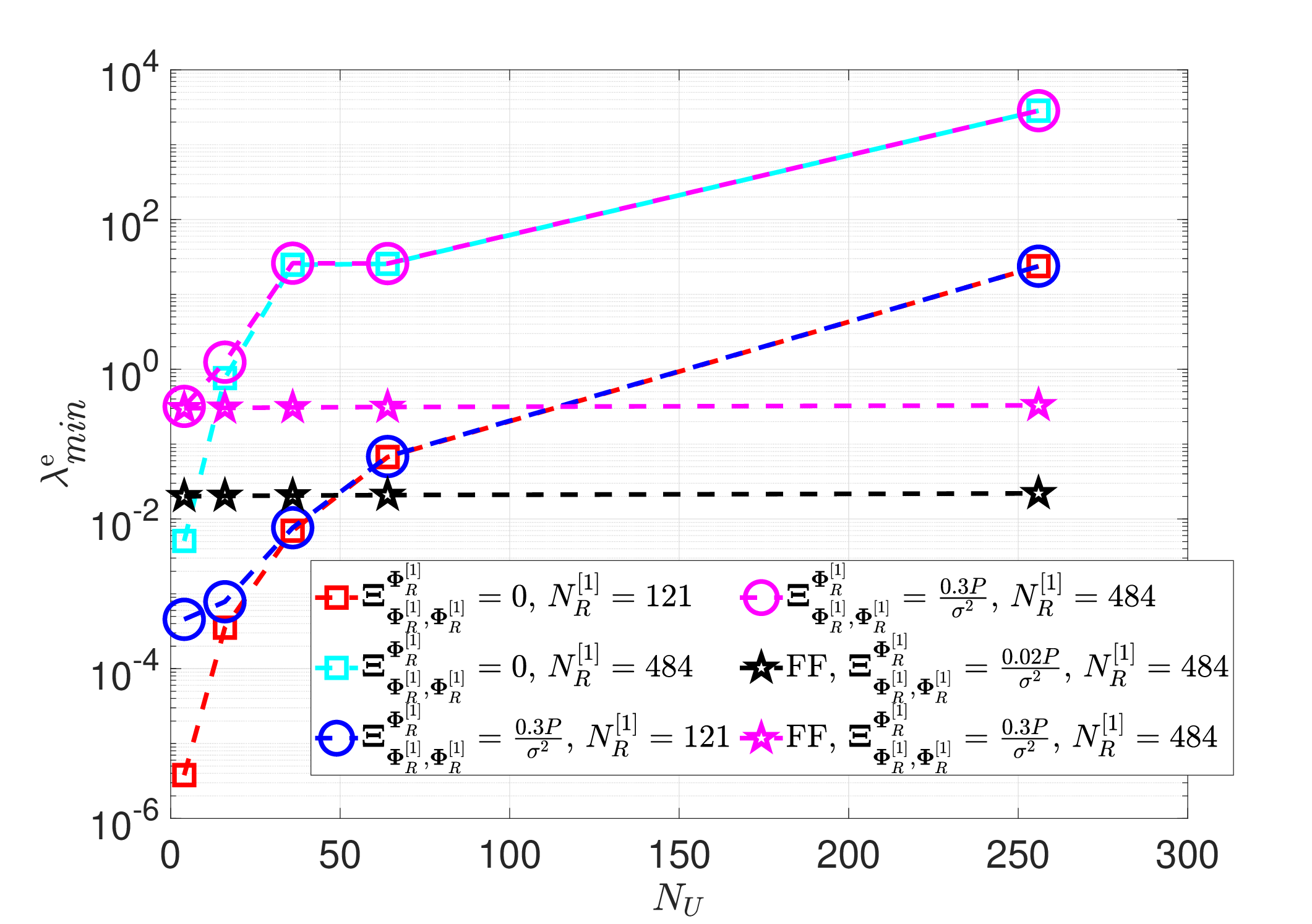}
\label{fig:Results/full_phi_r_lambda_vs_NU}}
\caption{(a) Case (b) investigates the direct relationship between the orientation offsets and other angle parameters by showing the OEB of the orientation offset of the first RIS vs. the number of receive antennas 
(b) Case (b) investigates the direct relationship between the orientation offsets and other angle parameters by showing $\lambda^{\mathrm{e}}_{\textit{min}}$ vs. $N_U$.}
\label{Results:full_phi_r_lambda_vs_NU}
\end{figure}
\begin{figure}[htb!]
\centering
\subfloat[]{\includegraphics[width=\linewidth]{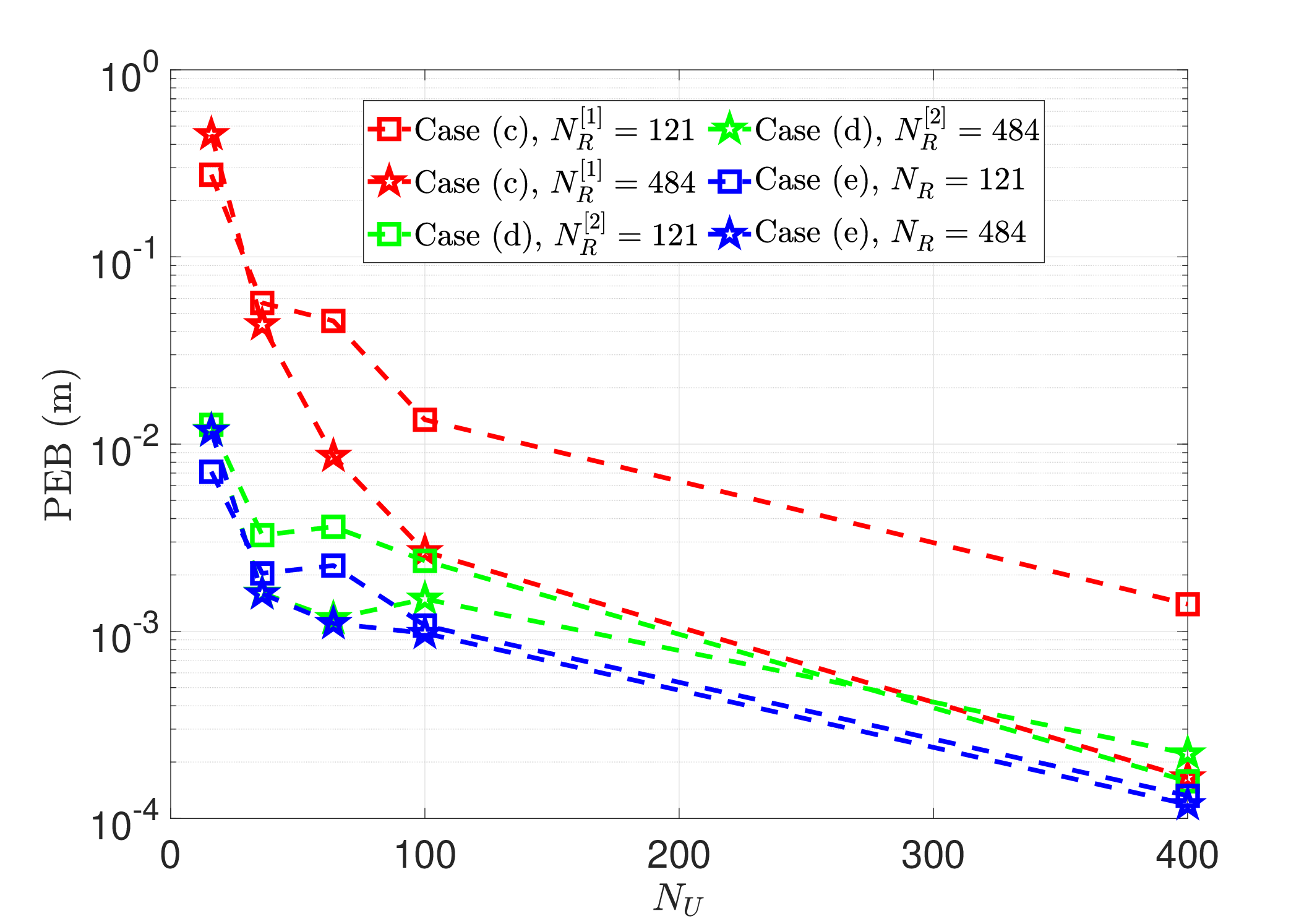}
\label{fig:Results/PEB_vs_NU}}
\hfil
\subfloat[]{\includegraphics[width=\linewidth]{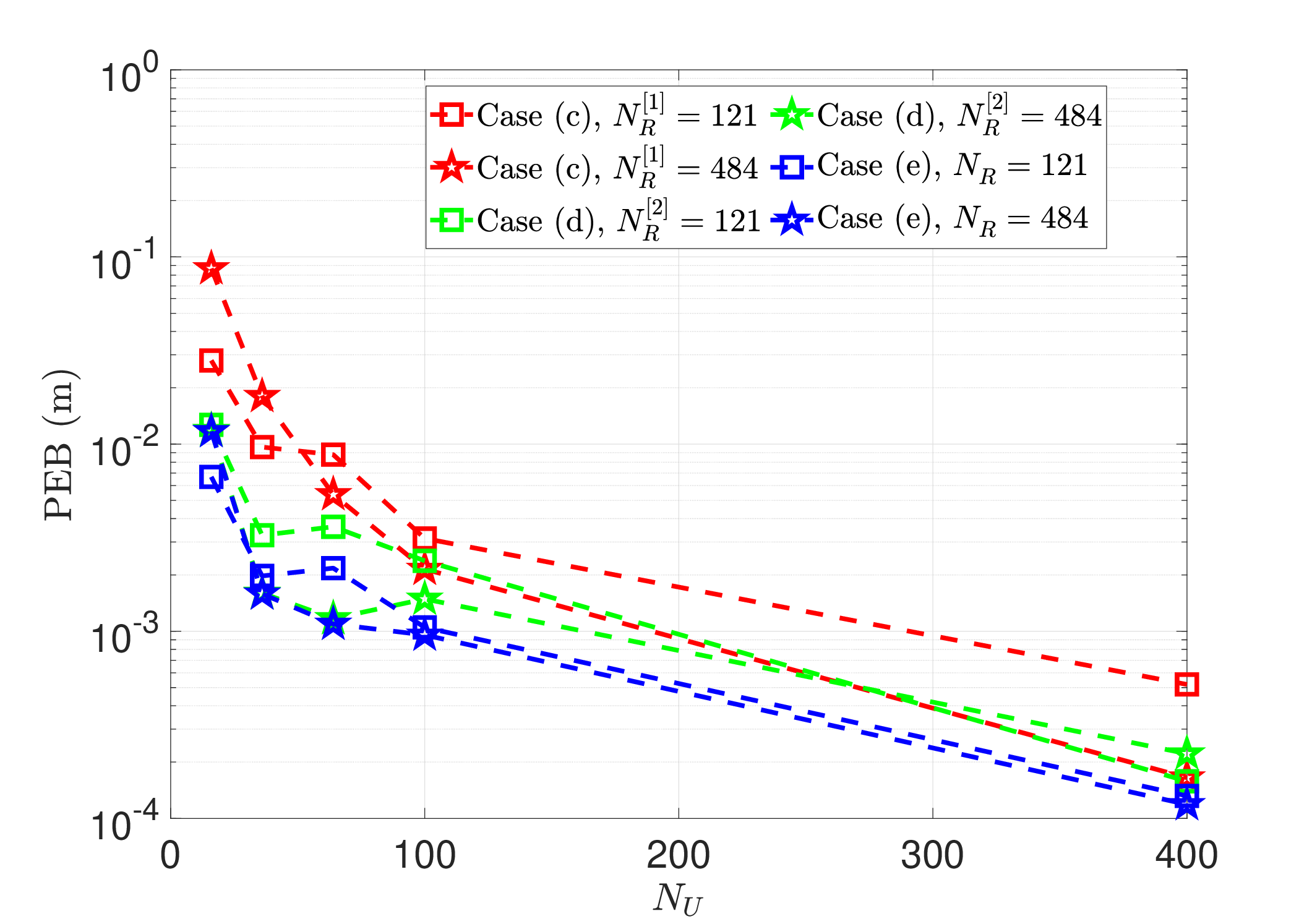}
\label{fig:Results/PEB_vs_NU_1}}
\caption{  PEB of UE vs. $N_U$: we consider the following cases; case (c) (red curves), 
Case (d) (green curves), 
and case (e) (blue curves) (a)  $\bm{\Xi}_{\bm{\Phi}_{R}^{[1]},\bm{\Phi}_{R}^{[1]}}^{\bm{\Phi}_{R}^{[1]}} = 0$
 and (b) $\bm{\Xi}_{\bm{\Phi}_{R}^{[1]},\bm{\Phi}_{R}^{[1]}}^{\bm{\Phi}_{R}^{[1]}} = 0.5 \frac{P}{\sigma^2}$.}
\label{Results:PEB_vs_NU_11}
\end{figure}

This trend is also observed in Fig. \ref{fig:Results/full_phi_r_oeb_vs_NU} and this serves as a sanity check for Theorem \ref{theorem:channel_parameter_estimatable}. 
While using the near-field model, the OEB is prohibitively large when $N_U$ is small, especially for a small RIS. 
More specifically, when $N_U$ is small, and there is little or no {\em a priori} information, the resulting EFIM, $\mathbf{J}_{\bm{y};\bm{\eta}^{[1]}_1}^{\mathrm{e}}$, is sometimes almost singular; hence $\mathbf{J}_{\bm{y};\bm{\eta}^{[1]}_1}^{\mathrm{e}}$ has a relatively small eigenvalue.  This is confirmed in Figs. \ref{fig:Results/phi_r_lambda_vs_NU} and \ref{fig:Results/full_phi_r_lambda_vs_NU} by presenting the smallest eigenvalue, $\lambda^{\mathrm{e}}_{\textit{min}}$, as a function of $N_U$.   \color{black}
\begin{figure}[htb!]
\centering
\subfloat[]{\includegraphics[width=\linewidth]{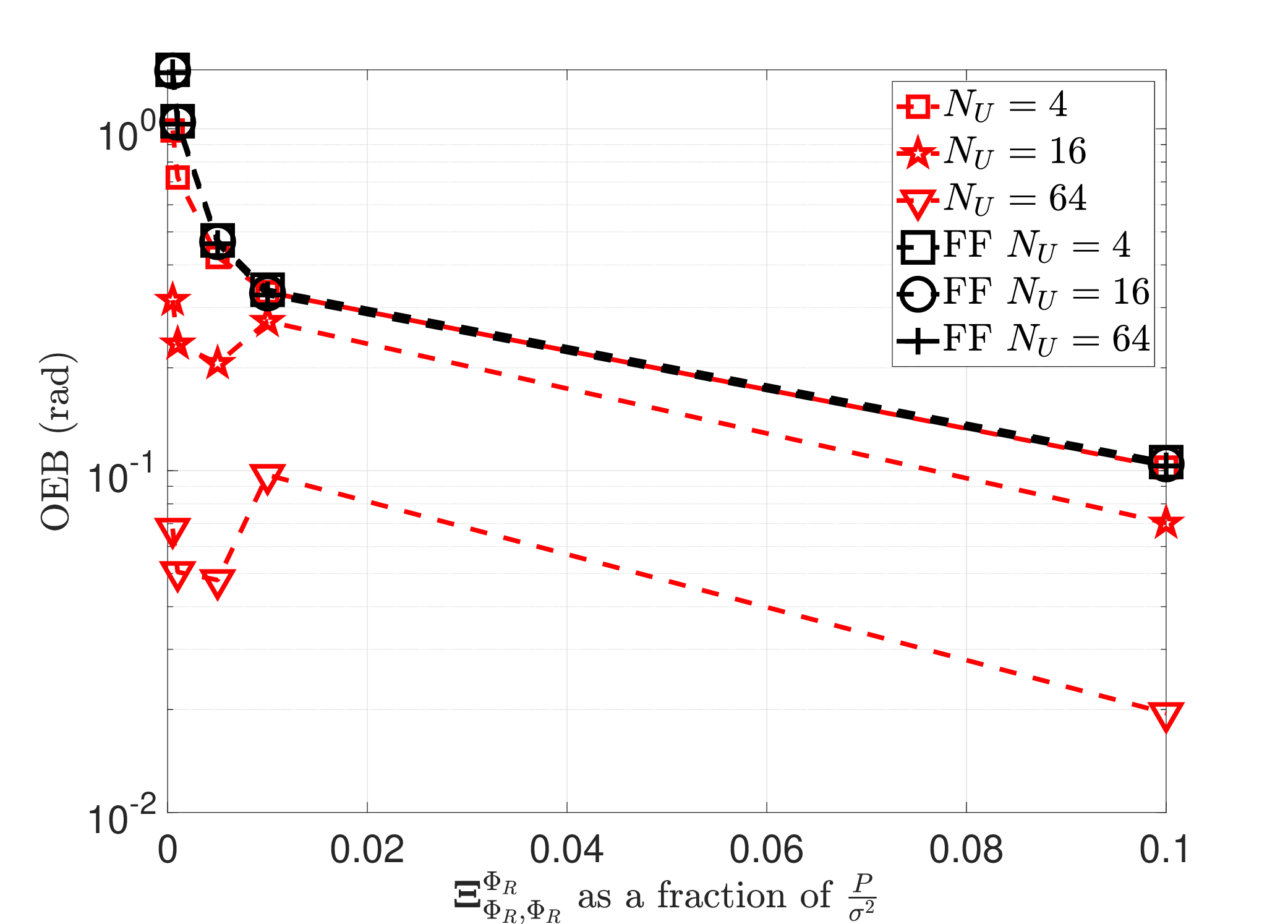}
\label{fig:Results/phi_r_oeb_vs_prior}}
\hfil
\subfloat[]{\includegraphics[width=\linewidth]{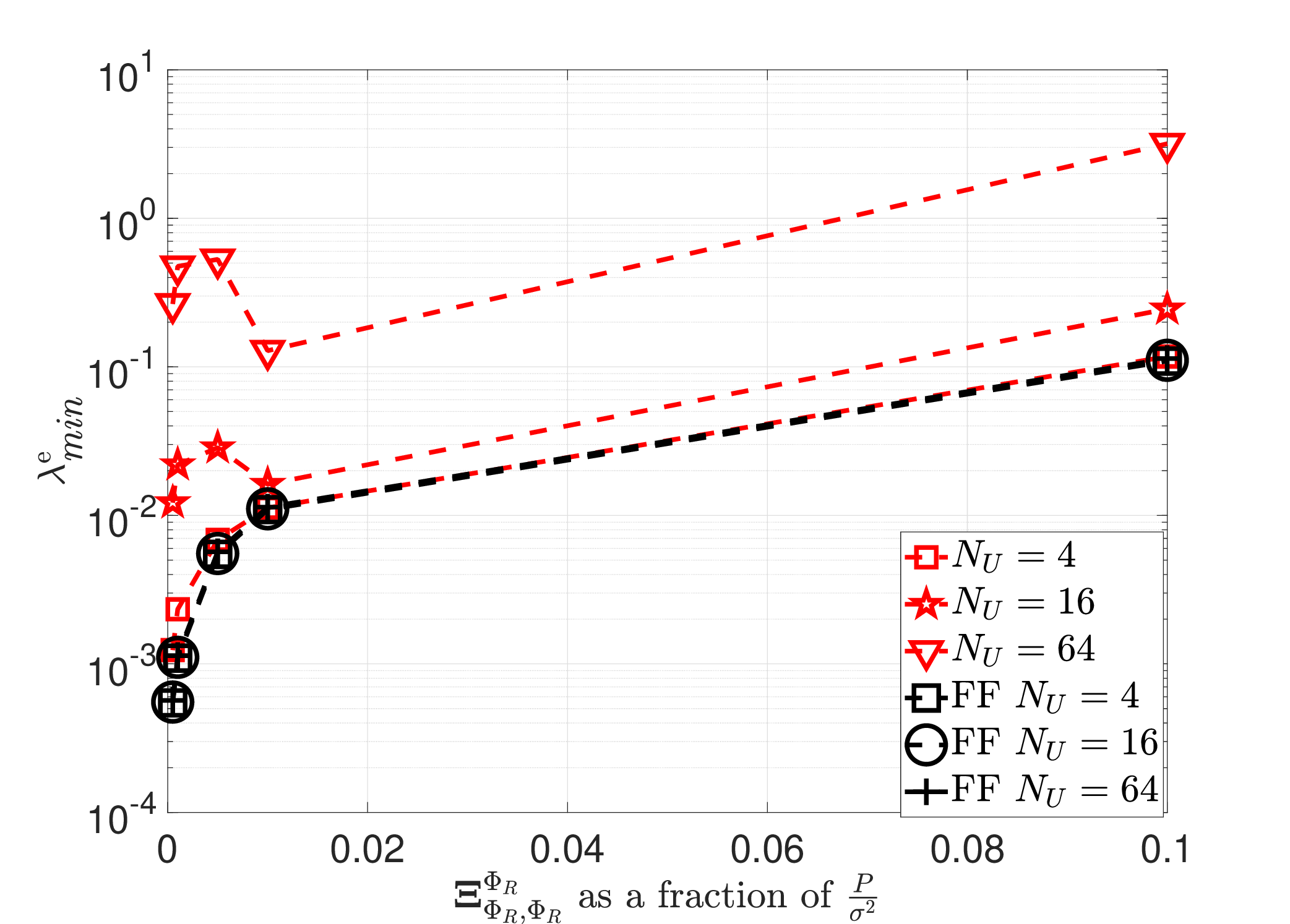}
\label{fig:Results/phi_r_lambda_vs_prior}}
\caption{(a) Case (a) - this investigates the direct relationship between the RIS orientation offset and the complex path gain by showing the OEB of the orientation offset of the first RIS vs. the available {\em a priori} information about the orientation offset 
(b) Case (a): the smallest eigenvalue (normalized by the SNR) of $\mathbf{J}_{\bm{y};\bm{\eta}^{[1]}_1}^{\mathrm{e}}$ vs. the available {\em a priori} information about the orientation offset 
.}
\label{Results:phi_r_oeb_lambda_vs_prior}
\end{figure}
In Fig. \ref{fig:Results/phi_r_lambda_vs_NU}, while using the near-field model, the smallest eigenvalue increases significantly with an increase in $N_U$, when using the far-field model, the smallest eigenvalue stays relatively constant irrespective of $N_U$. This observation validates Lemmas \ref{lemma:farfield_channel_parameter_estimatable} and \ref{lemma:near-field_channel_parameter_estimatable}. More specifically, it validates that in the far-field, the OEB is independent of $N_U$ and only depends on the {\em a priori} information.

\subsection{Effect of Number of Receive Antennas on the Positioning Error of the UE}
We investigate the UE's PEB under near-field propagation as a function of a varying number of RIS elements and a function of a varying number of receive antennas.  The PEB is obtained by inverting $\mathbf{J}_{\bm{y};\bm{\kappa}_1}^{\mathrm{e}}$ in Lemma \ref{lemma:prior_near_location_EFIM_1}, and the resulting PEB is presented in Fig. \ref{fig:Results/PEB_vs_NU} assuming there is no {\em a priori} information and in Fig \ref{fig:Results/PEB_vs_NU_1} assuming $\bm{\Xi}_{\bm{\Phi}_{R}^{[1]},\bm{\Phi}_{R}^{[1]}}^{\bm{\Phi}_{R}^{[1]}} = 0.5 \frac{P}{\sigma^2}$. \color{black} The LOS is ignored to focus on the impact of misorientation on the PEB, and positioning the UE with the RISs is considered under three different cases of parameterization. In case (c), only the signal received at the UE through the misoriented RIS is used to position the UE; in case (d), only the signal received at the UE through the perfectly located RIS is used to position the UE. In case (e), the signals received at the UE from both RISs (misoriented and perfectly located) are used to position the UE. For case (c) and at small values of $N_U$, the PEB is significantly affected by the misorientation of the RIS. However, as $N_U$ increases, the effect of misorientation is reduced, especially when the misoriented RIS is large, i.e., $N_R^{[1]} = 484$. Finally, in Fig. \ref{fig:Results/PEB_vs_NU}, we observe that the PEB is dominated by the PEB of the perfectly located RIS.

\subsection{Effect of {\em a priori} Information about the RIS Orientation Offset}
We now examine RIS offset correction for the first two parameterization cases under the conditions of increasing {\em a priori} information about the RIS offset, $\bm{\Xi}_{\bm{\Phi}_{R}^{[1]},\bm{\Phi}_{R}^{[1]}}^{\bm{\Phi}_{R}^{[1]}}$, and with no {\em a priori} information about the complex path gains, $\mathbf{\Xi}_{{\beta}_{R}^{[1]},{\beta}_{R}^{[1]}}^{{\beta}_{R}^{[1]}}$. The near-field and far-field models are investigated with different numbers of receive antennas. In Fig. \ref{fig:Results/phi_r_oeb_vs_prior}, as intuition would suggest, the OEB decreases significantly with an increase in $\bm{\Xi}_{\bm{\Phi}_{R}^{[1]},\bm{\Phi}_{R}^{[1]}}^{\bm{\Phi}_{R}^{[1]}}$. When using the near-field model, $N_U$ significantly affects the OEB, this is not the case when the far-field model is used. 
The disparity in the OEB behavior between the near-field and the far-field models is also evident in Fig.  \ref{fig:Results/phi_r_lambda_vs_prior}; the variations in the smallest eigenvalue of the $\mathbf{J}_{\bm{y};\bm{\eta}^{[1]}_1}^{\mathrm{e}}$ mirror the variations in OEB under similar conditions. 
\section{Conclusion}
In this work, we developed a Bayesian framework for localization that accounts for uncertainty in the RIS location. Specifically, we considered the downlink of a MIMO system operating with OFDM with RISs providing reflected signals to a UE in addition to an LOS path under both near-field and far-field propagation regimes. We showed through the EFIM for channel parameters that an unknown phase offset in the received signal at the UE makes it impossible to estimate and correct any RIS orientation offset through these received signals in the far-field. However, in the near-field, an unknown phase offset in the received signal at the UE does not hinder the estimation of the RIS orientation offset when there is more than one receive antenna. This non-trivial result indicates that the possibility of correcting an RIS orientation offset through the received signals at the UE only exists while the UE is in the near-field. Furthermore, through the EFIM, we showed that regardless of RIS size and propagation regime, the RISs are only helpful for localization if there is {\em a priori} information about their location. Finally, through numerical analysis of the EFIM, we demonstrated the loss in information when the far-field model is {\em incorrectly} applied to the signals received at a UE which is operating in the near-field propagation regime. A rigorous extension of this work involves investigating the degradation in UE localization performance when the localization algorithm has no idea about the misorientation and misalignment of the RISs. Such an extension would involve the computation of the Kullback divergence and a misspecified Cramer Rao bound.

\appendix 
	 \label{appendix_distance_approximation:far_field_1} 
	\subsection{Far-Field Distance Approximation}
	The distance from the $g^{\text{th}}$ element on the $G^{\text{th}}$ entity located at $\bm{p}_{G}$ to the $v^{\text{th}}$ element on the $V^{\text{th}}$ entity located at $\bm{p}_{V}$ can be written as $$d_{\bm{p}_{g} \bm{p}_{v}} = \norm{{\bm{p}}_{v} - {\bm{p}}_{g}   } =  \norm{d_{\bm{p}_{G}\bm{p}_V} \Delta_{\bm{p}_{G}\bm{p}_V} + ( {\bm{s}}_{v}- {\bm{s}}_{g})}.$$
Using the first order binomial expansion of $\sqrt{1 + x}$ and ignoring the term $\norm{{\bm{s}}_{v}- \bm{s}_g}^2/d_{\bm{p}_{G}\bm{p}_V}^2$, we have  $$d_{\bm{p}_{g} \bm{p}_{v}} = d_{\bm{p}_{G}\bm{p}_V} + \Delta_{\bm{p}_{G}\bm{p}_V}^{\mathrm{T}} ({\bm{s}}_{v}- {\bm{s}}_{g}),$$
the consequence of the corresponding delay, $\tau_{\bm{p}_{g} \bm{p}_{v}}$, on the signal can be expressed as $$e^{-j 2 \pi f_{n} \tau_{\bm{p}_{g}\bm{p}_v}}= e^{-j 2 \pi f_{n}
\tau_{\bm{p}_{G}\bm{p}_V}} e^{-j 2 \frac{\pi}{\lambda_n}  \Delta_{\bm{p}_{G}\bm{p}_V}^{\mathrm{T}} ({\bm{s}}_{v}- {\bm{s}}_{g})}.$$
\subsection{Proof of Proposition \ref{proposition:EFIM_channel}}
\label{appendix:proof_PD}
In general, an  ${N_{\eta}} \times {N_{\eta}}$ symmetric matrix, $\mathbf{J}$ is positive definite, if $\bm{v}^{\mathrm{T}}\mathbf{J}\bm{v} > 0$, for all non-zero $\bm{v} \in \mathbb{C}^{{N_{\eta}} \times 1}$. Now, suppose a zero lies on the diagonal of $\mathbf{J}$ such that $[\mathbf{J}]_{[n,n]} = 0$ and suppose we select a vector $\bm{v}$ with all zero entries except at its $n^{\text{th}}$ entry. With this selection, we have $\bm{v}^{\mathrm{T}}\mathbf{J}\bm{v} = [\bm{v}]_{[n]}^{\mathrm{T}}[\mathbf{J}]_{[n,n]}[\bm{v}]_{[n]} = 0$ and by definition, $\mathbf{J} \nsucc 0$. The proof is complete.

	\subsection{Misoriented and Misaligned First Derivatives}
	\label{Appendix:first_derivative}
	We rewrite the noise free part of the signal as
$$
	\mu_{t,u}[n]=\mu_{t,u, \mathrm{DB}}[n]+\mu_{t,u, \mathrm{DBR}}[n], $$$$
		\mu_{t,u}[n]=\mu_{t,u, \mathrm{DB}}[n]+\sum_{m=1}^{M_1}\mu_{t,u, \mathrm{DBR}}^{[m]}[n],
		$$
with
$$
\mu_{t,u, \mathrm{DB}}[n] \triangleq  \beta^{[{0}]} {\rho_{}^{[{0}]}} \sum_{{d} = 1}^{{N}_{D}} \sum_{{{b}} = 1}^{N_{{B}}} \mu_{t,u, {db}}[n] \\
$$ 
and 
$$
\mu_{t,u, \mathrm{DBR}}^{[m]}[n] \triangleq \gamma_{t}^{[{m}]}\beta^{[{m}]} \sum_{{d} = 1}^{{N}_{D}} \sum_{{{b}} = 1}^{N_{{B}}}  \sum_{r=1}^{N_{R}^{[m]}}   e^{j  \vartheta_{{{r}}}^{[{m}]} } {\rho_{\bm{r}_r}^{[{m}]}} \mu_{t,u, {dbr}}^{[m]}[n].
$$
The signals inside the summations are $$
\mu_{t,u, {db}}[n] \triangleq x_{{d}}[n]  e^{-j2\pi f_{n}({\epsilon}_{}^{[{0}]} - \tau_{ { {\bm{b}_{b}} \bm{p}_{{d}}}} + \tau_{{\bm{b}_{b}} {\bm{u}_{u}}}^{[0]})}$$ and $$
 \mu_{t,u, {dbr}}^{[m]}[n]\triangleq x_{{d}}[n] \times    e^{-j 2 \pi f_{n}({\epsilon}_{}^{[m]} - \tau_{ { {\bm{b}_{b}} \bm{p}_{{d}}}} +\tau_{{\bm{b}_{b}} {\bm{r}_{r}}}^{[{m}]} + \tau_{{\bm{r}_{r}} {\bm{u}_{u}}}^{[{m}]}) }.
$$
\subsubsection{Misoriented and Misaligned First Derivatives Concerning the Geometric Channel Parameters}
	For the LOS path, we define $\nu \in \{{\theta_{{\bm{b}}_{{B}^{}}{\bm{u}}_{{U}^{}}}},{\phi_{{\bm{b}}_{{B}^{}}{\bm{u}}_{{U}^{}}}}, \bm{\Phi}_U  \}$ and present the first derivatives concerning the path from the BS to a possibly misoriented and misaligned UE
$$
\nabla_{{\nu} } {\mu}_{t, {{u}}}[n] =   \beta^{[{0}]} {\rho_{}^{[{0}]}} \sum_{{d} = 1}^{{N}_{D}} \sum_{{{b}} = 1}^{N_{{B}}} \nabla_{{\nu} }\mu_{t,u, {db}}[n], 
$$
$$
\begin{aligned}
    \nabla_{{\tau_{{\bm{b}}_{{B}^{}}{\bm{u}}_{{U}^{}}}^{[{0}]}} } {\mu}_{t, {{u}}}[n] &=   \beta^{[{0}]}  \sum_{{d} = 1}^{{N}_{D}} \sum_{{{b}} = 1}^{N_{{B}}} \left[ \mu_{t,u, {db}}[n] \nabla_{{\tau_{{\bm{b}}_{{B}^{}}{\bm{u}}_{{U}^{}}}^{[{0}]}} }{\rho_{}^{[{0}]}} \right.  \\ & \left.+  {\rho_{}^{[{0}]}}  \nabla_{{\tau_{{\bm{b}}_{{B}^{}}{\bm{u}}_{{U}^{}}}^{[{0}]}} }\mu_{t,u, {db}}[n] \right].
    \end{aligned}
$$
	To derive the FIM related to the RIS, we define $\nu \in \{ \theta_{{\bm{g}}_{{G}^{}}{\bm{v}}_{{V}^{}}}, \phi_{{\bm{g}}_{{G}^{}}{\bm{v}}_{{V}^{}}}, \tau_{{\bm{g}}_{{G}^{}}{\bm{v}}_{{V}^{}}},\bm{\Phi}_{R}^{[m]}\}$ and present the first derivatives  concerning the parameters in the RIS path for both the BS-RIS link and the RIS-UE link
	\begin{equation}
\label{appendix_equ:FIM_first_derivative_theta_RIS}
\begin{aligned}
\nabla_{{\nu} } {\mu}_{t, {{u}}}[n] &= \gamma_{t}^{[{m}]}\beta^{[{m}]} \sum_{{d} = 1}^{{N}_{D}} \sum_{{{b}} = 1}^{N_{{B}}}  \sum_{r=1}^{N_{R}^{[m]}}   e^{j  \vartheta_{{{r}}}^{[{m}]} } \\ & \times \left[ \mu_{t,u, {dbr}}^{[m]}[n]  \nabla_{{\nu} } {\rho_{\bm{r}_r}^{[{m}]}} + {\rho_{\bm{r}_r}^{[{m}]}} \nabla_{{\nu} } \mu_{t,u, {dbr}}^{[m]}[n]  \right],
   \end{aligned}
\end{equation}
also,
	\begin{equation}
\label{appendix_equ:FIM_first_derivative_orient}
\begin{aligned}
\nabla_{{\bm{\Phi}^{}_{U}} } {\mu}_{t, {{u}}}[n] &= \gamma_{t}^{[{m}]}\beta^{[{m}]} \sum_{{d} = 1}^{{N}_{D}} \sum_{{{b}} = 1}^{N_{{B}}} \sum_{m = 1}^{M_1} \sum_{r=1}^{N_{R}^{[m]}}   e^{j  \vartheta_{{{r}}}^{[{m}]} } \\ & \times \left[ \mu_{t,u, {dbr}}^{[m]}[n]  \nabla_{{\bm{\Phi}^{}_{U}} } {\rho_{\bm{r}_r}^{[{m}]}} + {\rho_{\bm{r}_r}^{[{m}]}} \nabla_{{\bm{\Phi}^{}_{U}} } \mu_{t,u, {dbr}}^{[m]}[n]  \right].
   \end{aligned}
\end{equation}
Defining $\nu \in \{\theta_{{\bm{b}}_{{B}^{}}{\bm{u}}_{{U}^{}}}, \phi_{{\bm{b}}_{{B}^{}}{\bm{u}}_{{U}^{}}}, \tau_{{\bm{b}}_{{B}^{}}{\bm{u}}_{{U}^{}}},\bm{\Phi}^{}_{U}\}$, the derivatives related to the BS-UE path are given by
$$
\label{appendix_equ:FIM_first_derivative_theta_BS_UE_1}
\nabla_{{\nu} } \mu_{t,u, {db}}[n] = (-j 2 \pi f_{n}) \nabla_{{\nu} }\tau_{{\bm{b}_{b}} {\bm{u}_{u}}} \mu_{t,u, {db}}[n]$$ { and }
$$
\nabla_{{\nu} } \tau_{{\bm{b}_{b}} {\bm{u}_{u}}} =  \nabla_{{\nu} }d_{{\bm{b}_{b}} {\bm{u}_{u}}} / {c}.$$
Defining $\nu \in \{\theta_{{\bm{g}}_{{G}^{}}{\bm{v}}_{{V}^{}}},\phi_{{\bm{g}}_{{G}^{}}{\bm{v}}_{{V}^{}}},\tau_{{\bm{g}}_{{G}^{}}{\bm{v}}_{{V}^{}}},\bm{\Phi}_{R}^{[m]},\bm{\Phi}^{}_{U}\}$, the derivatives related to both the RIS-UE and the BS-RIS links are given by
$$\nabla_{{\nu} } \mu_{t,u, {dbr}}^{[m]}[n] = (-j 2 \pi f_{n}) \nabla_{{\nu} }\tau_{{\bm{g}_{g}} {\bm{v}_{v}}} \mu_{t,u, {dbr}}^{[m]}[n]$$
{ and }
$$\nabla_{{\nu} } \tau_{{\bm{g}_{g}} {\bm{v}_{v}}} =  \nabla_{{\nu} }d_{{\bm{g}_{g}} {\bm{v}_{v}}} / {c}.$$
The derivatives related to the distances with respect to geometric channel parameters are presented next
$$
\nabla_{{\theta_{{\bm{g}}_{{G}^{}}{\bm{v}}_{{V}^{}}}} } d_{{\bm{g}_{g}} {\bm{v}_{v}}} =  \frac{d_{{\bm{g}_{G}} {\bm{v}_{V}}}}{d_{{\bm{g}_{g}} {\bm{v}_{v}}}} (\nabla_{{\theta_{{\bm{g}}_{{G}^{}}{\bm{v}}_{{V}^{}}}} } \Delta_{{{\bm{g}}_{{G}^{}}{\bm{v}}_{{V}^{}}}} ) ^{\mathrm{T}} (\bm{s}_{v} - \bm{s}_{g}),$$

$$
\nabla_{{\phi_{{\bm{g}}_{{G}^{}}{\bm{v}}_{{V}^{}}}} } d_{{\bm{g}_{g}} {\bm{v}_{v}}} =  \frac{d_{{\bm{g}_{G}} {\bm{v}_{V}}}}{d_{{\bm{g}_{g}} {\bm{v}_{v}}}} (\nabla_{{\phi_{{\bm{g}}_{{G}^{}}{\bm{v}}_{{V}^{}}}} } \Delta_{{{\bm{g}}_{{G}^{}}{\bm{v}}_{{V}^{}}}} ) ^{\mathrm{T}} (\bm{s}_{v} - \bm{s}_{g}), $$
$$
\nabla_{{\tau_{{\bm{g}}_{{G}^{}}{\bm{v}}_{{V}^{}}}} } d_{{\bm{g}_{g}} {\bm{v}_{v}}} =  {c} \frac{d_{{\bm{g}_{G}} {\bm{v}_{V}}} + \Delta_{{{\bm{g}}_{{G}^{}}{\bm{v}}_{{V}^{}}}} ^{\mathrm{T}} (\bm{s}_{v} - \bm{s}_{g})}{d_{{\bm{g}_{g}} {\bm{v}_{v}}}}.$$
The derivatives related to the distances with respect to orientation of the entities are presented next
\begin{align}
\label{appendix_equ:FIM_first_derivative_theta_orient}
\nabla_{{\bm{\Phi}_V} } d_{{\bm{g}_{g}} {\bm{v}_{v}}} &=  \frac{d_{{\bm{g}_{G}} {\bm{v}_{V}}}}{d_{{\bm{g}_{g}} {\bm{v}_{v}}}} ( \Delta_{{{\bm{g}}_{{G}^{}}{\bm{v}}_{{V}^{}}}} ) ^{\mathrm{T}} (\nabla_{{\bm{\Phi}_V} }\bm{Q}_V\Tilde{\bm{s}_{v}}) 
\\ &- \frac{1}{d_{{\bm{g}_{g}} {\bm{v}_{v}}}} (\nabla_{{\bm{\Phi}_V} }\bm{Q}_V\Tilde{\bm{s}_{v}})^{\mathrm{T}}({\bm{s}_{g}}), \\
\label{appendix_equ:FIM_first_derivative_phi_orient}
\nabla_{{\bm{\Phi}_G} } d_{{\bm{g}_{g}} {\bm{v}_{v}}} &= - \frac{d_{{\bm{g}_{G}} {\bm{v}_{V}}}}{d_{{\bm{g}_{g}} {\bm{v}_{v}}}} ( \Delta_{{{\bm{g}}_{{G}^{}}{\bm{v}}_{{V}^{}}}} ) ^{\mathrm{T}} (\nabla_{{\bm{\Phi}_G} }\bm{Q}_G\Tilde{\bm{s}_{g}}) \\ &-\frac{1}{d_{{\bm{g}_{g}} {\bm{v}_{v}}}}  (\bm{s}_v)^{\mathrm{T}}(\nabla_{{\bm{\Phi}_G} }\bm{Q}_G{\Tilde{\bm{s}}_{g}}).
\end{align}
The derivatives related to the unit vectors from one entity to another entity with respect to geometric  parameters are presented as $\nabla_{{\theta_{{\bm{g}}_{{G}^{}}{\bm{v}}_{{V}^{}}}} } \Delta_{{{\bm{g}}_{{G}^{}}{\bm{v}}_{{V}^{}}}}  = [\cos {{\phi_{{\bm{g}}_{{G}^{}}{\bm{v}}_{{V}^{}}}} } \cos {{\theta_{{\bm{g}}_{{G}^{}}{\bm{v}}_{{V}^{}}}} }, \sin {{\phi_{{\bm{g}}_{{G}^{}}{\bm{v}}_{{V}^{}}}} } \cos {{\theta_{{\bm{g}}_{{G}^{}}{\bm{v}}_{{V}^{}}}} }, -\sin {{\theta_{{\bm{g}}_{{G}^{}}{\bm{v}}_{{V}^{}}}} }]^{\mathrm{T}}$ and $\nabla_{{\phi_{{\bm{g}}_{{G}^{}}{\bm{v}}_{{V}^{}}}} } \Delta_{{{\bm{g}}_{{G}^{}}{\bm{v}}_{{V}^{}}}}  = [-\sin {{\phi_{{\bm{g}}_{{G}^{}}{\bm{v}}_{{V}^{}}}} } \sin {{\theta_{{\bm{g}}_{{G}^{}}{\bm{v}}_{{V}^{}}}} }, \cos {{\phi_{{\bm{g}}_{{G}^{}}{\bm{v}}_{{V}^{}}}} } \sin {{\theta_{{\bm{g}}_{{G}^{}}{\bm{v}}_{{V}^{}}}} }, 0]^{\mathrm{T}}.
$

The pathloss related derivatives are also presented. The derivatives related to the LOS are  $\nabla_{{\tau_{{\bm{b}}_{{B}^{}}{\bm{u}}_{{U}^{}}}} }{\rho_{}^{[{0}]}} = -{c}\frac{\lambda}{4\pi}{d^{-2}_{{\bm{b}}_{{B}^{}}{\bm{u}}_{{U}^{}}}}.$
The derivative of the pathloss with respect to $\nu \in \{{\theta_{{\bm{r}}_{{R}^{}}{\bm{u}}_{{U}^{}}}}, {\theta_{{\bm{r}}_{{R}^{}}{\bm{u}}_{{U}^{}}}}, {\tau_{{\bm{r}}_{{R}^{}}{\bm{u}}_{{U}^{}}}} \}$ are presented
	\begin{equation}
\label{appendix_equ:FIM_first_derivative_pathloss_theta_RIS_1}
\begin{aligned}
 \nabla_{{\nu} }\rho_{\bm{r}_{r}} &=    \frac{\lambda^2 \rho_{r, {\bm{r}}_{{r}^{}}{\bm{b}}_{{b}^{}}}^{q_0}}{16 \pi \left(d_{{\bm{b}}_{{b}^{}}{\bm{r}}_{{r}^{}}}\right)^{q_0 +1}}  \left[ q_0 d_{ {\bm{r}}_{{r}^{}}{\bm{u}}_{{u}^{}}^{}}^{-(q_0 +1)}   \rho_{r,{\bm{r}}_{{r}^{}}{\bm{u}}_{{u}^{}}}^{q_0 - 1} \nabla_{{\nu}  } \rho_{r,{\bm{r}}_{{r}^{}}{\bm{u}}_{{u}^{}}}^{} \right.  \\ & \left. -(q_0 +1)  \rho_{r,{\bm{r}}_{{r}^{}}{\bm{u}}_{{u}^{}}}^{q_0}  d_{ {\bm{r}}_{{r}^{}}{\bm{u}}_{{u}^{}}^{}}^{-q_0 -2 }  \nabla_{{\nu} }d_{ {\bm{r}}_{{r}^{}}{\bm{u}}_{{u}^{}}^{}}^{}    \right]\epsilon_{p},
        \end{aligned}
\end{equation}
and the derivative of the pathloss with respect to $ \nu \in \{ {\theta_{{\bm{b}}_{{B}^{}}{\bm{r}}_{{R}^{}}}}, {\phi_{{\bm{b}}_{{B}^{}}{\bm{r}}_{{R}^{}}}}, {\tau_{{\bm{b}}_{{B}^{}}{\bm{r}}_{{R}^{}}}} \}$ are presented
	\begin{equation}
\label{appendix_equ:FIM_first_derivative_pathloss_theta_RIS_3}
\begin{aligned}
 \nabla_{{\nu} }\rho_{\bm{r}_{r}} &=    \frac{\lambda^2 \rho_{r, {\bm{r}}_{{r}^{}}{\bm{u}}_{{u}^{}}}^{q_0}}{16 \pi \left(d_{{\bm{r}}_{{r}^{}}{\bm{u}}_{{u}^{}}}\right)^{q_0 +1}}  \left[ q_0 d_{ {\bm{b}}_{{b}^{}}{\bm{r}}_{{r}^{}}^{}}^{-(q_0 +1)}   \rho_{r,{\bm{r}}_{{r}^{}}{\bm{b}}_{{b}^{}}}^{q_0 - 1} \nabla_{{\nu}  } \rho_{r,{\bm{r}}_{{r}^{}}{\bm{b}}_{{b}^{}}}^{}  \right. \\ & \left. - (q_0 +1)  \rho_{r,{\bm{r}}_{{r}^{}}{\bm{b}}_{{b}^{}}}^{q_0}  d_{ {\bm{b}}_{{b}^{}}{\bm{r}}_{{r}^{}}^{}}^{-q_0 -2 }  \nabla_{{\nu} }d_{ {\bm{b}}_{{b}^{}}{\bm{r}}_{{r}^{}}^{}}^{}    \right]\epsilon_{p},
        \end{aligned}
\end{equation}
where
$$
\nabla_{{\theta_{{\bm{g}}_{{G}^{}}{\bm{v}}_{{V}^{}}}}  } \rho_{r,{{\bm{g}}_{{g}^{}}{\bm{v}}_{{v}^{}}}}^{} =   \left[{{d_{{\bm{g}}_{{G}^{}}{\bm{v}}_{{V}^{}}}}}  \nabla_{{\theta_{{\bm{g}}_{{G}^{}}{\bm{v}}_{{V}^{}}}} }\Delta_{ {{\bm{g}}_{{G}^{}}{\bm{v}}_{{V}^{}}}}^{}   \right]^{\mathrm{T}} \bm{Q}_{G} \bm{a}_{\Tilde{\bm{p}}_{{G}}}, $$
$$
\nabla_{{\phi_{{\bm{g}}_{{G}^{}}{\bm{v}}_{{V}^{}}}}  } \rho_{r,{{\bm{g}}_{{g}^{}}{\bm{v}}_{{v}^{}}}}^{} =   \left[{{d_{{\bm{g}}_{{G}^{}}{\bm{v}}_{{V}^{}}}}}  \nabla_{{\phi_{{\bm{g}}_{{G}^{}}{\bm{v}}_{{V}^{}}}} }\Delta_{ {{\bm{g}}_{{G}^{}}{\bm{v}}_{{V}^{}}}}^{}   \right]^{\mathrm{T}} \times \bm{Q}_{G} \bm{a}_{\Tilde{\bm{p}}_{{G}}}, $$
$$
\nabla_{{\tau_{{\bm{g}}_{{G}^{}}{\bm{v}}_{{V}^{}}}}  } \rho_{r,{{\bm{g}}_{{g}^{}}{\bm{v}}_{{v}^{}}}}^{} =   {c}\left[\Delta_{ {{\bm{g}}_{{G}^{}}{\bm{v}}_{{V}^{}}}}^{}   \right]^{\mathrm{T}} \bm{Q}_{G} \bm{a}_{\Tilde{\bm{p}}_{{G}}}, $$
$$
\nabla_{{\bm{\Phi}_G{}}  } \rho_{r,{{\bm{g}}_{{g}^{}}{\bm{v}}_{{v}^{}}}}^{} =   \bigg[  d_{{\bm{p}}_{{G}}{\bm{p}}_{{V}^{}}}{\Delta_{{\bm{p}}_{{G}}{\bm{p}}_{{V}^{}}}} +  \bm{Q}_{V} \Tilde{\bm{s}}_{v}    \bigg]^{\mathrm{T}}
\nabla_{{\bm{\Phi}_G{}}  }\bm{Q}_{G} \bm{a}_{\Tilde{\bm{p}}_{{G}}},        
$$
and 
$$
\nabla_{{\bm{\Phi}_V{}}  } \rho_{r,{{\bm{g}}_{{g}^{}}{\bm{v}}_{{v}^{}}}}^{} =   \bigg[  \nabla_{{\bm{\Phi}_V{}}  } \bm{Q}_{V} \Tilde{\bm{s}}_{v}    \bigg]^{\mathrm{T}}
\bm{Q}_{G} \bm{a}_{\Tilde{\bm{p}}_{{G}}}.  $$        
\subsubsection{Misoriented and Misaligned First Derivatives Concerning the Nuisance Channel Parameters} 
The derivatives of the signal with respect to the  LOS related nuisance parameters in the BS-UE path are $$\nabla_{{\epsilon}_{}^{[{0}]} } {\mu}_{t, {{u}}}[n] =  (-j 2 \pi f_{n}) \beta^{[{0}]} {\rho_{}^{[{0}]}} \sum_{{d} = 1}^{{N}_{D}} \sum_{{{b}} = 1}^{N_{{B}}} \mu_{t,u, {db}}[n],$$ $$\nabla_{{\beta}_{\mathrm{R}}^{[{0}]}} {\mu}_{t, {{u}}}[n] =  {\rho_{}^{[{0}]}} \sum_{{d} = 1}^{{N}_{D}} \sum_{{{b}} = 1}^{N_{{B}}} \mu_{t,u, {db}}[n],$$ and $$\nabla_{{\beta}_{\mathrm{I}}^{[{0}]}} {\mu}_{t, {{u}}}[n] =  j{\rho_{}^{[{0}]}} \sum_{{d} = 1}^{{N}_{D}} \sum_{{{b}} = 1}^{N_{{B}}} \mu_{t,u, {db}}[n]. 
$$
The derivatives of the signal with respect to the  RIS related nuisance parameters are
$$\nabla_{{\epsilon}_{}^{[{m}]} } {\mu}_{t, {{u}}}[n] =  (-j 2 \pi f_{n}) \gamma_{t}^{[{m}]}\beta^{[{m}]} \sum_{{d} = 1}^{{N}_{D}} \sum_{{{b}} = 1}^{N_{{B}}}  \sum_{r=1}^{N_{R}^{[m]}}   e^{j  \vartheta_{{{r}}}^{[{m}]} } {\rho_{\bm{r}_r}^{[{m}]}} \mu_{t,u, {dbr}}^{[m]}[n],$$ $$\nabla_{{\beta}_{\mathrm{R}}^{[{m}]}} {\mu}_{t, {{u}}}[n] =  \gamma_{t}^{[{m}]} \sum_{{d} = 1}^{{N}_{D}} \sum_{{{b}} = 1}^{N_{{B}}}  \sum_{r=1}^{N_{R}^{[m]}}   e^{j  \vartheta_{{{r}}}^{[{m}]} } {\rho_{\bm{r}_r}^{[{m}]}} \mu_{t,u, {dbr}}^{[m]}[n],$$ and $$\nabla_{{\beta}_{\mathrm{I}}^{[{m}]}} {\mu}_{t, {{u}}}[n] =   j\nabla_{{\beta}_{\mathrm{I}}^{[{m}]}} {\mu}_{t, {{u}}}[n].$$

\subsection{Proof of Lemma \ref{lemma:farfield_channel_parameter_estimatable}}
\label{Appendix:farfield_channel_parameter_estimatable}
This proof is concerned with the $m^{\text{th}}$ RIS path, hence the superscript $(\cdot)^{[m]}$ is dropped when notationally convenient. Due to Proposition \ref{proposition:EFIM_channel}, it suffices to show that the EFIM of a parameter vector consisting of the RIS orientation, $\bm{\Phi}^{[m]}_{R}$, and the complex path gains $\beta^{[m]}$ is zero. More specifically, it suffices to show that for $\Tilde{\bm{\eta}}^{[m]} = [\bm{\Phi}^{[m]}_{R} , \beta^{[m]}]$, the resultant EFIM, $ \mathrm{J}_{ \bm{\bm{y}}; \Tilde{\bm{\eta}}^{[m]}}^{\mathrm{e}} = 0$. To show this, we drop the unit vector parameterization  in Proposition \ref{proposition:far_field_1} (for example, $\bm{a}_{UR} = \bm{a}_{UR}(\Delta_{\bm{r}_{R}\bm{u}_U}^{[m]})$), and we derive the first derivatives based on  Assumptions \ref{assumption:RIS_restrict} and \ref{assumption:pathloss_subsumed} 
\begin{align}
\begin{split}
    \label{lemma_appendix_equ:far_field_1}
\nabla_{\bm{\Phi}_R}\bm{\mu}_{t}[n] &=    \frac{j2\pi}{\lambda} \gamma_{t}^{[{m}]} \beta^{[{m}]}   \bm{a}_{UR} \bm{a}_{RU}^{\mathrm{H}}  \bm{K}_{RU} \bm{\Gamma}^{[m]} \bm{a}_{RB} \bm{a}_{BR}^{\mathrm{H}} \\ & \times \bm{F}[n] \bm{x}[n]  e^{-j 2 \pi f_{c}(\tau_{{\bm{b}_{B}} \bm{r}_{R}}^{[m]}   + \tau_{\bm{r}_{R} {\bm{u}_{U}}}^{[m]})} \\ &-    \frac{j2\pi}{\lambda} \gamma_{t}^{[{m}]} \beta^{[{m}]} \bm{a}_{UR} \bm{a}_{RU}^{\mathrm{H}}   \bm{\Gamma}^{[m]} \bm{K}_{RB}\bm{a}_{RB} \bm{a}_{BR}^{\mathrm{H}}   \\ & \times\bm{F}[n] \bm{x}[n]  e^{-j 2 \pi f_{c}(\tau_{{\bm{b}_{B}} \bm{r}_{R}}^{[m]}   + \tau_{\bm{r}_{R} {\bm{u}_{U}}}^{[m]})}, \end{split}
\end{align}
\begin{align}
\begin{split}
\label{lemma_appendix_equ:far_field_2}
\nabla_{{\beta}_R}\bm{\mu}_{t}[n] &=  \gamma_{t}^{[{m}]}    \bm{a}_{UR} \bm{a}_{RU}^{\mathrm{H}}  \bm{\Gamma}^{[m]} \bm{a}_{RB} \bm{a}_{BR}^{\mathrm{H}} \bm{F}[n] \bm{x}[n] \\ & \times e^{-j 2 \pi f_{c}(\tau_{{\bm{b}_{B}} \bm{r}_{R}}^{[m]}   + \tau_{\bm{r}_{R} {\bm{u}_{U}}}^{[m]})},
\end{split}
\end{align}
\begin{align}
\label{lemma_appendix_equ:far_field_3}
\nabla_{{\beta}_I}\bm{\mu}_{t}[n] &=  j\nabla_{{\beta}_R}\bm{\mu}_{t}[n],
\end{align}
where $ \bm{K}_{RU} =   \text{diag}\left[\Delta_{{{\bm{r}}_{{R}^{}}{\bm{u}}_{{U}^{}}}}  ^{\mathrm{T}} (\nabla_{{\bm{\Phi}_R} }\bm{Q}_R\Tilde{\bm{S}_{r}})\right]$ and $  \bm{K}_{RB} =   \text{diag}\left[\Delta_{{{\bm{b}}_{{B}^{}}{\bm{r}}_{{R}^{}}}}  ^{\mathrm{T}} (\nabla_{{\bm{\Phi}_R} }\bm{Q}_R\Tilde{\bm{S}_{r}})\right]$. The FIM from observations, $\mathbf{J}_{ \bm{\bm{y}}| \Tilde{\bm{\eta}}^{[m]}}$, is obtained by substituting the above first derivatives into (\ref{equ:observation_FIM_1}) and it has the structure 
$$
\begin{aligned}
\mathbf{J}_{ \bm{\bm{y}}| \Tilde{\bm{\eta}}^{[m]}} =\left[\begin{array}{ccccc}
\mathbf{J}_{ \bm{\bm{y}}|{\bm{\Phi}}_{R}} & 
\mathbf{J}_{ \bm{\bm{y}}|{\bm{\Phi}}_{R},{{\beta}}_R} &  \mathbf{J}_{ \bm{\bm{y}}|{\bm{\Phi}}_{R},{{\beta}}_I} \\
\mathbf{J}_{ \bm{\bm{y}}|{\bm{\Phi}}_{R},{{\beta}}_R}^{\mathrm{T}} & 
\mathbf{J}_{ \bm{\bm{y}}|{{\beta}}_R} &  0 \\
\mathbf{J}_{ \bm{\bm{y}}|{\bm{\Phi}}_{R},{{\beta}}_I}^{\mathrm{T}} & 
0 &  \mathbf{J}_{ \bm{\bm{y}}|{{\beta}}_I} \\
\end{array}\right],\end{aligned}
$$
and the general FIM is $\mathbf{J}_{ \bm{\bm{y}}; \Tilde{\bm{\eta}}^{[m]}} =  \mathbf{J}_{ \bm{\bm{y}}| \Tilde{\bm{\eta}}^{[m]}} + \bm{\Xi}_{\Tilde{\bm{\eta}}^{[m]},\Tilde{\bm{\eta}}^{[m]}}^{\Tilde{\bm{\eta}}^{[m]}}$. The diagonal nature of the FIM from {\em a priori} information about the orientation and the complex path gains allows us to write the general EFIM as
\begin{align}
\begin{split}
    \mathbf{J}_{ \bm{\bm{y}}; \Tilde{\bm{\eta}}^{[m]}}^{\mathrm{e}} &= \mathbf{J}_{ \bm{\bm{y}}|{\bm{\Phi}}_{R}} + \bm{\Xi}_{\bm{\Phi}_{R},\bm{\Phi}_{R}}^{\bm{\Phi}_{R}}  - [\mathbf{J}_{ \bm{\bm{y}}|{{\beta}}_{R}} + \bm{\Xi}_{{\beta}_{R},{\beta}_{R}}^{{\beta}_{R}} ]^{-1} \\ & \times [\mathbf{J}_{ \bm{\bm{y}}|{\bm{\Phi}}_{R},{{\beta}}_R}\mathbf{J}_{ \bm{\bm{y}}|{\bm{\Phi}}_{R},{{\beta}}_R}^{\mathrm{T}} + \mathbf{J}_{ \bm{\bm{y}}|{\bm{\Phi}}_{R},{{\beta}}_I}\mathbf{J}_{ \bm{\bm{y}}|{\bm{\Phi}}_{R},{{\beta}}_I}^{\mathrm{T}}], \\
    \end{split}
    \end{align}
\begin{align}
    \begin{split}
\mathbf{J}_{ \bm{\bm{y}}; \Tilde{\bm{\eta}}^{[m]}}^{\mathrm{e}}  &= \mathbf{J}_{ \bm{\bm{y}}|{\bm{\Phi}}_{R}} +  \bm{\Xi}_{\bm{\Phi}_{R},\bm{\Phi}_{R}}^{\bm{\Phi}_{R}}  - [\mathbf{J}_{ \bm{\bm{y}}|{{\beta}}_{R}} + \bm{\Xi}_{{\beta}_{R},{\beta}_{R}}^{{\beta}_{R}} ]^{-1}\mathbf{J}_{ \bm{\bm{y}}|{{\beta}}_{R}}\mathbf{J}_{ \bm{\bm{y}}|{\bm{\Phi}}_{R}},
\end{split}
\end{align}
the second equation results from noticing $$\mathbf{J}_{ \bm{\bm{y}}|{{\beta}}_{R}}\mathbf{J}_{ \bm{\bm{y}}|{\bm{\Phi}}_{R}} = \mathbf{J}_{ \bm{\bm{y}}|{\bm{\Phi}}_{R},{{\beta}}_R}\mathbf{J}_{ \bm{\bm{y}}|{\bm{\Phi}}_{R},{{\beta}}_R}^{\mathrm{T}} + \mathbf{J}_{ \bm{\bm{y}}|{\bm{\Phi}}_{R},{{\beta}}_I}\mathbf{J}_{ \bm{\bm{y}}|{\bm{\Phi}}_{R},{{\beta}}_I}^{\mathrm{T}}$$. The proof follows from the second equation, as $\mathbf{J}_{ \bm{\bm{y}}; \Tilde{\bm{\eta}}^{[m]}}^{\mathrm{e}} = 0$, when both $\bm{\Xi}_{\bm{\Phi}_{R},\bm{\Phi}_{R}}^{\bm{\Phi}_{R}} = 0$ and $\bm{\Xi}_{{\beta}_{R},{\beta}_{R}}^{{\beta}_{R}} = 0$.

\subsection{Proof of Lemma \ref{lemma:near-field_channel_parameter_estimatable}}
\label{Appendix:near-field_channel_parameter_estimatable}
This proof is concerned with the $m^{\text{th}}$ RIS path, hence the superscript $(\cdot)^{[m]}$ is dropped when notationally convenient. Similar to Lemma \ref{lemma:farfield_channel_parameter_estimatable},  it suffices to show that the EFIM of a parameter vector, $\Tilde{\bm{\eta}}^{[m]} = [\bm{\Phi}^{[m]}_{R},\beta^{[m]}]$ is zero. To proceed, the first derivatives of $\bm{\Phi}_{R},\beta_{R},\beta_{I}$ can be obtained  using the derivatives presented in Appendix \ref{Appendix:first_derivative}. However, a simplified expression for the resulting FIM will not be easily obtainable. Hence, without loss of generality we assume $N_B =  1$ and $N_D = 1$. These conditions allow us to rewrite the portion of the received signal in (\ref{equ:receive_processing_2}) related to the RIS as
\begin{equation}
\label{appendix_equ:receive_processing_1}
\begin{aligned}
{\mu}_{t, {{u}}}[n] &=   
  \sum_{{m} = 1}^{{M}_1} \gamma_{t}^{[{m}]}\beta^{[{m}]}    \sum_{{{r}} = 1}^{N_{{R}}^{[{m}]}} e^{j  \vartheta_{{{r}}}^{[{m}]} } e^{-j 2 \pi f_{c}(\tau_{{\bm{b}_{b}} {\bm{r}_{r}}}^{[{m}]} + \tau_{{\bm{r}_{r}} {\bm{u}_{u}}}^{[{m}]}) } \\ & \times e^{j 2 \pi f_{n}\tau_{ { {\bm{b}_{b}} \bm{p}_{{D}}}} } x_{{d}}[n].
  \end{aligned}
\end{equation}
In matrix form, we can write $${\mu}_{t, {{u}}}[n] =   
  \sum_{{m} = 1}^{{M}_1} \gamma_{t}^{[{m}]}\beta^{[{m}]} \bm{a}^{\mathrm{T}}(\bm{u}_u)\bm{\Gamma}^{[m]} \bm{a}(\bm{b}_b)  e^{j 2 \pi f_{n}\tau_{ { {\bm{b}_{b}} \bm{p}_{{d}}}} } x_{{d}}[n],$$
where $\bm{a}(\bm{u}_u) = [e^{-j 2 \pi f_{c} \tau_{{\bm{r}_{1}} {\bm{u}_{u}}}^{[{m}]}},\cdots, e^{-j 2 \pi f_{c} \tau_{{\bm{r}_{N_R^{[m]}}} {\bm{u}_{u}}}^{[{m}]} }]^{\mathrm{T}}$ and $\bm{a}(\bm{b}_b) = [e^{-j 2 \pi f_{c} \tau_{{\bm{b}_{b}} {\bm{r}_{1}} }^{[{m}]}},\cdots, e^{-j 2 \pi f_{c} \tau_{{\bm{b}_{b}}{\bm{r}_{N_R^{[m]}}} }^{[{m}]} }]^{\mathrm{T}}$. With the above equation, the general FIM is 
\begin{equation}
\begin{aligned}
&\mathbf{J}_{ \bm{\bm{y}};\Tilde{\bm{\eta}}^{[m]}} = \\
&\left[\begin{array}{ccccc}
\mathbf{J}_{ \bm{\bm{y}}|{\bm{\Phi}}_{R}} + \bm{\Xi}_{\bm{\Phi}_{R},\bm{\Phi}_{R}}^{\bm{\Phi}_{R}}& 
\mathbf{J}_{ \bm{\bm{y}}|{\bm{\Phi}}_{R},{{\beta}}_R} &  \mathbf{J}_{ \bm{\bm{y}}|{\bm{\Phi}}_{R},{{\beta}}_I} \\
\mathbf{J}_{ \bm{\bm{y}}|{\bm{\Phi}}_{R},{{\beta}}_R}^{\mathrm{T}} & 
\mathbf{J}_{ \bm{\bm{y}}|{{\beta}}_R} + \bm{\Xi}_{{\beta}_{R},{\beta}_{R}}^{{\beta}_{R}} &  0 \\
\mathbf{J}_{ \bm{\bm{y}}|{\bm{\Phi}}_{R},{{\beta}}_I}^{\mathrm{T}} & 
0 &  \mathbf{J}_{ \bm{\bm{y}}|{{\beta}}_I} + \bm{\Xi}_{{\beta}_{R},{\beta}_{R}}^{{\beta}_{R}}\\
\end{array}\right].\end{aligned}
\end{equation}
Defining $\bm{K}^{}(\bm{g}_g) = \text{diag}\left[\nabla_{{\bm{\Phi}}_{R} }\tau_{{\bm{r}_{1}} {\bm{g}_{g}}}, \nabla_{{\bm{\Phi}}_{R} }\tau_{{\bm{r}_{2}} {\bm{g}_{g}}}, \cdots, \nabla_{{\bm{\Phi}}_{R} }\tau_{{\bm{r}_{N_R}} {\bm{g}_{g}}}\right]$, the FIMs from observations are written as shown in (\ref{equ:near_field}) - (\ref{equ:near_field_3}).
\begin{figure*}
\begin{align}
\begin{split}
\label{equ:near_field}
&\mathbf{J}_{ \bm{\bm{y}}|{\bm{\Phi}}_{R}} = 2/N_0 (2\pi f_c)^2\sum_{n = 1}^{N} |x_d[n]|^2 |\beta^{[m]}|^2 \times \\
&\sum_{u = 1}^{N_U}\Re \left\{\bm{a}^{\mathrm{H}}(\bm{b}_b) \bigg[ \bm{\Gamma}^{[m] \mathrm{H}}\bm{K}^{*}(\bm{u}_u)  + \bm{K}^{\mathrm{H}}(\bm{b}_b)\bm{\Gamma}^{[m] \mathrm{H}}     \bigg]\bm{a}(\bm{u}_u) \bm{a}^{\mathrm{T}}(\bm{u}_u)\bigg[ \bm{K}^{\mathrm{T}}(\bm{u}_u)\bm{\Gamma}^{[m]}  + \bm{\Gamma}^{[m] }\bm{K}^{}(\bm{b}_b)     \bigg] \bm{a}(\bm{b}_b)  \right\},
\end{split}
\end{align}
\end{figure*}
\begin{figure*}
\begin{align}
\begin{split}
\label{equ:near_field_1}
\mathbf{J}_{ \bm{\bm{y}}|{\bm{\Phi}}_{R},{{\beta}}_R} =2/ N_0 (2 \pi f_c)&\sum_{n = 1}^{N} |x_d[n]|^2 \times \\ &\sum_{u = 1}^{N_U}\Re \left\{ j\beta^{[m] \mathrm{H}} \bm{a}^{\mathrm{H}}(\bm{b}_b) \bigg[ \bm{\Gamma}^{[m] \mathrm{H}}\bm{K}^{*}(\bm{u}_u)  + \bm{K}^{\mathrm{H}}(\bm{b}_b)\bm{\Gamma}^{[m] \mathrm{H}}     \bigg]\bm{a}(\bm{u}_u) \bm{a}^{\mathrm{T}}(\bm{u}_u) \bm{\Gamma}^{[m]} \bm{a}(\bm{b}_b)  \right\},
\end{split}
\end{align}
\end{figure*}
\begin{figure*}
\begin{align}
\begin{split}
\label{equ:near_field_2}
\mathbf{J}_{ \bm{\bm{y}}|{\bm{\Phi}}_{R},{{\beta}}_I} = -2/N_0 (2 \pi f_c)&\sum_{n = 1}^{N} |x_d[n]|^2 \times \\ &\sum_{u = 1}^{N_U}\Re \left\{ \beta^{[m] \mathrm{H}} \bm{a}^{\mathrm{H}}(\bm{b}_b) \bigg[ \bm{\Gamma}^{[m] \mathrm{H}}\bm{K}^{*}(\bm{u}_u)  + \bm{K}^{\mathrm{H}}(\bm{b}_b)\bm{\Gamma}^{[m] \mathrm{H}}     \bigg]\bm{a}(\bm{u}_u) \bm{a}^{\mathrm{T}}(\bm{u}_u) \bm{\Gamma}^{[m]} \bm{a}(\bm{b}_b)  \right\},
\end{split}
\end{align}
\end{figure*}

\begin{figure*}
\begin{align}
\label{equ:near_field_3}
\mathbf{J}_{ \bm{\bm{y}}|{{\beta}}_R} =\mathbf{J}_{ \bm{\bm{y}}|{{\beta}}_I} = 2/N_0 \sum_{n = 1}^{N} |x_d[n]|^2 \sum_{u = 1}^{N_U} | \bm{a}^{\mathrm{T}}(\bm{u}_u) \bm{\Gamma}^{[m]} \bm{a}(\bm{b}_b)|^2.
\end{align}
\end{figure*}
Hence, the general EFIM is
\begin{equation}
\begin{aligned}
\mathbf{J}_{ \bm{\bm{y}}; \Tilde{\bm{\eta}}^{[m]}}^{\mathrm{e}} &= \mathbf{J}_{ \bm{\bm{y}}|{\bm{\Phi}}_{R}} + \bm{\Xi}_{\bm{\Phi}_{R},\bm{\Phi}_{R}}^{\bm{\Phi}_{R}}  - [\mathbf{J}_{ \bm{\bm{y}}|{{\beta}}_{R}} + \bm{\Xi}_{{\beta}_{R},{\beta}_{R}}^{{\beta}_{R}} ]^{-1}\\ & \times[\mathbf{J}_{ \bm{\bm{y}}|{\bm{\Phi}}_{R},{{\beta}}_R}\mathbf{J}_{ \bm{\bm{y}}|{\bm{\Phi}}_{R},{{\beta}}_R}^{\mathrm{T}} + \mathbf{J}_{ \bm{\bm{y}}|{\bm{\Phi}}_{R},{{\beta}}_I}\mathbf{J}_{ \bm{\bm{y}}|{\bm{\Phi}}_{R},{{\beta}}_I}^{\mathrm{T}}],\end{aligned}
\end{equation}
and with appropriate substitutions, it can be shown that $\mathbf{J}_{ \bm{\bm{y}}; \Tilde{\bm{\eta}}^{[m]}}^{\mathrm{e}} = 0$ when $N_U = 1$, and $\mathbf{J}_{ \bm{\bm{y}}; \Tilde{\bm{\eta}}^{[m]}}^{\mathrm{e}} > 0$ when $N_U > 1$. The proof is complete.

\subsection{Entries of the Transformation Matrix}
\label{appendix:entries_of_trans}
Defining $\nu$ as a parameter related to position, $\bm{p}$, of the centroid of an entity  or a parameter related to the misalignment, $\bm{\xi}$, of an entity; the entries of the transformation matrix related to ${\theta_{{\bm{g}}_{{G}^{}}{\bm{v}}_{{V}^{}}}}$, ${\phi_{{\bm{g}}_{{G}^{}}{\bm{v}}_{{V}^{}}}}$, and ${\tau_{{\bm{g}}_{{G}^{}}{\bm{v}}_{{V}^{}}}}$ are 
$$
\begin{aligned}\nabla_{{\nu}} {\theta_{{\bm{g}}_{{G}^{}}{\bm{v}}_{{V}^{}}}} &= -\sqrt{\frac{1}{\sin^2{\theta_{{\bm{g}}_{{G}^{}}{\bm{v}}_{{V}^{}}}}}} \left[ d^{-1}_{{\bm{g}_{G}} {\bm{v}_{V}}} \nabla_{{\nu}}(z_{{{V}^{}}}   
- z_{{{G}^{}}}) \right. \\ & \left. - (z_{{{V}^{}}}  - z_{{{G}^{}}})d^{-2}_{{\bm{g}_{G}} {\bm{v}_{V}}} \nabla_{{\nu}}d^{}_{{\bm{g}_{G}} {\bm{v}_{V}}} \right],
\end{aligned}
$$

$$
\begin{aligned}\nabla_{{\nu}} {\phi_{{\bm{g}}_{{G}^{}}{\bm{v}}_{{V}^{}}}} &= \cos^2{\theta_{{\bm{g}}_{{G}^{}}{\bm{v}}_{{V}^{}}}} \left[ (x_{{{V}^{}}}  - x_{{{G}^{}}})^{-1} \nabla_{{\nu}}(y_{{{V}^{}}}  -  y_{{{G}^{}}}) \right. \\ & \left. - (x_{{{V}^{}}}  - x_{{{G}^{}}})^{-2} (y_{{{V}^{}}}  - y_{{{G}^{}}}) \nabla_{{\nu}}(x_{{{V}^{}}}  - x_{{{G}^{}}})\right],
\end{aligned}$$ and finally $$\nabla_{{\nu}} {\tau_{{\bm{g}}_{{G}^{}}{\bm{v}}_{{V}^{}}}} = \nabla_{{\nu}} {d_{{\bm{g}}_{{G}^{}}{\bm{v}}_{{V}^{}}}} / {c}.
$$
If $\bm{\nu} =   {\bm{p}_{V}}$ or $\bm{\nu} =   {\bm{\xi}_{V}}$, then we have
$\nabla_{{\bm{\nu}}} {d_{{\bm{g}}_{{G}^{}}{\bm{v}}_{{V}^{}}}}= \bm{p}_V /  {d_{{\bm{g}}_{{G}^{}}{\bm{v}}_{{V}^{}}}}$. Likewise, if $\bm{\nu} =   {\bm{p}_{G}}$ or $\bm{\nu}  = {\bm{\xi}_{G}}$, then we have $\nabla_{{\bm{\nu}}} {d_{{\bm{g}}_{{G}^{}}{\bm{v}}_{{V}^{}}}}= -\bm{p}_G /  {d_{{\bm{g}}_{{G}^{}}{\bm{v}}_{{V}^{}}}}$. 

\subsection{Proof of Remark \ref{remark_no_priori_RIS}}
\label{appendix_no_priori_RIS} 
Now to analyze the contribution of the $m^{\text{th}}$ RIS. We pick the $m^{\text{th}}$ RIS path and analyze Lemma 5. Here, $m \in {1, \cdots, M_1}$. The contribution to $\mathbf{J}_{\bm{y};\bm{\kappa}_1}^{\mathrm{e}}$ from only the $m^{\text{th}}$ RIS is presented in (\ref{equ:prior_location_EFIM_1_only_1}).
\begin{figure*}
\begin{align}
\begin{split}
\label{equ:prior_location_EFIM_1_only_1}
&\mathbf{J}_{\bm{y};\bm{\kappa}_1}^{\mathrm{e}}=  \overline{\mathbf{\Upsilon}}_{\bm{\kappa}}^{[m]} \mathbf{J}_{\bm{y}|\bm{\eta}^{[m]}_1}^{\mathrm{e}} \overline{\mathbf{\Upsilon}}_{\bm{\kappa}}^{[m]\mathrm{T}} - \\
& \Bigg[ \Bigg[ \overline{\mathbf{\Upsilon}}_{\bm{\kappa}}^{[m]} \mathbf{J}_{\bm{y}|\bm{\eta}^{[m]}_1}^{\mathrm{e}} \overline{\overline{\mathbf{\Upsilon}}}_{\bm{\kappa}}^{[m]\mathrm{T}} + \Tilde{\bm{\Xi}}_{\bm{\kappa}^{}_{0},\bm{\kappa}^{[m]}} \Bigg] \left(\overline{\overline{\mathbf{\Upsilon}}}_{\bm{\kappa}}^{[m]} \mathbf{J}_{\bm{y}|\bm{\eta}^{[m]}_1}^{\mathrm{e}} \overline{\overline{\mathbf{\Upsilon}}}_{\bm{\kappa}}^{[m]\mathrm{T}} + \Tilde{\bm{\Xi}}_{\bm{\kappa}^{[m]},\bm{\kappa}^{[m]}}\right)^{-1}\Bigg[ \overline{\mathbf{\Upsilon}}_{\bm{\kappa}}^{[m]} \mathbf{J}_{\bm{y}|\bm{\eta}^{[m]}_1}^{\mathrm{e}} \overline{\overline{\mathbf{\Upsilon}}}_{\bm{\kappa}}^{[m]\mathrm{T}} + \Tilde{\bm{\Xi}}_{\bm{\kappa}^{}_{0},\bm{\kappa}^{[m]}} \Bigg]^{\mathrm{T}}\Bigg],
\end{split}
\end{align}
\end{figure*}
Now, with no {\em a priori} information, $\Tilde{\bm{\Xi}}_{\bm{\kappa}^{}_{0},\bm{\kappa}^{[m]}} = 0$ and $\Tilde{\bm{\Xi}}_{\bm{\kappa}^{[m]}_{},\bm{\kappa}^{[m]}} = 0$. Hence, we can write (\ref{equ:prior_location_EFIM_1_only_1}) as (\ref{equ:prior_location_EFIM_1_only_1b}).
\begin{figure*}
\begin{align}
\label{equ:prior_location_EFIM_1_only_1b}
&\mathbf{J}_{\bm{y};\bm{\kappa}_1}^{\mathrm{e}}=  \overline{\mathbf{\Upsilon}}_{\bm{\kappa}}^{[m]} \mathbf{J}_{\bm{y}|\bm{\eta}^{[m]}_1}^{\mathrm{e}} \overline{\mathbf{\Upsilon}}_{\bm{\kappa}}^{[m]\mathrm{T}} -\Bigg[ \Bigg[ \overline{\mathbf{\Upsilon}}_{\bm{\kappa}}^{[m]} \mathbf{J}_{\bm{y}|\bm{\eta}^{[m]}_1}^{\mathrm{e}} \overline{\overline{\mathbf{\Upsilon}}}_{\bm{\kappa}}^{[m]\mathrm{T}}  \Bigg] \left(\overline{\overline{\mathbf{\Upsilon}}}_{\bm{\kappa}}^{[m]} \mathbf{J}_{\bm{y}|\bm{\eta}^{[m]}_1}^{\mathrm{e}} \overline{\overline{\mathbf{\Upsilon}}}_{\bm{\kappa}}^{[m]\mathrm{T}} \right)^{-1}\Bigg[ \overline{\mathbf{\Upsilon}}_{\bm{\kappa}}^{[m]} \mathbf{J}_{\bm{y}|\bm{\eta}^{[m]}_1}^{\mathrm{e}} \overline{\overline{\mathbf{\Upsilon}}}_{\bm{\kappa}}^{[m]\mathrm{T}}\Bigg]^{\mathrm{T}}\Bigg].
\end{align}
\end{figure*}

Finally, (\ref{equ:prior_location_EFIM_1_only_1b}) can be reduced to
$\mathbf{J}_{\bm{y};\bm{\kappa}_1}^{\mathrm{e}}=  \overline{\mathbf{\Upsilon}}_{\bm{\kappa}}^{[m]} \mathbf{J}_{\bm{y}|\bm{\eta}^{[m]}_1}^{\mathrm{e}} \overline{\mathbf{\Upsilon}}_{\bm{\kappa}}^{[m]\mathrm{T}} -\Bigg[ \overline{\mathbf{\Upsilon}}_{\bm{\kappa}}^{[m]} 
\mathbf{J}_{\bm{y}|\bm{\eta}^{[m]}_1}^{\mathrm{e}}\overline{\mathbf{\Upsilon}}_{\bm{\kappa}}^{[m]\mathrm{T}}  \Bigg] = \bm{0}.$
Hence, the $m^{\text{th}}$ RIS  contributes no information for localization when there is no {\em a priori} information about its location.
\color{black}
{
\bibliographystyle{IEEEtran}
\bibliography{refs}

\begin{thebibliography}{10}
\providecommand{\url}[1]{#1}
\csname url@samestyle\endcsname
\providecommand{\newblock}{\relax}
\providecommand{\bibinfo}[2]{#2}
\providecommand{\BIBentrySTDinterwordspacing}{\spaceskip=0pt\relax}
\providecommand{\BIBentryALTinterwordstretchfactor}{4}
\providecommand{\BIBentryALTinterwordspacing}{\spaceskip=\fontdimen2\font plus
\BIBentryALTinterwordstretchfactor\fontdimen3\font minus
  \fontdimen4\font\relax}
\providecommand{\BIBforeignlanguage}[2]{{%
\expandafter\ifx\csname l@#1\endcsname\relax
\typeout{** WARNING: IEEEtran.bst: No hyphenation pattern has been}%
\typeout{** loaded for the language `#1'. Using the pattern for}%
\typeout{** the default language instead.}%
\else
\language=\csname l@#1\endcsname
\fi
#2}}
\providecommand{\BIBdecl}{\relax}
\BIBdecl

\bibitem{zekavat2011handbook}
R.~Zekavat and R.~M. Buehrer, \emph{{Handbook of Position Location: Theory,
  Practice and Advances}}.\hskip 1em plus 0.5em minus 0.4em\relax John Wiley \&
  Sons, 2011, vol.~27.

\bibitem{6894213}
S.~D'Oro, L.~Galluccio, G.~Morabito, and S.~Palazzo, ``Exploiting object group
  localization in the internet of things: Performance analysis,'' \emph{IEEE
  Trans. on Veh. Technol.}, vol.~64, no.~8, pp. 3645--3656, Aug. 2015.

\bibitem{nguyen2017wireless}
C.~L. Nguyen, O.~Georgiou, Y.~Yonezawa, and Y.~Doi, ``The wireless localization
  matching problem,'' \emph{IEEE Internet Things J.}, vol.~4, no.~5, pp.
  1312--1326, Oct. 2017.

\bibitem{8454389}
M.~Z. Win, F.~Meyer, Z.~Liu, W.~Dai, S.~Bartoletti, and A.~Conti, ``Efficient
  multisensor localization for the internet of things: Exploring a new class of
  scalable localization algorithms,'' \emph{IEEE Signal Processing Magazine},
  vol.~35, no.~5, pp. 153--167, Sep. 2018.

\bibitem{8306879}
S.~Kuutti, S.~Fallah, K.~Katsaros, M.~Dianati, F.~Mccullough, and
  A.~Mouzakitis, ``A survey of the state-of-the-art localization techniques and
  their potentials for autonomous vehicle applications,'' \emph{IEEE Internet
  Things J.}, vol.~5, no.~2, pp. 829--846, Apr. 2018.

\bibitem{7707439}
S.~Güler, B.~Fidan, S.~Dasgupta, B.~D.~O. Anderson, and I.~Shames, ``Adaptive
  source localization based station keeping of autonomous vehicles,''
  \emph{IEEE Trans. on Automatic Control}, vol.~62, no.~7, pp. 3122--3135, July
  2017.

\bibitem{9781656}
M.~Z. Win, Z.~Wang, Z.~Liu, Y.~Shen, and A.~Conti, ``Location awareness via
  intelligent surfaces: A path toward holographic {NLN},'' \emph{IEEE Veh.
  Technology Magazine}, vol.~17, no.~2, pp. 37--45, June 2022.

\bibitem{9729782}
Z.~Wang, Z.~Liu, Y.~Shen, A.~Conti, and M.~Z. Win, ``Location awareness in
  beyond {5G} networks via reconfigurable intelligent surfaces,'' \emph{IEEE J.
  Sel. Areas Commun.}, vol.~40, no.~7, pp. 2011--2025, July 2022.

\bibitem{8264743}
S.~Hu, F.~Rusek, and O.~Edfors, ``Beyond massive {MIMO}: The potential of
  positioning with large intelligent surfaces,'' \emph{IEEE Trans. on Signal
  Processing}, vol.~66, no.~7, pp. 1761--1774, Apr. 2018.

\bibitem{9508872}
A.~Elzanaty, A.~Guerra, F.~Guidi, and M.-S. Alouini, ``Reconfigurable
  intelligent surfaces for localization: Position and orientation error
  bounds,'' \emph{IEEE Trans. on Signal Processing}, vol.~69, pp. 5386--5402,
  Aug. 2021.

\bibitem{9625826}
D.~Dardari, N.~Decarli, A.~Guerra, and F.~Guidi, ``{LOS/NLOS} near-field
  localization with a large reconfigurable intelligent surface,'' \emph{IEEE
  Trans. on Wireless Commun.}, vol.~21, no.~6, pp. 4282--4294, June 2022.

\bibitem{9500663}
Z.~Abu-Shaban, K.~Keykhosravi, M.~F. Keskin, G.~C. Alexandropoulos,
  G.~Seco-Granados, and H.~Wymeersch, ``Near-field localization with a
  reconfigurable intelligent surface acting as lens,'' in \emph{Proc., IEEE
  Intl. Conf. on Commun. (ICC)}, 2021.

\bibitem{9782100}
A.~Fascista, M.~F. Keskin, A.~Coluccia, H.~Wymeersch, and G.~Seco-Granados,
  ``{RIS}-aided joint localization and synchronization with a single-antenna
  receiver: Beamforming design and low-complexity estimation,'' \emph{IEEE J.
  of Sel. Topics in Signal Processing}, to appear.

\bibitem{9528041}
K.~Keykhosravi, M.~F. Keskin, S.~Dwivedi, G.~Seco-Granados, and H.~Wymeersch,
  ``Semi-passive {3D} positioning of multiple {RIS}-enabled users,'' \emph{IEEE
  Trans. on Veh. Technol.}, vol.~70, no.~10, pp. 11\,073--11\,077, Oct. 2021.

\bibitem{9774917}
K.~Keykhosravi, M.~F. Keskin, G.~Seco-Granados, P.~Popovski, and H.~Wymeersch,
  ``{RIS}-enabled {SISO} localization under user mobility and spatial-wideband
  effects,'' \emph{IEEE J. of Sel. Topics in Signal Processing}, pp. 1--1, to
  appear.

\bibitem{5109631}
M.~Sun and K.~C. Ho, ``Successive and asymptotically efficient localization of
  sensor nodes in closed-form,'' \emph{IEEE Trans. on Signal Processing},
  vol.~57, no.~11, pp. 4522--4537, Nov. 2009.

\bibitem{4802193}
H.~Wymeersch, J.~Lien, and M.~Z. Win, ``Cooperative localization in wireless
  networks,'' \emph{Proceedings of the IEEE}, vol.~97, no.~2, pp. 427--450,
  Feb. 2009.

\bibitem{6725647}
A.~Simonetto and G.~Leus, ``Distributed maximum likelihood sensor network
  localization,'' \emph{IEEE Trans. on Signal Processing}, vol.~62, no.~6, pp.
  1424--1437, Mar. 2014.

\bibitem{7383332}
W.~Yuan, N.~Wu, B.~Etzlinger, H.~Wang, and J.~Kuang, ``Cooperative joint
  localization and clock synchronization based on gaussian message passing in
  asynchronous wireless networks,'' \emph{IEEE Trans. Veh. Technol.}, vol.~65,
  no.~9, pp. 7258--7273, Sep. 2016.

\bibitem{6731596}
H.~Shen, Z.~Ding, S.~Dasgupta, and C.~Zhao, ``Multiple source localization in
  wireless sensor networks based on time of arrival measurement,'' \emph{IEEE
  Trans. on Signal Processing}, vol.~62, no.~8, pp. 1938--1949, Apr. 2014.

\bibitem{6955781}
S.~Bartoletti, W.~Dai, A.~Conti, and M.~Z. Win, ``A mathematical model for
  wideband ranging,'' \emph{IEEE J. of Sel. Topics in Signal Processing},
  vol.~9, no.~2, pp. 216--228, Mar. 2015.

\bibitem{4600186}
D.~Dardari, C.-C. Chong, and M.~Win, ``Threshold-based time-of-arrival
  estimators in uwb dense multipath channels,'' \emph{IEEE Trans. on Commun.},
  vol.~56, no.~8, pp. 1366--1378, Aug. 2008.

\bibitem{8264804}
S.~Mazuelas, A.~Conti, J.~C. Allen, and M.~Z. Win, ``Soft range information for
  network localization,'' \emph{IEEE Trans. on Signal Processing}, vol.~66,
  no.~12, pp. 3155--3168, June 2018.

\bibitem{8827486}
A.~Conti, S.~Mazuelas, S.~Bartoletti, W.~C. Lindsey, and M.~Z. Win, ``Soft
  information for localization-of-things,'' \emph{Proceedings of the IEEE},
  vol. 107, no.~11, pp. 2240--2264, Nov. 2019.

\bibitem{5762798}
M.~Z. Win, A.~Conti, S.~Mazuelas, Y.~Shen, W.~M. Gifford, D.~Dardari, and
  M.~Chiani, ``Network localization and navigation via cooperation,''
  \emph{IEEE Commun. Magazine}, vol.~49, no.~5, pp. 56--62, May 2011.

\bibitem{5571900}
Y.~Shen and M.~Z. Win, ``Fundamental limits of wideband localization — part
  {I}: A general framework,'' \emph{IEEE Trans. on Info. Theory}, vol.~56,
  no.~10, pp. 4956--4980, Oct. 2010.

\bibitem{5571889}
Y.~Shen, H.~Wymeersch, and M.~Z. Win, ``Fundamental limits of wideband
  localization — part {II}: Cooperative networks,'' \emph{IEEE Trans. on
  Info. Theory}, vol.~56, no.~10, pp. 4981--5000, Oct. 2010.

\bibitem{9606768}
Y.~Xiong, N.~Wu, Y.~Shen, and M.~Z. Win, ``Cooperative localization in massive
  networks,'' \emph{IEEE Trans. on Info. Theory}, vol.~68, no.~2, pp.
  1237--1258, Feb. 2022.

\bibitem{7364259}
Y.~Han, Y.~Shen, X.-P. Zhang, M.~Z. Win, and H.~Meng, ``Performance limits and
  geometric properties of array localization,'' \emph{IEEE Trans. on Info.
  Theory}, vol.~62, no.~2, pp. 1054--1075, Feb. 2016.

\bibitem{garcia2017direct}
N.~Garcia, H.~Wymeersch, E.~G. Larsson, A.~M. Haimovich, and M.~Coulon,
  ``Direct localization for massive {MIMO},'' \emph{IEEE Trans. on Signal
  Processing}, vol.~65, no.~10, pp. 2475--2487, May 2017.

\bibitem{8240645}
A.~Shahmansoori, G.~E. Garcia, G.~Destino, G.~Seco-Granados, and H.~Wymeersch,
  ``Position and orientation estimation through millimeter-wave {MIMO in {5G}}
  systems,'' \emph{IEEE Trans. on Wireless Commun.}, vol.~17, no.~3, pp.
  1822--1835, Mar. 2018.

\bibitem{8356190}
Z.~Abu-Shaban, X.~Zhou, T.~Abhayapala, G.~Seco-Granados, and H.~Wymeersch,
  ``Error bounds for uplink and downlink {3D} localization in {5G} millimeter
  wave systems,'' \emph{IEEE Trans. on Wireless Commun.}, vol.~17, no.~8, pp.
  4939--4954, Aug. 2018.

\bibitem{8515231}
R.~Mendrzik, H.~Wymeersch, G.~Bauch, and Z.~Abu-Shaban, ``Harnessing {NLOS}
  components for position and orientation estimation in {5G} millimeter wave
  {MIMO},'' \emph{IEEE Trans. on Wireless Commun.}, vol.~18, no.~1, pp.
  93--107, Jan. 2019.

\bibitem{8755880}
A.~Fascista, A.~Coluccia, H.~Wymeersch, and G.~Seco-Granados, ``Millimeter-wave
  downlink positioning with a single-antenna receiver,'' \emph{IEEE Trans. on
  Wireless Commun.}, vol.~18, no.~9, pp. 4479--4490, Sep. 2019.

\bibitem{fascista2021downlink}
------, ``Downlink single-snapshot localization and mapping with a
  single-antenna receiver,'' \emph{IEEE Trans. on Wireless Commun.}, vol.~20,
  no.~7, pp. 4672--4684, July 2021.

\bibitem{li2019massive}
X.~Li, E.~Leitinger, M.~Oskarsson, K.~{\AA}str{\"o}m, and F.~Tufvesson,
  ``Massive {MIMO}-based localization and mapping exploiting phase information
  of multipath components,'' \emph{IEEE Trans. on Wireless Commun.}, vol.~18,
  no.~9, pp. 4254--4267, Sep. 2019.

\bibitem{guerra2018single}
A.~Guerra, F.~Guidi, and D.~Dardari, ``Single-anchor localization and
  orientation performance limits using massive arrays: {MIMO vs.
  beamforming},'' \emph{IEEE Trans. on Wireless Commun.}, vol.~17, no.~8, pp.
  5241--5255, Aug. 2018.

\bibitem{9082200}
F.~Ghaseminajm, Z.~Abu-Shaban, S.~S. Ikki, H.~Wymeersch, and C.~R. Benson,
  ``Localization error bounds for {5G mmWave} systems under {I/Q} imbalance,''
  \emph{IEEE Trans. on Veh. Technol.}, vol.~69, no.~7, pp. 7971--7975, July
  2020.

\bibitem{boulogeorgos2022outage}
A.-A.~A. Boulogeorgos, A.~Alexiou, and M.~Di~Renzo, ``Outage performance
  analysis of {RIS}-assisted {UAV} wireless systems under disorientation and
  misalignment,'' \emph{arXiv:2201.12056}, 2022.

\bibitem{9201330}
T.~Ma, Y.~Xiao, X.~Lei, W.~Xiong, and Y.~Ding, ``Indoor localization with
  reconfigurable intelligent surface,'' \emph{IEEE Commun. Letters}, vol.~25,
  no.~1, pp. 161--165, Sep. 2021.

\bibitem{emenonye2022fundamentals}
D.-R. Emenonye, H.~S. Dhillon, and R.~M. Buehrer, ``Fundamentals of {RIS}-aided
  localization in the far-field,'' \emph{submitted to {\em IEEE Trans. on
  Wireless Commun.}, available online: arxiv.org/abs/2206.01652}, 2022.

\bibitem{9963716}
Z.~Esmaeilbeig, K.~V. Mishra, A.~Eamaz, and M.~Soltanalian, ``Cramér–rao
  lower bound optimization for hidden moving target sensing via
  multi-{IRS}-aided radar,'' \emph{IEEE Signal Processing Letters}, vol.~29,
  pp. 2422--2426, 2022.

\bibitem{esmaeilbeig2022joint}
Z.~Esmaeilbeig, A.~Eamaz, K.~V. Mishra, and M.~Soltanalian, ``Joint waveform
  and passive beamformer design in multi-irs aided radar,''
  \emph{arXiv:2210.14458}, 2022.

\bibitem{9827797}
Z.~Esmaeilbeig, K.~V. Mishra, and M.~Soltanalian, ``{IRS}-aided radar: Enhanced
  target parameter estimation via intelligent reflecting surfaces,'' in
  \emph{2022 IEEE 12th Sensor Array and Multichannel Signal Processing Workshop
  (SAM)}, Jun. 2022, pp. 286--290.

\bibitem{9732186}
S.~Buzzi, E.~Grossi, M.~Lops, and L.~Venturino, ``Foundations of {MIMO} radar
  detection aided by reconfigurable intelligent surfaces,'' \emph{IEEE Trans.
  on Signal Processing}, vol.~70, pp. 1749--1763, Mar. 2022.

\bibitem{9264225}
X.~Wang, Z.~Fei, J.~Guo, Z.~Zheng, and B.~Li, ``Ris-assisted spectrum sharing
  between mimo radar and mu-miso communication systems,'' \emph{IEEE Wireless
  Commun. Letters}, vol.~10, no.~3, pp. 594--598, Nov. 2021.

\bibitem{9454375}
S.~Buzzi, E.~Grossi, M.~Lops, and L.~Venturino, ``Radar target detection aided
  by reconfigurable intelligent surfaces,'' \emph{IEEE Signal Processing
  Letters}, vol.~28, pp. 1315--1319, 2021.

\bibitem{9500724}
C.~Psomas, H.~A. Suraweera, and I.~Krikidis, ``On the association with
  intelligent reflecting surfaces in spatially random networks,'' in
  \emph{Proc., IEEE Intl. Conf. on Commun. (ICC)}, 2021.

\bibitem{lavalle2006planning}
S.~M. LaValle, \emph{{Planning Algorithms}}.\hskip 1em plus 0.5em minus
  0.4em\relax Cambridge University Press, 2006.

\bibitem{hampton2013introduction}
J.~R. Hampton, \emph{Introduction to {MIMO} communications}.\hskip 1em plus
  0.5em minus 0.4em\relax Cambridge university press, 2013.

\bibitem{ellingson2021path}
S.~W. Ellingson, ``Path loss in reconfigurable intelligent surface-enabled
  channels,'' in \emph{Proc. IEEE 32nd Annu. Int. Symp. Pers., Indoor, Mobile
  Radio Commun. (PIMRC)}, Oct. 2021.

\bibitem{stutzman2012antenna}
W.~L. Stutzman and G.~A. Thiele, \emph{{Antenna Theory and Design}}.\hskip 1em
  plus 0.5em minus 0.4em\relax John Wiley \& Sons, 2012.

\bibitem{horn2012matrix}
R.~A. Horn and C.~R. Johnson, \emph{{Matrix Analysis}}.\hskip 1em plus 0.5em
  minus 0.4em\relax Cambridge University Press, 2012.

\bibitem{kay1993fundamentals}
S.~M. Kay, \emph{{Fundamentals of Statistical Signal Processing: Estimation
  Theory}}.\hskip 1em plus 0.5em minus 0.4em\relax Prentice-Hall, Inc., 1993.

\bibitem{keykhosravi2021multi}
K.~Keykhosravi and H.~Wymeersch, ``Multi-{RIS} discrete-phase encoding for
  interpath-interference-free channel estimation,'' \emph{arXiv:2106.07065},
  2021.

\bibitem{9064586}
Z.~Abu-Shaban, H.~Wymeersch, T.~Abhayapala, and G.~Seco-Granados,
  ``Single-anchor two-way localization bounds for {5G} mmwave systems,''
  \emph{IEEE Trans. on Veh. Technology}, vol.~69, no.~6, pp. 6388--6400, Apr.
  2020.

\bibitem{1097833}
J.-Y. Lee and R.~Scholtz, ``Ranging in a dense multipath environment using an
  {UWB} radio link,'' \emph{IEEE Journal on Sel. Areas in Commun.}, vol.~20,
  no.~9, pp. 1677--1683, Dec. 2002.

\bibitem{9335528}
F.~Guidi and D.~Dardari, ``Radio positioning with {EM} processing of the
  spherical wavefront,'' \emph{IEEE Trans. on Wireless Commun.}, vol.~20,
  no.~6, pp. 3571--3586, Jan. 2021.

\end{thebibliography}
}

\end{document}